\documentclass[aps,prx,twocolumn,superscriptaddress]{revtex4-2}
\pdfoutput=1 

\usepackage[T1]{fontenc} 
\usepackage{units}
\usepackage{amsmath}
\usepackage{amsthm}
\usepackage{amssymb}
\usepackage{graphicx,float,caption,tikz,subcaption,pgfplots}
\usepackage{color}
\usepackage{xcolor}
\usepackage{bbold}
\usepackage{titlesec} 
\usepackage{comment}
\usepackage{natbib}
\usepackage{enumitem}
\usepackage[capitalise]{cleveref}
\usepackage{ulem}
\definecolor{myurlcolor}{rgb}{0,0,0.7}
\definecolor{myrefcolor}{rgb}{0.1,0,0.9}
\usepackage{overpic}

\newcommand{\rank}{\operatorname{rank}}

\DeclareMathOperator{\tr}{tr}
\DeclareMathOperator{\XOR}{XOR}

\newtheorem{theorem}{Theorem}
\newtheorem{lemma}{Lemma}
\newtheorem{remark}{Remark}
\newtheorem{corollary}{Corollary}[theorem]

\usepackage[linesnumbered, ruled,vlined]{algorithm2e}

\graphicspath{{./images/},{./imagesAppendix/}}

\renewcommand{\eqref}[1]{Eq.~(\ref{#1})} 
\newcommand{\figref}[1]{Fig.~(\ref{#1})} 

\def\app#1#2{%
  \mathrel{%
    \setbox0=\hbox{$#1\sim$}%
    \setbox2=\hbox{%
      \rlap{\hbox{$#1\propto$}}%
      \lower1.1\ht0\box0%
    }%
    \raise0.25\ht2\box2%
  }%
}

\ifx\proof\undefined
\newenvironment{proof}[1][\protect\proofname]{\par
	\normalfont\topsep6\p@\@plus6\p@\relax
	\ivlist
	\itemindent\parindent
	\item[\hskip\labelsep\scshape #1]\ignorespaces
}{%
	\endtrivlist\@endpefalse
}
\providecommand{\proofname}{Proof}
\fi

\newtheorem{proposition}{Proposition}

\newcommand{\bra}[1]{\langle #1|}
\newcommand{\ket}[1]{|#1 \rangle}
\makeatother
\newcommand{\braket}[2]{\langle #1 \vert #2 \rangle}
\newcommand{\ketbra}[2]{|#1 \rangle\!\langle #2 |}
\newcommand{\abs}[1]{\left|#1\right|}

\newcommand{\defeq}{\mathrel{\mathop:}=}

\newcommand{\Tr}{\mathrm{Tr}}

\providecommand{\factname}{Fact}
\providecommand{\theoremname}{Theorem}
\providecommand{\claimname}{Claim}
\providecommand{\lemmaname}{Lemma}
\providecommand{\definitionname}{Definition}

\def\bbbone{{\mathchoice {\rm 1\mskip-4mu l} {\rm 1\mskip-4mu l}
{\rm 1\mskip-4.5mu l} {\rm 1\mskip-5mu l}}}

\definecolor{KB}{rgb}{0.4,0.3,0.9}

\definecolor{THc}{rgb}{0.9,0.3,0.2}

\newcommand{\be}{\begin{equation}}
\newcommand{\ee}{\end{equation}}
\newcommand{\ba}{\begin{eqnarray}}
\newcommand{\ea}{\end{eqnarray}}

\newcommand{\pur}{\operatorname{Pur}}

\newcommand{\st}[1]{\ketbra{#1}{#1}}
\newcommand{\stab}{\operatorname{STAB}}
\newcommand{\pstab}{\operatorname{PSTAB}}
\newcommand{\norm}[1]{\left\| #1 \right\| }
\newcommand{\sma}[1]{M^{(NL,\epsilon)}_{RS}(#1)}

\newtheorem{definition}{\protect\definitionname}
\newtheorem{conjecture}{Conjecture}

\newcommand{\CC}[1]{\textcolor{blue}{(CC) #1}}
\newcommand{\GC}[1]{\textcolor{orange}{(GC #1)}}

\begin{document} 

\title{Gravitational backreaction is Magical}

\author{ChunJun Cao}
\affiliation{Department of Physics, Virginia Tech, Blacksburg, Virginia 24061, USA}

\author{Gong Cheng}
\affiliation{Department of Physics, Virginia Tech, Blacksburg, Virginia 24061, USA}

\author{Alioscia Hamma}
\affiliation{Dipartimento di Fisica ‘Ettore Pancini’, Università degli Studi di Napoli Federico II, Via Cintia 80126, Napoli, Italy}
\affiliation{INFN, Sezione di Napoli, Napoli, Italy}
\affiliation{Scuola Superiore Meridionale, Largo S. Marcellino 10, 80138 Napoli, Italy}

\author{Lorenzo Leone}
\affiliation{Dahlem Center for Complex Quantum Systems, Freie Universit\"at Berlin, 14195 Berlin, Germany}

\author{William Munizzi}
\affiliation{Department of Physics, Arizona State University, Tempe, Arizona 85281, USA}

\author{Savatore F.E. Oliviero}
\affiliation{NEST, Scuola Normale Superiore and Istituto Nanoscienze, Consiglio Nazionale delle Ricerche, Piazza dei Cavalieri 7, IT-56126 Pisa, Italy}

\begin{abstract}
    We study the interplay between magic and entanglement in quantum many-body systems. We show that nonlocal magic, which is supported by the quantum correlations is lower bounded by the nonflatness of entanglement spectrum and upper bounded by the amount of entanglement in the system. We then argue that a smoothed version of nonlocal magic bounds the hardness of classical simulations for incompressible states. In conformal field theories, we conjecture that the nonlocal magic should scale linearly with entanglement entropy but sublinearly when an approximation of the state is allowed. We support the conjectures using both analytical arguments based on unitary distillation and numerical data from an Ising CFT. If the CFT has a holographic dual, then we prove that the nonlocal magic vanishes if and only if there is no gravitational backreaction. Furthermore, we show that nonlocal magic is approximately equal to the rate of change of the minimal surface area in response to the change of cosmic brane tension in the bulk. 
\end{abstract}

\maketitle
\flushbottom

\section{Introduction}
\label{sec:intro}
Entanglement is an important quantum resource and an integral part of our understanding of quantum many-body physics and quantum gravity, such as topological order \cite{kitaev_topological_2006,levin_detecting_2006,hamma_bipartite_2005}, nonequilibrium dynamics \cite{hosur_chaos_2016,vonkeyserlingk_operator_2018,nahum_operator_2018,skinner_measurementinduced_2019}
, spacetime \cite{VanRaamsdonk:2010pw}, and black holes \cite{Almheiri_2013,Maldacena_2013}.  In the anti-de~Sitter/conformal field theory (AdS/CFT) correspondence \cite{Maldacena_1999,Witten:1998qj}, entanglement in the CFT is important for emerging spacetime geometry \cite{Czech_Lampros,Czech_2014,Czech_2017,Radon,Bao:2019bib} in the dual gravity theory, e.g. via the Ryu-Takayanagi formula \cite{Ryu_2006,Lewkowycz_2013,Faulkner_2013,Hubeny_2007}. Surprisingly, this connection between geometry and entanglement holds not only for holographic CFTs, but also for more general quantum many-body systems like tensor network toy models, which have been enormously successful in reproducing an analogous Ryu-Takayanagi formula \cite{HarlowRT}, the emergent bulk geometry, and subregion operator reconstruction through quantum error correction \cite{Pastawski_2015,Hayden_2016,Yang_2016,Harris_2018,ABSC,Steinberg_2023}. This is a profound development as it suggests the lessons from holography may also apply beyond the confines of AdS \cite{Jacobson_2016,Cao_2017,Cao_2018}.

However, the entanglement patterns in the tensor network models alone do not capture the full quantum landscape spanned by holography. Despite many recent advances \cite{dong2023holographic,akers2024background,cheng2022random,ABSC,HMERA,Bao:2019}, it is still unclear how gravity can emerge in such models. In particular, neither the holographic stabilizer codes \cite{Pastawski_2015} nor the random tensor networks \cite{Hayden_2016} can fully capture the CFT entanglement spectrum and gravitational backreaction. stabilizer tensor networks also fail to capture power-law correlations, robust multipartite entanglement, and nontrivial area operators \cite{Akers:2019gcv,Hayden:2021gno,nogo}. From a resource-theoretic perspective, what are these tensor network models missing compared to the low energy states in holographic theories? We show in this work that the answer is magic \cite{bravyi_universal_2005,veitch_resource_2014,stabrenyi,bu_stabilizer_2023}, or more precisely, \emph{nonlocal magic}.

Quantumness comes in two layers: entanglement gives the power of building correlations stronger than classical and violates Bell's inequalities while quantum advantage characterizes the hardness of simulating quantum systems on a classical computer. The latter is distinct from entanglement  --- a task involving a highly entangled system is not always hard to simulate classically as it can be achieved purely using Clifford operations that are classically simulable. 
This notion of classical hardness that constitutes the second layer of quantumness is intimately connected to the amount of nonstabilizerness, also known as magic, in the system. Although magic alone cannot generate the intricate patterns of complexity that are crucial for the  complex behavior in a quantum wave function, when used in conjunction with Clifford operations, nonstabilizerness \cite{bravyi_universal_2005} is both necessary and sufficient in realizing (fault-tolerant) universal quantum computation. Therefore, it is the remaining piece needed for quantum advantage and for simulating holographic conformal field theories.

In addition to being an important resource for fault-tolerant quantum computation \cite{bravyi_universal_2005,veitch_resource_2014} and quantum simulation, pioneering work has established magic as an important ingredient for characterizing quantum many-body systems \cite{white_conformal_2021,sarkar_characterization_2020,Liu_2022,tarabunga2023manybody}, such as dynamics \cite{stabrenyi,chaosbymagic,sewell_mana_2022,rattacaso_stabilizer_2023}, quantum phases \cite{leone2023phase,Niroula:2023meg}, quantum circuits~\cite{leone_quantum_2021,oliviero_transitions_2021,bejan2023dynamical},and randomness\cite{vairogs2024extracting}. In the context of holography, \cite{white_conformal_2021,magicising,tarabunga2023manybody} showed that magic is abundant in CFTs and is therefore expected to play an important role for reproducing the correct CFT entanglement spectrum, for generating power-law correlations, for building nontrivial area operators in holographic codes, and for reproducing the correct multipartite entanglement in holographic geometries \cite{Hayden:2021gno}\footnote{Although this is not noted by the authors explicitly, it is clear that holographic states require $O(1/{G}_N)$ tripartite entanglement but cannot be predominantly GHZ-type \cite{Nezami:2016zni}. }.  

There are also many questions surrounding the role played by magic. Empirically, the amount of nonstabilizerness or non-Gaussianity \cite{veitch_resource_2014,campbell_catalysis_2011,leone_nonstabilizerness_2023,Hebenstreit_2019,saxena_quantifying_2022,bu_stabilizer_2023,bu2023discrete,weedbrook_gaussian_2012} present in a quantum process appears to correlate with the hardness of classical simulations \cite{zhang2024unconditional}, e.g. in stabilizer and matchgate simulations \cite{aaronson_improved_2004,gottesman1997stabilizer,Jozsa_2008,hebenstreit_computational_2020,bravyi_improved_2016,bravyi_trading_2016,bravyi_simulation_2019} as well as in Monte Carlo sampling \cite{magicMC}.
However, its precise connection with complexity is yet unclear. While it is proposed 
 \cite{white_conformal_2021,nogo} that the replication of the CFT entanglement spectrum and emergent gravity in AdS/CFT requires magic, the specific mechanism through which magic accomplishes this also remains uncertain.
Furthermore, although the amount of magic present in a system can be illuminating all by itself, it is becoming clear the distribution of magic is equally, if not more, important for understanding  nonequilibrium dynamics and entanglement spectrum \cite{flatness}. For example, the amount of magic is generally expected to scale volumetrically with the number of qubits in quantum many-body systems. The tensor product of nonstabilizer states, CFT ground states, and Haar random states all have a high magic density and volume law magic scaling, and yet their physical properties and their usefulness for quantum computation are completely different. Therefore, a more profound understanding of the interplay between  entanglement and magic will shed new light on the structure of quantum matter, quantum information, and gravity. {\color{black} More specifically, a robust connection between magic and holography will provide important clues for emerging gravity in tensor network toy models and help constrain magic-state resources in the quantum simulation of conformal field theories. Since magic is generally hard to compute in quantum many-body systems, such a connection will also enable the computation or estimation of magic using geometric means by doing bulk gravity calculations.}

In this work, we report multiple advances in respond to the above queries. We define nonlocal magic and offer compelling evidence for how it is connected to the hardness in classically simulating incompressible states. We provide rigorous bounds as well as computable estimates for nonlocal magic in any quantum system and show that it is lower bounded by the antiflatness of the entanglement spectrum and upper bounded by various functions of the R\'enyi entropies. When applied to CFTs, we propose a straightforward relationship between magic, entropy, and antiflatness. For theories with holographic dual, we show that the nonlocal magic controls the amount of gravitational backreaction in response to stress energy, and thus critical for the emergence of gravity. 


\section{Main results}
In this section, we explain the main results of this paper and lay down informally the setup and strategy of this work. Then, in the following sections, we derive them rigorously. A key goal of this paper is to show that  the nonlocal magic is responsible for the nonflat entanglement spectrum in a CFT and for the backreaction in AdS through the AdS-CFT dictionary. \textcolor{black}{We also identify inequalities between nonlocal magic, spectral nonflatness, and entanglement for general quantum systems.}
\begin{figure}
    \centering
    \includegraphics[width=\linewidth]{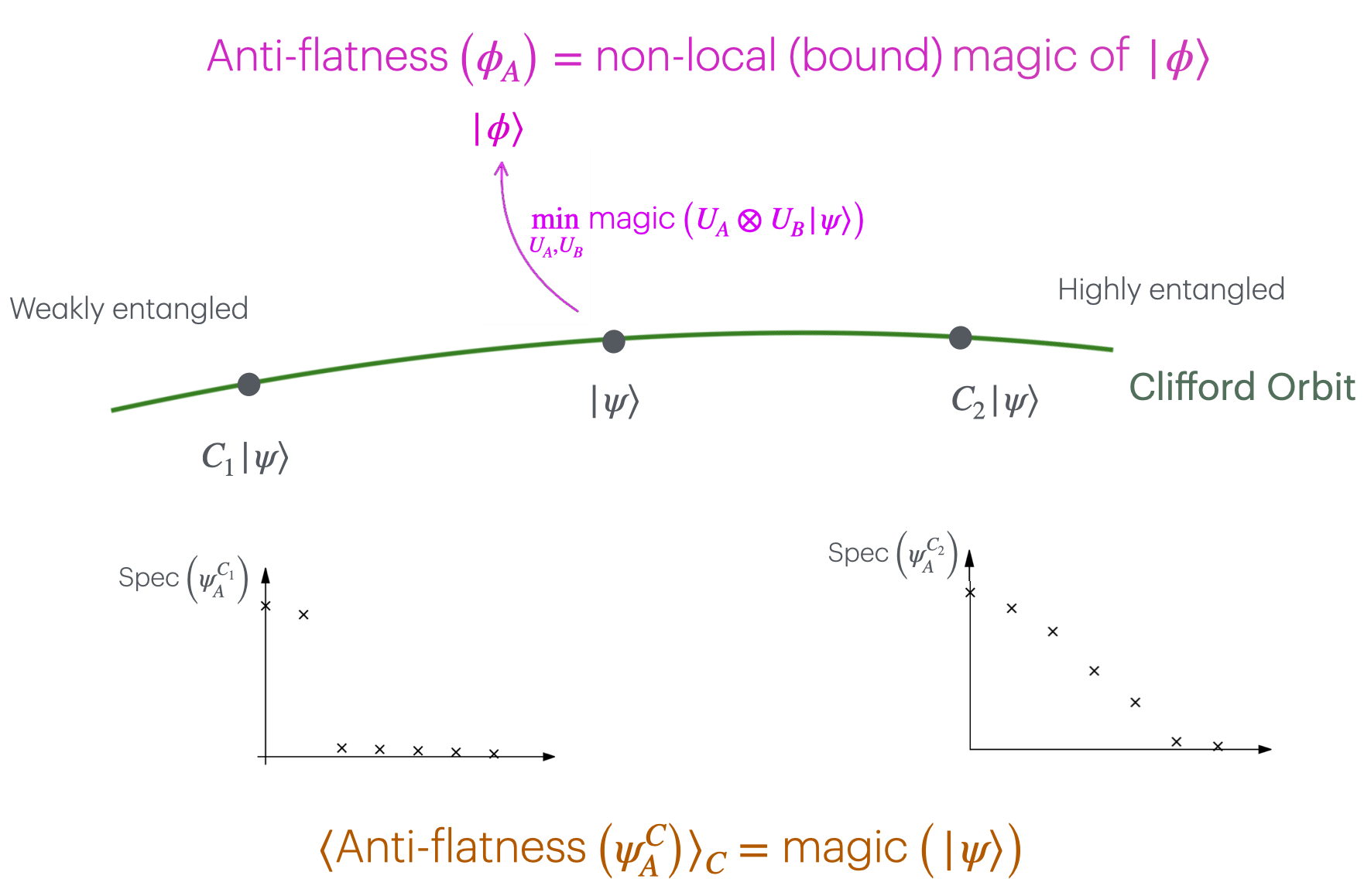}
    \caption{States on the same Clifford orbit have the same magic but most of them are very entangled and the average flatness of their entanglement spectrum is the magic of the full state $|\psi\rangle$. Moreover, the nonlocal (bound) magic in a state is equal to the antiflatness of the magic-reduced state $|\phi\rangle$. This is a manifestation of the entanglement-magic duality. }
    \label{fig:mf2}
\end{figure}

Since the seminal work of Ryu and Takayanagi~\cite{Ryu_2006}, a number of entries have been added to the AdS-CFT dictionary where one can connect quantum information-theoretic quantities on the boundary to geometric quantities in the bulk. Notably, the correspondence can be used to find the holographic dual to functions of the spectrum of a reduced density operator $\psi_A$ in the conformal field theory \cite{Dong:2018lsk}, where $A$ is a subsystem of the CFT. 
The strategy of this work is to find a holographic dual of magic in a state $\psi$  by connecting it to the spectrum of its reduced density operator $\psi_A$.

At first sight, this may seem like an impossible task. There are several reasons, due to the fact that  the way magic relates to  spectral properties is complicated.

First of all,
 magic is generally a property of the full state $\psi$, so how can the spectrum of $\psi_A$ give us information on the magic of the full parent state $\psi$? 
 The connection between spectral properties of the reduced state  $\psi_A$ and the magic of the full state $\psi$ comes from the very remarkable fact that the magic of a state $\psi$ is related to the 
average deviation from the flat spectrum of the spectrum of the reduced density operator $\psi_A$ through the Clifford orbit \cite{flatness,Keeler:2022ajf}. The Clifford orbit preserves the magic, but entangles the system \cite{Keeler:2023xcx,Munizzi:2023ihc,hfhg-2z68}, therefore populating the spectrum of the reduced density operator. In fact, there is no need to take this average, as long as the spectrum of the subsystem density operator possesses an entropy obeying volume law. In this case, its antiflatness is enough to probe the magic of the full state, see Fig.\ref{fig:mf2} for a pictorial representation.

A second difficulty comes from the fact that the above result connecting local antiflatness and global magic comes from a resource theory of magic that considers magic-free only the states that can be purified in stabilizer states, we call such a theory $\stab_0$. This is the resource theory of magic established by the null set of stabilizer entropy \cite{stab0-inprep}.  A consequence of this theory is that there are magicful states from which cannot be distilled pure magic-states by means of Clifford (free) operations. 
We will first therefore first develop this theory by employing as monotones both the trace distance $M_{dist}$ and relative entropy of resource $M_R$. They will both be useful later to establish our results.

The third difficulty is related to CFT as the states in this theory are not hosting volume law for entanglement. In order to exploit the flatness-magic correspondence for a theory that generally has an area-law scaling of entanglement, we must focus on the boundary $\partial A$ between the subregions $A$ and its complement $B$ where most of the entanglement is being mediated. On the Hilbert subspace supported on $\partial A$, the density operator $\psi_{\partial A}$ is well populated and as a consequence, we can compute its magic through the spectrum. This gives rise to the notion of nonlocal magic and its connection to antiflatness is a manifestation of the entanglement-magic duality \cite{iannotti2025entanglementstabilizerentropiesrandom}. This is the magic that cannot be extracted by local unitaries and is bound to the system like bound entanglement \cite{PhysRevLett.80.5239}. 

The main results of this work are grouped in two parts: 1) quantum information-theoretic results that rigorously define nonlocal magic 
 for both the magic measures defined above, namely the trace distance of nonlocal magic $M_{dist}^{(NL)}$ and the relative entropy of nonlocal magic $M_{R}^{(NL)}$ 
 and relate them to spectral quantities. In particular, it {\color{black} is necessary and sufficient for} \textit{antiflatness}   $\mathcal F$ \cite{flatness}, that is, a measure of how much the spectrum of a density operator is far from a flat distribution; and 2) the application of these tools to AdS/CFT by first making precise the relation between entanglement, nonlocal magic and spectral flatness in a CFT. Then for holographic CFTs, we show that {\color{black}  nonlocal magic is necessary for gravitational backreaction.} {\color{black} Leveraging this connection, we provide both a quantitative estimate of the nonlocal magic resource needed to simulate conformal field theories and show that holographic quantities can be leveraged to determine quantum many-body nonlocal magic.}
 
\subsection{Quantum information-theoretic results}

The first result is that, given the bipartition $AB$, for a subsystem $A$ of a quantum state $\psi_{AB}$,  $M_{dist}^{(NL)}$ is lower bounded by the antiflatness $\mathcal F(\psi_A)$ and upper bounded by the entanglement:
\begin{equation}
\begin{split}
		&\mathcal{F}(\psi_{ A})/8 \le M_{dist}^{(NL)}(\psi_{AB})\\ &\le  \sqrt{1-e^{-S_{max}(A)}+e^{S_\infty(A)}\left(1-\frac{e^{\log d \lfloor S_{max}(A)/\log d \rfloor}}{e^{S_{max}(A)}}\right)}
\end{split}
\end{equation}
where $\lfloor\cdot\rfloor$ is the floor function, $S_{max}(A):=\log \rank{\psi_A}$, $S_{\infty}(A)=\lambda_{\text{max}}(\psi_A)$. For this, we assume the total Hilbert space is a tensor product of qubits (or qudits) with uniform local dimension $d$. 

Second, by using the nonlocal magic measured by relative entropy $M^{(NL)}_{RS}$, one can find another relationship between magic in a  quantum state $\psi$ and its entanglement: 
\begin{equation}\label{eq:relstab0}
        S_{max}(A)-S(A) \leq  M^{(NL)}_{RS}(\psi_{AB})\leq \log d \lceil   S_{max}(A)/\log d  \rceil,
\end{equation}
where $\lceil\cdot\rceil$ is the ceiling function. The lower and upper bounds in the above equation are essentially tight for weakly entangled states. 
\cref{eq:relstab0} has also the advantage of allowing one to find good estimates for $M^{(NL)}_{RS}(\psi_{AB})$ in terms of the Schmidt coefficients of $\psi_{AB}$ (see \cref{estimate_prop}). This is important because nonlocal magic is otherwise very difficult to calculate.  Moreover, the relative entropy of magic allows us to define the smoothed (nonlocal) magic {\color{black} for  systems in the continuum, such as a quantum field theory, }as 
\begin{equation}\label{smoothednlmagic}
\sma{\psi_{AB}}\defeq \min_{\Vert\chi-\psi_{AB\Vert}<\epsilon}M_{RS}^{(NL)}(\chi).
\end{equation}
For a pure state therefore, we obtain the bounds 
\begin{equation}
\begin{split}
     S_{max}^{\epsilon}(A)-(1-\epsilon)^{-1}S(A)&\leq \sma{\psi_{AB}}\\
     &\leq \log d \lceil   S^{\epsilon}_{max(A)}/\log d  \rceil,
     \label{eqn:smoothinequality0}
\end{split}
\end{equation}
with the smoothed maximal entropy is defined as $S_{max}^{\epsilon}(A):=\min_{\Vert\chi-\psi_A\Vert<\epsilon} \ln(\rank(\chi))$. Since the lower bound quantifies the compressibility of a state, we show that incompressible states with low entanglement, but high nonlocal magic, can still be difficult to classically simulate.

Finally, for a system of qubits, we use the spectral information $\{\lambda_i\}$ of state $\psi$, along with the magic measure known as stabilizer 2-R\'enyi entropy $\mathcal{M}_2$, to estimate the nonlocal magic $\mathcal{M}_2(\{\lambda_i\})$. This calculation yields a tighter upper bound on nonlocal magic than \cref{eq:relstab0}, stating
\begin{equation}
\mathcal{M}_2^{NL}(\psi_{AB})\leq \min\{2S_2(A),4(S_{max}(A)-S_{1/2}(A))\},
\end{equation}
where $S_{n}(A)$ are the R\'enyi-$n$ entropies of $\psi_A$ and $S_{max}$ is the logarithm of Schmidt rank, which is taken to be an integer power of two. 

For each of the above measures, we show that nonlocal magic vanishes if and only if the entanglement spectrum is flat, see Lemma \ref{lemmaNL}.

\subsection{AdS/CFT results}

We now state the main holographic result of this work. One can use nonlocal magic to derive an RT-like formula for gravitational backreaction, defined as the susceptibility of a backreacted surface area $\mathcal{A}$ with respect to the insertion of a cosmic brane with tension $\mathcal{T}$. The first step is connecting backreaction to spectral quantities.
{\color{black} The first spectral quantity of interest is the capacity of entanglement $C_E$, which is also the variance of the entanglement spectrum with respect to the probability defined by the density operator itself, see Eq.\ref{sigmalogrho}. The main result is
\be
\frac{\partial \mathcal{A}}{\partial \mathcal{T}}=-C_E(\psi)
\ee
we show that $C_E(\psi)$ is a measure of antiflatness. Then, thanks to \cref{def:nonlocfromcapacity} and \cref{th:capacity}, we show that we can connect antiflatness to nonlocal magic. Namely, we know that nonlocal magic $|M_{C_E}^{NL}-C_{E}|\le 1$ and therefore
\be
\frac{\partial \mathcal{A}}{\partial \mathcal{T}}=-M_{C_E}^{NL}(\psi)\pm 1
\ee
the above implies that, in case of extensive nonlocal magic,
\ba
\frac{\partial \mathcal{A}}{\partial \mathcal{T}}\simeq-M_{C_E}^{NL}(\psi).
\ea
Hence the gravitational backreaction {\it is} the nonlocal magic in all the regions responding to the stress energy of the cosmic brane.
}

\textcolor{black}{As we discuss multiple possible measures of magic, qualitatively similar relations can also be obtained for these other measures. First, we connect the backreaction to spectral quantities. By restricting to small subregion $A$ or in the near-flat limit, we have the approximation,}
%
\ba\label{eq:nFoverpurity}
\left.\frac{\partial \mathcal{A}}{\partial \mathcal{T}}\right\vert_{\mathcal{T}=0} \approx -\Big(\frac{4G}{\mathrm{Pur}(\psi_A)}\Big)^2\mathcal{F}(\psi_A),
\ea
\textcolor{black}{which holds when the higher order moments of spectrum (beyond variance) are negligible.}  
 
Together with the relation between antiflatness and non local magic, Theorem \ref{th:magicdist},
we find  
\begin{equation}
\begin{split}
M_{dist}^{(NL)}(\psi_{AB})&\ge \frac{1}{8}\Big(\frac{\mathrm{Pur}(\psi_A)}{4G}\Big)^2\left\vert\frac{\partial \mathcal{A}}{\partial \mathcal{T}}\right\vert_{\mathcal{T}=0} \\
&\ge \frac{1}{8}\Big(\frac{e^{-\mathcal A/4G}}{4G}\Big)^2\left\vert\frac{\partial \mathcal{A}}{\partial \mathcal{T}}\right\vert_{\mathcal{T}=0} \propto \frac{1}{8}\left|\frac{\partial e^{-2\mathcal A/4G}}{\partial \mathcal{T}} \right|_{\mathcal{T}=0} 
\end{split}
\end{equation}
the left-hand side is the magic in the CFT side, and the right-hand side (RHS) of the above equation is a measure of the backreaction in AdS. As we prove in  \cref{section:brane}, the above equation also implies that backreaction is nonzero only if nonlocal magic is non vanishing.  

Further exploiting the structure of entanglement in CFT (see \cref{eqn:NLMflatness}), we can also obtain a simpler relation that holds more generally \textcolor{black}{without constraint on the spectrum:} 
%
\ba\label{branemagicNL0}
\left\vert\frac{\partial \mathcal{A}}{\partial \mathcal{T}}\right\vert_{\mathcal{T}=0} 
\approx \frac{(4G)^2}{\kappa} \mathcal{M}_2^{NL}(\psi_{AB})
\ea
which shows a more direct relation between gravitational backreaction and nonlocal magic based on the stabilizer 2-R\'enyi entropy for some constant $\kappa$.

Now for more general CFTs that need not have holographic duals, the above relations continue to hold with suitable substitutions of $\mathcal{T}\rightarrow (n-1)/4Gn$ and $\mathcal{A}/4G\rightarrow \tilde{S}_n$ where $\tilde{S}_n$ is a function of R\'enyi entropy defined by Ref.~\cite{Dong:2018lsk}. 
{\color{black} Provided that an approximate unitary distillation argument holds, we provide compelling analytical arguments that the amount of nonlocal magic is proportional to the an additive antiflatness measure in the system. We further show that the exact nonlocal magic in the CFT scales as $S(A)$ whereas the smoothed nonlocal magic scales as $\sqrt{S(A)}$. We also backup these claims with numerical evidence in an Ising CFT. Since the UV regulated non-stabilizer resource cost in an infinite dimensional system like a quantum field theory can remain finite, the approximation also permits a quadratic reduction compared to the na\"ive expectation in a quantum simulations.}
We then conjecture that such relations hold for general CFTs and apply this conjecture to evaluate magic for selected examples in holographic CFT using \cref{eqn:smoothinequality0}. Specifically, we do so for the static thermofield double state, and for nonequilibrium dynamics after local and global quantum quenches. We also examine the magic evolution in a time-evolved wormhole geometry described by a thermal field double state.
\section{NONLOCAL Magic}

\subsection{Magic measures}\label{sec:magicmeasure}
In this section, we introduce several measures of magic that will be central to supporting the claims in this paper. In order to properly establish a magic-state-resource theory, it is essential that we define an initial null set for such a resource theory. To achieve this purpose, we introduce three null sets, which we label as $\pstab$,  $\stab_0$, and $\stab$. Then we derive the free operations on such sets. 

Additionally, we must introduce several useful concepts: the Pauli group, the Clifford group, and the set of stabilizer quantum states. Consider the Hilbert space of single-qudit $\mathcal{H} = \mathbb{C}^d $, on which we define the following Pauli operators
\begin{equation}
	X\ket{i}=\ket{i+1} \quad Z\ket{j}=\omega^j\ket{j}, 
	\label{eq:XZ}
\end{equation}
where $\omega\equiv \exp(2 i \pi/d)$. The selection of operators in \eqref{eq:XZ} likewise defines the qudit computational basis $\{\ket{i}\}_i^d$.

The Pauli group $\tilde{\mathcal{P}}$ is defined as follows
\begin{equation}
	\tilde{\mathcal{P}}\equiv \langle \tilde{\omega}\bbbone, X,Z \rangle 
	\label{eq:Pauligroup}
\end{equation}
where $\langle \cdot \rangle$ labels the set generated by $\{\tilde{\omega}\bbbone, X,Z\}$, and $\tilde{\omega}=\omega$ for $d$ odd, and $\tilde{\omega}=\exp[i\pi /d] $ for $d$ even. When the number of qudits is $n$, the Pauli group $\tilde{\mathcal{P}}_n$ is defined as the $n$-fold tensor product of the single qudit Pauli group $\tilde{\mathcal{P}}$. 

The Clifford group $\mathcal{C}(d^n)$ is defined as the normalizer of the Pauli group, meaning that for any $U\in\mathcal{C}(d^n)$ we have $U^{\dagger}\tilde{\mathcal{P}}_n U\equiv\tilde{\mathcal{P}}_n$. The group $\mathcal{C}(d^n)$ is a multiplicative matrix group. For qubits $d=2$, it  can be generated by the Hadamard, phase, and controlled-Z quantum gates
\begin{equation}\label{CliffordGates}
   \operatorname{H}\equiv \frac{1}{\sqrt{2}}\begin{bmatrix}1&1\\1&-1\end{bmatrix}, \quad \operatorname{P}\equiv \begin{bmatrix}1&0\\0&i\end{bmatrix}, \quad     \operatorname{CZ} \equiv \begin{bmatrix}
            1 & 0 & 0 & 0\\
            0 & 1 & 0 & 0\\
	    0 & 0 & 1 & 0\\
	    0 & 0 & 0 & -1
            \end{bmatrix}.
\end{equation}
For general $d$ the generators are~\cite{jafarzadeh_randomized_2020} the controlled-$Z$ $\operatorname{CZ}$, the quantum Fourier transform $\operatorname{F}$ and the phase gate $\operatorname{P}$, whose action of the $d$-computational basis is
\begin{equation}\label{cliffordqutitgates}
\begin{split}
&\operatorname{CZ}\ket{ii^{\prime}}:=\omega^{ii^{\prime}}\ket{ii^{\prime}} \quad \operatorname{F}\ket{i}:=\frac{1}{\sqrt{d}}\sum_{i\in\mathbb{Z}_d}\omega^{ii^{\prime}}\ket{i^{\prime}} \\
&\operatorname{P}\ket{i}:=\omega^{s(s+\phi_d)/2}\ket{s}
\end{split}
\end{equation}
where $\phi_d=1$ if $d$ is odd, $0$ otherwise. 
Notably, circuits composed of the Clifford gates in \eqref{CliffordGates} can be efficiently simulated on a classical computer \cite{gottesman_heisenberg_1998,aaronson_improved_2004}.

At this point, one can define the notion of stabilizer states for pure states. We first say that a pure state $\ket{\phi}$ is stabilized by $P\in \tilde{\mathcal{P}}_n$ if
$P\ket{\phi}=\ket{\phi}$. Then we define the pure stabilizer states as the set
\begin{equation}
\pstab^{(n)}:=\{ \ket{\phi}\bra{\phi}=\frac{1}{|G|}\sum_{P\in G} P | G\subset \tilde{\mathcal{P}}_n, \; G \,\mbox{abelian} \}
\end{equation}
with the cardinality of $G$ being $|G|=d^n$ and $G$ is a group of commuting Pauli operators. Notice that $\pstab^{(n)}$ is the orbit through the Clifford group of any computational basis state for $n$ qudits, i.e., $\pstab^{(n)}= \{ C\ket{i_1\ldots i_n}| C\in\mathcal{C}(d^n)\}$. The notion of pure stabilizer states conveys the fact of a set of resources that is closed under Clifford operations.

For mixed-states, the most primitive notion of stabilizer states is that of $\stab_0$ in Ref.~\cite{stabrenyi}, defined as the set of states $\sigma=1/{d^n}\sum_{P\in G} P$, where $G$ is a group of commuting Pauli operators (see Ref.~\cite{nielsen_quantum_2000}). In Ref.~\cite{stabrenyi}, $\stab_0^{(n)}$ is introduced as the set of states for which the stabilizer entropy (SE) is zero and  SE is a good monotone for  $\pstab^{(n)}$, see Ref.~\cite{Leone_2024_monotonicity}.
From a more foundational perspective, $\stab_0^{(n)}$ is the set of states that can be purified in $\pstab^{(n)}$ and they can only yield trivial probability distributions, see \cite{stab0-inprep}

When one allows for general probabilities distributions we obtain the convex hull of $\pstab^{(n)}$, namely
 $\stab^{(n)}:=\{\sigma|\sigma=\sum_i p_i \ket{\phi_i}\bra{\phi_i},~ |\phi_i\rangle\in \mathrm{PSTAB}^{(n)}\}$.  Note that $\stab_0^{(n)}\subset\stab^{(n)}$. 
 
The next step in the definition of our measures of magic is to define the free operations of $\stab^{(n)}$ and $\stab_0^{(n)}$. For $\stab^{(n)}$ the free operations are given in Ref.~\cite{veitch_resource_2014}, and we list them here for the sake of completeness: 
\begin{enumerate}
	\item Clifford unitaries. $\rho\rightarrow U\rho U^{\dagger}$ with $U\in \mathcal{C}(d^n)$.
	\item Composition with stabilizer states, $\rho \rightarrow \rho \otimes \sigma $ with $\sigma $ a stabilizer state.
	\item Computational basis measurement on the first qudit, $\rho \rightarrow (\st{i}\otimes \bbbone_{n-1}) \rho (\st{i}\otimes \bbbone_{n-1} )/\Tr(\rho \st{i}\otimes \bbbone_{n-1})$ with probability $\Tr(\rho \st{i}\otimes \bbbone_{n-1})$
	\item Partial trace of the first qudit, $\rho \rightarrow \Tr_{1}(\rho)$
	\item The above operations conditioned on the outcomes of measurements or classical randomness. 
\end{enumerate}
It is straightforward to show that operations $(1)-(4)$ also apply to $\stab_0^{(n)}$ (see~\cref{app:stab0invariance}). However, it is important to note that stabilizer operations conditioned on measurements or classical randomness do not belong to the set of free operations for $\stab_0^{(n)}$. This is an important feature of the $\stab_0^{(n)}$ resource theory as it counts nonflat probabilities as resources. It is the key element to use deviation from flatness as the resource that connects magic in CFT to geometry in AdS.

Given the notion of null sets and free operations, one can then proceed to introduce suitable measures of magic. Let us start by defining the trace distance of magic:
\begin{definition}[Trace distance of magic$_0$]\label{def:mod_tr_dist_mag}
	The trace distance of magic$_0$ of a state $\psi$ is given by:
	\begin{equation}
	M_{\text{dist}}(\psi):=\min_{\sigma\in\stab_0^{(n)}}\frac{1}{2}\left\|\psi-\sigma  \right\|_{1}
  \end{equation} 
\end{definition}
\begin{proposition}
	The trace distance of magic satisfies the following properties: 
 \begin{enumerate}
\item Faithfulness: $M_{\text dist}(\rho)=0$ if and only if $\rho$ is a stabilizer state. 
\item Monotonicity: for all completely positive trace-preserving channels $\xi$ preserving $\stab_0^{(n)}$,  $M_{\text dist}(\xi(\rho))\le M_{\text dist}(\rho)$
\item Subadditivity: $M_{\text dist}(\rho_1\otimes \rho_2)\le M_{\text dist}(\rho_1)+ M_{\text dist}(\rho_2)$
\end{enumerate}
\end{proposition}
\begin{proof} 
	\begin{enumerate}
		\item By definition $M_{dist}(\psi)=0$ if and only $\psi\in\stab_0^{(n)}$, and so $\psi$ is a stabilizer state.
  \item The monotonicity descends from the monotonicity of the trace distance under trace-preserving CP maps. Because given a map $\xi:\stab_0^{(n)}\mapsto\stab_0^{(n^\prime)}$ we have
  \begin{align}
    M_{\text dist}(\xi(\rho))&=\min_{\sigma\in\stab_0^{(n^\prime)}}\frac{1}{2}\norm{\xi(\rho)-\sigma}_1\\&=\min_{\sigma\in\stab_0^{(n^\prime)}}\frac{1}{2}\norm{\xi(\rho-\sigma)}_1\\
    &\le\min_{\sigma\in\xi(\stab_0^{(n)})}\frac{1}{2}\norm{\xi(\rho-\sigma)}_1\\
    &\le\min_{\sigma\in\stab_0^{(n)}}\frac{1}{2}\norm{\rho-\sigma}_1=M_{\text dist}(\rho)
  \end{align}
where we used that $\stab_0^{(n^\prime)}\subseteq \xi(\stab_0^{(n)})$, the proof of the last statement is straightforward. One must observe that since $\xi$ is expressed in terms of stabilizer operations, the only operations that reduce the dimension are partial traces. Therefore, it is evident that since states in $\stab_0^{(n)}$ are mapped to stabilizer states in $\stab_0^{(n^\prime)}$ after a partial trace, the statement must hold true because there are more states whose partial trace returns the same state. 
\item Subadditivity:
\begin{equation}
\begin{split}
    M_{\text{dist}}(\rho_L)&= M_{\text dist}(\rho_1\otimes \rho_2)\\
    &=\frac{1}{2}\min_{\sigma\in\stab_0^{(n)}}\norm{\rho_1\otimes\rho_2-\sigma}_1
\\&=\frac{1}{2}\min_{\sigma\in\stab_0^{(n)}}\|\rho_1\otimes\rho_2-\sigma_1\otimes\sigma_2+\sigma_1\otimes\sigma_2-\sigma\|_1
\\&\leq\frac{1}{2}\|\rho_1\otimes\rho_2-\sigma_1\otimes\sigma_2\|_1\\
&\quad +\frac{1}{2}\min_{\sigma\in\stab_0^{(n)}}\|\sigma_1\otimes \sigma_2-\sigma\|_1
\\&\leq\frac{1}{2}\|\rho_1\otimes\rho_2+\rho_1\otimes\sigma_2 -\rho_1\otimes\sigma_2 -\sigma_1\otimes\sigma_2\|_1
\\&\leq\frac{1}{2}\|\rho_1\|_1\|\rho_2-\sigma_2\|_1+\frac{1}{2}\|\sigma_2\|_1\|\rho_1-\sigma_1\|_1
\\&\leq\frac{1}{2}\|\rho_2-\sigma_2\|_1+\frac{1}{2}\|\rho_1-\sigma_1\|_1
\end{split}
\end{equation}
where we used that $\sigma_1,\sigma_2$ are two stabilizer states, then  $\min_{\sigma\in\stab_0^{(n)}}\norm{\sigma_1\otimes\sigma_2-\sigma}=0$, and the tightest bound is obtained by minimizing over $\sigma_1$ and $\sigma_2$ proving the statement.
	\end{enumerate}
\end{proof}
 One can also define an entropic quantity the  Relative stabilizer entropy of magic:
\begin{definition}[Relative stabilizer Entropy of Magic]\label{def:rel_stab_magic}
The relative stabilizer entropy of magic of $\rho$ is given by
\begin{align}
    M_{RS}(\rho)= \min_{\sigma \in \stab_0^{(n)}}S(\rho||\sigma)  
 \label{eqn:rel_stab_magic}
\end{align}
\end{definition}
\begin{proposition}
    The relative stabilizer entropy is a magic monotone, i.e., 1. it is zero iff $\rho\in \stab_0^{(n)}$, 2. is invariant under Clifford conjugation, 3. is nonincreasing on average under stabilizer measurement, 4. is nonincreasing under partial trace and 5. is invariant under stabilizer composition.
\end{proposition}

\begin{proof}
The proof is similar to Ref.~\cite[Appendix A]{veitch_resource_2014}, where the only difference is the definition of $\stab^{(n)}$. Here we recount for completeness.
\begin{enumerate}
    \item Note that $S(\rho||\sigma)\geq0$ where equality is attained iff $\rho=\sigma$. Hence it only vanishes when $\rho\in \stab_0^{(n)}$, which by our definition is a stabilizer state. 
	
	\item Recall that $\rm STAB_0$ is invariant under Cliffords, therefore for $U\in\mathcal{C}(d^n)$ 
    \begin{equation}
    \begin{split}
 M_{RS}(U\rho U^{\dagger}) &= \min_{\sigma\in \stab_0^{(n)}} S(U\rho U^{\dagger}||\sigma)\\
 &=\min_{\sigma\in \stab_0^{(n)}} S(\rho ||U^{\dagger}\sigma U)\\
 &=\min_{\sigma\in \stab_0^{(n)}}S(\rho||\sigma).
    \end{split}
    \end{equation}
	
	\item The action of partial stabilizer measurements of the form $V_i=I\otimes |i\rangle\langle i|$ for some Pauli basis state $|i\rangle$ on $\stab_0$ returns a stabilizer state up to normalization. Using that $p_i=\Tr[\rho V_i], q_i=\Tr[\sigma V_i]$ and $\rho_i=V_i\rho V_i^{\dagger}, \sigma_i=V_i\sigma V_i^{\dagger}$, we can reuse the proof from Ref.~\cite{veitch_resource_2014} and note that $$\sum_i p_i S\left(\left.\frac{\rho_i}{p_i}\right\Vert\frac{\sigma_i}{q_i}\right)\leq S(\rho||\sigma).$$ The rest follows because $\sigma_i/q_i$ is again a stabilizer state.
    \item By Lieb and Ruskai \cite{lieb_ruskai}, it is shown that quantum relative entropy is nonincreasing under partial trace, i.e., $S(\Tr_B(\rho_{AB})||\Tr_B(\sigma_{AB}))\leq S(\rho||\sigma)$. 
    \item It is known that for any state $\tau$, $S(\rho\otimes \tau||\sigma\otimes \tau)=S(\rho||\sigma)$, hence the desired result follows when we take $\tau\in \stab_0^{(n)}$.
\end{enumerate}
\end{proof}
\subsection{(Anti)Flatness}\label{sec:antiflatness}
Flatness is the property of a quantum state that describes how close its spectrum is to a flat spectrum. From the operational point of view, the flatness of a state describes how flat is the classical probability distribution over a basis of pure states in which we can decompose it. Of course, this does not imply that this state will return a flat probability distribution for the measurements in any other basis. As an example of flat states, both the completely mixed state and pure states possess flat spectrum. Another notable example \cite{flammia_topological_2009} are the ground states of string-net Hamiltonians, e.g. the toric code and its generalizations. 

Flat states are the free states for the resource theory of flatness. We thus define the null set as
\be
\operatorname{FLAT}^{(n)}:=\left\{\sigma\in\mathcal{H}\,|\, \sigma^2= \frac{\sigma}{\rank{\sigma}}\right\}
\ee
Let us now define the following measure of antiflatness, that is, how far is a spectrum from the flat one. Of course, this quantity must measure the resource defined by $\operatorname{FLAT}^{(n)}$.
\begin{definition}\label{flatnessdist}
 We define the antiflatness of $\psi_A$ as \cite{flatness}
\begin{equation}\label{flatnessdef}
  \mathcal{F}(\psi_A)=\Tr(\psi_A^3)-\Tr^2(\psi_A^2)
\end{equation}
This quantity is very natural as it can be defined classically as the variance of a probability distribution $p(x)$ according to the probability distribution itself. More concretely, if one defines $\langle x\rangle_p := \sum_x xp(x)$, and one defines $\Delta p^2 := \langle (p-\langle p\rangle_p)^2\rangle_p$, then one has 
\ba
\mathcal{F}(\psi_A)= \Delta \lambda^2
\ea 
with $\{\lambda\}\equiv\mbox{spec} [\psi_A]$.
Of course, this quantity is zero on the flat states, that is,
 $\mathcal{F}(\sigma)=0$ for $\sigma\in\operatorname{FLAT}^{(n)}$ as it is immediate to verify. 
\end{definition}

There is a profound connection between antiflatness and magic. It connects magic, which is a property of the full state, to bipartite entanglement, and thus to the spectrum of a reduced density operator. In particular, it has been shown that\cite{flatness}, given a pure state $\psi_{AB}$ in a bipartite Hilbert space $\mathcal H = \mathcal H_A\otimes\mathcal H_B$, its linearized stabilizer entropy $M_{lin}$ is the average antiflatness of $\psi_A$ on the Clifford orbit,  that is,
\begin{equation}\label{flatnessdth1}
\langle\mathcal F(\psi_A^C)\rangle_C =f(d_A,d_B) M_{lin} (\psi_{AB})
\end{equation}
where $\psi^C_A = \Tr_B \psi_{AB}^C \equiv \Tr_B (C\psi_{AB} C^{\dag})$. It is also true that antiflatness shows typicality. Later, we will use this property to connect magic to spectral properties. The main message of \cref{flatnessdth1} is that, as long as the state $\psi$ is very entangled, and therefore $\psi_A$ is full rank, one can use the spectral quantity $\mathcal F(\psi_A)$ to probe magic. Note that---by definition---every density matrix is full rank on its support. 
This will come in handy in the next section.

 It is possible to define another monotone for the resource theory of flatness through the quantum relative entropy,
\be
\mathcal{F}_R(\rho)=\min_{\sigma \in \mathrm{FLAT}^{(n)}}S(\rho\Vert\sigma).
\ee
One can prove the following proposition
\begin{proposition}\label{prop:qrf}
Given a state $\rho\in\mathcal H$, it holds that
\begin{equation}
    \mathcal F_R(\rho)= S_{max}(\rho)-S(\rho)
\end{equation}  
\end{proposition}
See \cref{qrf} for a proof. 
Note that $\mathrm{FLAT}^{(n)} \supset \stab_0^{(n)}$ where $\stab_0^{(n)}$ is the set of states with zero stabilizer R\'enyi entropy, hence $\min_{\sigma\in \stab_0^{(n)}}S(\rho||\sigma)\geq F_R(\rho)$, therefore the flatness lower bounds the total subregion magic for any state. 
The same would not be true if $\stab^{(n)}$ is the usual stabilizer polytope, because it overlaps with $\mathrm{FLAT}^{(n)}$ but is not a subset as one can take a classical mixture of it such that the eigenvalues of $\rho$ are not equal (or zero).

Finally, let us define yet another flatness that will be natural for holography. Recall from Ref.~\cite{Dong1} that a variant of the R\'enyi entropy is given by,

\begin{equation}
\tilde{S}_n(\rho)=n^2\partial_n\left(\frac{n-1}{n}S_n(\rho)\right)
    =-n^2\partial_n(\frac{\log\Tr(\rho^n)}{n}).
\end{equation}

If we rewrite $\Tr(\rho^n)$ in terms of the spectrum $\{\lambda_k\}$ of $\rho$, it becomes

\begin{equation}
    \tilde{S}_n(\rho)=-n^2\partial_n(\frac{\log(\sum_{k}\lambda_k^n)}{n})=\log(\sum_{k}\lambda_k^n)-n\frac{\sum_k\lambda_k^n\log\lambda_k}{\sum_k \lambda_k^n}.
\end{equation}

Now we take the derivative of this expression and obtain another definition of antiflatness. In fact, this quantity is known as the \textit{Capacity of Entanglement}, which has been explored in the context of condensed matter system \cite{PhysRevLett.105.080501,PhysRevB.83.115322} and in quantum gravity \cite{PhysRevD.99.066012,Nakaguchi:2016zqi,Bueno:2022jbl,Zurek:2022xzl}  {\color{black} where it has an interpretation to leading order as the integrated metric fluctuations over the bulk minimal surface in holographic theories~\cite{PhysRevD.99.066012}.}

\begin{proposition}\label{def:braneflatness}
    $\partial_n\tilde{S}_n (\rho)$ is a measure of antiflatness in that $\partial_n\tilde{S}_n (\rho)=0$ if and only if $\rho$ has a flat spectrum.
\end{proposition}
\begin{proof}
Expanding the definition using the set of eigenvalues of $\rho$.
    \begin{equation}\label{eqn:holononflat}
    \begin{split}
    &\partial_n\tilde{S}_n(\rho)\\
    =&-n\frac{(\sum_k\lambda_k^n\log^2\lambda_k)(\sum_l\lambda_l^n)-(\sum_k\lambda_k^n\log\lambda_k)^2}{(\sum_k\lambda_k^n)^2}\\
    =&-n\frac{(\sum_{kl}\lambda_k^n\lambda_l^n\log^2\lambda_k)-(\sum_{kl}\lambda_k^n\lambda_l^n\log\lambda_k\log\lambda_l)}{(\sum_k\lambda_k^n)^2}\\
    =&-n\frac{\sum_{(kl)}\lambda_k^n\lambda_l^n(\log^2\lambda_k+\log^2\lambda_l-2\log\lambda_k\log\lambda_l)}{(\sum_k\lambda_k^n)^2}\\
    =&-n\frac{\sum_{(kl)}\lambda_k^n\lambda_l^n\log^2\frac{\lambda_k}{\lambda_l}}{(\sum_k\lambda_k^n)^2},
\end{split}
\end{equation}
where $\sum_{(kl)}$ denotes sum over each pair of distinct indices $k\neq l$. Note that each term in the numerator is non-negative. Therefore $\partial_n\tilde{S}_n=0$ if and only if $\log{\frac{\lambda_k}{\lambda_l}}=0$, which is equivalent to $\lambda_k=\lambda_i$ for all $k,l$. 
\end{proof}

This antiflatness (\cref{eqn:holononflat}) can be connected to (\cref{flatnessdef})
by first noticing that the antiflatness $\mathcal{F}(\rho)$ corresponds to the variance of $\rho$. The proof is straightforward
\begin{equation}
\begin{split}
\mathcal{F}(\rho)=&\tr(\rho^3)-\tr^2(\rho^2) =\tr(\rho \, \rho^2) -\tr^2(\rho \, \rho)\\
=&\langle\rho^2\rangle_\rho - \langle \rho \rangle_\rho^2 = \operatorname{Var}_\rho(\rho)
\end{split}
\end{equation}
Let us connect this definition with the derivative at $n=1$. 
Let $\rho\equiv\sum_k\lambda_k \ket{\lambda_k}\bra{\lambda_k}$.
Note that the following relation can also be written as a variance, by defining $p_k=\frac{\lambda_k^n}{\sum_k\lambda_k^n}$, it is easy to observe that $\sum p_k=1$ and we can define the state 
\ba
\Xi:=\sum_k p_k \ket{\lambda_k}\bra{\lambda_k}
\ea
 and so 
\begin{equation}
\begin{split}
\partial_n\tilde{S}_n(\rho)=&-n \sum_{kl}p_k p_l (\log^2 \lambda_k-\log\lambda_k\log\lambda_l) \\
=&-n\langle\log^2\rho\rangle_{\Xi}+n\langle\log\rho\rangle_{\Xi}^2=-n\operatorname{Var}_\Xi(\log\rho)
\end{split}
\end{equation}
Let us compute it for $n=1$
\begin{equation}
\begin{split}\label{sigmalogrho}
\left.{\partial_n \tilde{S}_n}(\rho)\right\vert_{n=1}=& -\sum_{kl}\lambda_k\lambda_l \log\lambda_k(\log\frac{\lambda_k}{\lambda_l})\\ 
=&-\sum_k \lambda_k \log^2 \lambda_k + \sum_{kl} \lambda_k\lambda_l\log\lambda_k\log\lambda_l \\
=&-\tr(\rho\log^2\rho)+\tr^2(\rho\log\rho)\\
=&-\langle\log^2\rho\rangle_\rho+\langle\log\rho\rangle_\rho^2=-\operatorname{Var}_{\rho}(\log\rho)\\
\equiv&- C_E (\rho)
\end{split}
\end{equation}
The quantity $C_E$ is also known as capacity of entanglement. 
Interestingly, when $n=1$, $\Xi$ coincides with $\rho$. Seeing $\log\rho$ as a function of $\rho$, the variances between the two quantities are connected. We make use of standard techniques of error propagation to get the relationship between $\operatorname{Var}_\rho(\rho)$ and $\operatorname{Var}_\rho(\log(\rho))$. 
\be \label{eqn:approxrenyid}
\operatorname{Var}_\rho(\log(\rho))\approx\frac{\operatorname{Var}_\rho(\rho)}{\langle \rho \rangle_\rho^2}=\frac{\operatorname{Var}_\rho(\rho)}{\operatorname{Pur}(\rho)^2}=\frac{\mathcal{F}(\rho)}{\operatorname{Pur}^2(\rho)},
\ee 
{\color{black}The approximation is valid when the spectrum of $\rho$ has a negligible higher order moments compared to the variance (see Appendix \ref{app:bound}).  Therefore, the two measures coincide in the near-flat or weak entanglement regime. }

{\color{black}
In fact,~\cref{eqn:holononflat} has a convenient rewriting as the variance of the modular Hamiltonian spectrum. Given a state $\rho\equiv \sum_k \lambda_k\ket{\lambda_k}\bra{\lambda_k}$,
its eigenvalues can be written as $\lambda_k:=\exp(- E_k)$. This defines the modular Hamiltonian, 
\ba
H:=&-\log \rho\\
=&\sum_k E_k\ketbra{\lambda_k}{\lambda_k}
\ea

Applying~\cref{sigmalogrho} and perform some simple algebra, one obtains that,
\begin{equation}
\begin{split}
\left.{\partial_n \tilde{S}_n}(\rho)\right\vert_{n=1}=&
-\left(\langle H^2 \rangle_{\rho} -\langle H\rangle_{\rho}^2 \right)\\
=& -\langle (E_k-E_l)^2\rangle_{kl} 
\end{split}
\end{equation}

This result extends naturally to any $n$. Noting that $\Xi=\exp(-n H)Z^{-1}[n]$, further algebra leads to
\ba
{\partial_n \tilde{S}_n}(\rho)=-n (\langle H^2 \rangle_{\Xi}-\langle H\rangle_{\Xi }^2).
\ea
}

\subsection{Nonlocal magic, entropy, and antiflatness}\label{sec:magicbounds}
In this section, we are going to introduce the concept of \emph{nonlocal magic}, and how it relates to both entanglement and antiflatness. 

\begin{definition}[multipartite nonlocal magic]
Given $M$ a measure of magic and $\psi_{A_1\dots A_n}\equiv\st{\psi _{A_1\dots A_n}}$ a pure state, we define as $n$-partite nonlocal magic
\begin{equation}
		M^{(n-\rm NL)} (\psi_{A_1\dots A_n}) := \min_{U=\otimes_{i=1}^n U_{A_i}} M(U \psi_{A_1\dots A_n} U^{\dagger}).
  \label{eq:nlmagic}
\end{equation}
\end{definition}
As we exclusively discuss the case of bipartite nonlocal magic when $n=2$ for the rest of this work, we set $A=A_1, B=A_2$ and simply refer to $M^{(NL)}=M^{(2-NL)}$ as nonlocal magic for convenience. 

Intuitively, nonlocal magic is the nonstabilizerness that lives in the correlation between $A$ and $B$ because $U_A\otimes U_B$ removes all ``local'' magic in $A$ or $B$ separately. This is distinct from other notions of long-range magic \cite{white_conformal_2021,Bao_2022,tarabunga2023critical}.  Note that $A,B$ themselves can be multi-qubit systems, so $U_A,U_B$ need not be single-qubit unitaries.

In this work, we will use as measures of magic $M_{\text{dist}}$ and the two relative entropies of magic $M_{R} and M_{SR}$. 

\subsubsection{Nonlocal  magic and flatness}
Let us start with a general relation valid for \textit{any} measure of antiflatness and \textit{any} measure of nonlocal magic. 

\begin{lemma}\label{lemmaNL}
A pure quantum state $\ket{\psi}$ possesses no nonlocal magic, that is,  $M^{NL}(\ket{\psi})=0$, iff $\ket{\psi}$ is unitarily locally equivalent to a state $\ket{\psi'} = U_A\otimes U_B\ket{\psi}$ with flat reduced density matrix $\psi_A'\equiv\tr_B\st{\psi'}$ with integer R\'enyi entropies\footnote{In this work, information is measured using bits. Accordingly, entropies are computed using $\log_d$.}. In formula,
\be
M_{NL}(\ket{\psi})=0\iff 
F(\psi_A)=0 \wedge \rank(\psi_A)=d^{r_A},\,\,r_A\in\mathbb{N} 
\ee
\begin{proof}
     Let us start from the left-to-right implication. We employ the fact that any faithful measure of magic $M(\ket{\psi})$ vanishes on the free states.
For any such measure, its nonlocal counterpart with respect to the bipartition $A|B$ is $M^{NL}(\ket{\psi}):= \min_{U_A\otimes U_B}M(U_A\otimes U_B\ket{\psi})$. Given $M^{(NL)}(\ket{\psi})=0$, then we know that there exist a bilocal unitary $U_A\otimes U_B$ such that $\ket{\psi^{\prime}}\equiv U_A\otimes U_B\ket{\psi}\in \stab_0^{(n)}$. Since $\stab_0^{(n)}$ is closed under partial trace, see \cref{sec:magicmeasure}, then $\psi_A'\in\stab_0^{(n)}$. We know that $\psi_A'\in\mathrm{FLAT}^{(n)}$. Moreover, being $\psi_A'\in\stab_0^{(n)}$ we know that $\rank(\psi_A)=d^{r_A}$ with $r_A\in\mathbb{N}$. Let us now show that also the converse is true.  Consider a flat state $\ket{\psi}$, that is, a state such that its reduced density matrix $\psi_A\equiv\tr_B\ketbra{\psi}{\psi}=\frac{1}{d^{r_A}}\sum_{i}\ketbra{\phi_i}{\phi_i}_A$ where the sum run on $d^{r_A}$ many rank-one projectors $\ketbra{\phi_i}{\phi_i}_A$. Note that we exploited the fact that $S_{\alpha}(A)=r_A\in\mathbb{N}$ for every $\alpha\in[0,\infty)$. Via the Schmidt decomposition, we can write the state as $\ket{\psi}=\sum_{i}\frac{1}{\sqrt{d^{r_A}}}\ket{\phi_i}_A\otimes \ket{\psi_i}_B$. Without loss of generality, we choose now $|A|<|B|$.  We further know that $\langle \phi_i|\phi_j\rangle=\langle \psi_i|\psi_j\rangle=\delta_{ij}$. Now  choose $U_A$ (respectively $U_B$) such that $U_{A}\ket{\phi_i}_A=\ket{i}_A$ (respectively $U_{B}\ket{\psi_i}_B=\ket{i}_B$) for $\ket{i}_A$ (respectively $\ket{i}_B$) being the computational basis on $A$ (respectively $B$). We obtain 
\be\label{magicflatness}
U_A\otimes U_B\ket{\psi}=\sum_{i}\frac{1}{\sqrt{d^{r_A}}}\ket{i}_A\otimes \ket{i}_B\equiv \ket{EPR}_{A\bar{A}}\otimes \ket{j}_{B\setminus \bar{A}}
\ee
where $\ket{EPR}_{A\bar{A}}$ is a EPR pair between the full $A$ and \textit{any} subsystem $\bar{A}\subset B$ such that $|A|=|\bar{A}|$, while $\ket{j}$ is a computational basis state on $B\setminus \bar{A}$.    Since  $\ket{EPR}_{A\bar{A}}\otimes \ket{j}_{B\setminus \bar{A}}$ is a stabilizer state, we obtain
\be
\begin{split}
0=M(U_A\otimes U_B\ket{\psi})\ge& \min_{U_A\otimes U_B}M(U_A\otimes U_B\ket{\psi})\\
=&M^{NL}(\ket{\psi})\ge 0
\end{split}
\ee
  \end{proof}
\end{lemma}
Notice that a vanishing nonlocal magic is a sufficient  condition for antiflatness to be zero. However, there are possibly states with noninteger R\'enyi entropies that can possess some nonlocal magic without being guaranteed that antiflatness is nonvanishing. With the additional condition of integer R\'enyi entropy, also the other implication holds, that is, a flat state implies vanishing non local magic for any sensible measure of non local magic. {\color{black}
Therefore, the only difference between states with $M_{NL}=0$ and states with antiflatness $\mathcal{F}=0$ is an entanglement spectrum multiple of $d$, where $d$ is the local dimension. This hints at the existence of a tighter connection between nonlocal magic and antiflatness, which we rigorously establish below. Consider the capacity of entanglement of a pure state $\ket{\psi}$, with reduced density matrix $\rho_A=\tr_B\ketbra{\psi}{\psi}$, $C_{E}(\psi)=-\partial_n\tilde{S}_n(\rho_A)|_{n=1}$~\cite{PhysRevD.99.066012}. In \cref{sigmalogrho}, we showed that it can be expressed as  
\be
C_{E}(\psi)=\tr(\rho_A\log^2\rho_A)-\tr^2(\rho\log\rho_A)
\ee
As a consequence of~\cref{eqn:approxrenyid}, one can show that $C_{E}(\psi)=0$ iff $\psi$ has a flat entanglement spectrum. In the next definition, we accommodate this measure for probing nonlocal magic.
\begin{definition}[nonlocal magic from antiflatness]\label{def:nonlocfromcapacity} Let $\psi$ a pure bipartite state on $n$ qudits. We define the nonlocal magic inherited from the capacity of entanglement as 
\be
M_{C_E}^{NL}(\psi)=\lceil\tr(\rho_A\log^2\rho_A)\rceil- S^{2}(\rho_A)
\ee
where $\lceil\cdot\rceil$ is the ceiling function.
\end{definition}
Notice that, up to a floor function, $M_{C_E}^{NL}$ is nothing but the capacity of entanglement. Let us now explore its properties as probe of nonlocal magic.
\begin{theorem}\label{th:capacity}
    Let $\psi$ a bipartite pure state on $n$ qudits. The following facts are true
    \begin{enumerate}
        \item $C_{E}(\psi)\le M_{C_E}^{NL}(\psi) $
        \item $M_{C_E}^{NL}(\psi)=0$ if and only if $\psi$ contains only local magic, that is, $M^{NL}(\psi)=0$;
        \item $|M_{C_E}^{NL}(\psi)-C_{E}(\psi)|\le 1$;
        \item for qubits ($d=2$), let $M_{dist}^{NL}(\psi)$ be the nonlocal magic measure defined through \cref{def:mod_tr_dist_mag}, then
        \be
        \frac{M_{C_E}^{NL}(\psi)}{n^2}\le M_{dist}^{NL}(\psi)+O(n^{-2})
        \ee
    \end{enumerate}
\begin{proof} Item 1 descends trivially from \cref{def:nonlocfromcapacity}. Let us show item 2.  From the fact that $M_{C_E}^{NL}(\psi)=0$, it follows that first $C_{E}(\psi)=0$, and that $S(\rho_A)^2\in\mathbb{N}$. Given that $C_{E}(\rho_A)=0$, then $\rho_A$ is flat $\rho_A=\frac{1}{R}\sum_{i}\ketbra{i}{i}$. Let us show that $R=d^{r_A}$ with $r_{A}\in\mathbb{N}$. We have the following equality $\log^2 R\in\mathbb{N}$, which implies $R=d^{\sqrt{N}}$ where $N\in\mathbb{N}$. Given that also $R\in\mathbb{N}$, then $\sqrt{N}\in\mathbb{N}$. We thus conclude that the state has local magic, thanks to \cref{lemmaNL}. The other direction of item 1, follows from \cref{lemmaNL} and a trivial calculation. Item 3, follows from the fact that for any function $f(\rho)$, then $|f(\rho)-\lceil f(\rho)\rceil|\le1$. This concludes the proof.
\end{proof}
\end{theorem}

}

{\color{black}
\subsubsection{Nonlocal magic, flatness, and entanglement}

In this section, as anticipated in the previous one, we highlight the strong connection between nonlocal magic, flatness, and entanglement. Specifically, we establish lower and upper bounds for nonlocal magic based on the trace distance of $M_{dist}$ defined in~\cref{def:mod_tr_dist_mag}.  
We have:
\begin{theorem}\label{th:magicdist} Let $\psi_{AB}$ be a pure state  in a bipartite Hilbert space $\mathcal H = \mathcal H_A\otimes\mathcal H_B$,
then
	\begin{equation}
        \begin{split}
		&\mathcal{F}(\psi_{ A})/8 \le M_{dist}^{(NL)}(\psi_{AB})\\
  &\le  \sqrt{1-e^{S_{max}(A)}+e^{S_{\infty}(A)}\left(1-\frac{e^{\log (d) \lfloor S_{max}(A)/\log d \rfloor}}{e^{S_{max}(A)}}\right)}
        \end{split}
	\end{equation}
\end{theorem}
where $\lfloor \cdot \rfloor$, is the floor function. The proof can be found in~\cref{app:proofthmd}.

 As we shall see in \cref{section:brane}, the \cref{lemmaNL} and \cref{th:magicdist} will have important consequences for the relationship between the nonlocal magic in the CFT side and gravity in AdS.

}

A similar result can also be obtained by considering instead of the notion of trace distance of magic, the one based on the relative stabilizer entropy.

\begin{theorem}\label{th:relstab}
Let $\psi_{AB}$ be a pure state, then 
    \begin{equation}\label{eq:relstab}
    \begin{split}
        S_{max}(A)-S(A)&=\mathcal F_R(\psi_A) \leq \min_{U_A}M_{RS}(U_A\psi_A U_A^{\dagger})\\
        &\leq M^{(NL)}_{RS}(\psi_{AB})\leq\log d \lceil   S_{max}(A)/\log d  \rceil.
    \end{split}
    \end{equation}
\end{theorem}
Here $S(A)=S(\rho_A)$ and $S_{max}(A)=S_{max}(\rho_A)$. The proof can be found in~\cref{proofth2}.
Let us briefly comment on the tightness of the bound. It is clear that when $|\psi\rangle_{AB}$ has a dominant Schmidt coefficient and many small trailing singular values, then the bound is essentially tight. A case in point is $\sqrt{1-\epsilon}|00\rangle+\sqrt{\epsilon}|11\rangle$. However, the upper bound is quite loose for states with near-flat spectrum, e.g. $\epsilon=1/2$. This is an artifact of choosing the maximally mixed state as a reference even though other stabilizer states clearly yield a lower distance.

A similar upper bound can be obtained with the usual relative entropy measure of magic. 
\begin{proposition}[Entanglement upper bounds NL magic]\label{prop:RelativeEnt}
    Suppose $\rho_{AB}$ is pure, and \begin{align}
        M^{(NL)}_R(\rho_{AB})=\min_{U_A\otimes U_B} M_R((U_A\otimes U_B)\rho_{AB}(U_A\otimes U_B)^{\dagger}),
    \end{align} then $M_R^{(NL)}(\rho_{AB})\leq S(A)=S(B)$, where $S(A)$ is the von Neumann entropy of subsystem $A$.
\end{proposition}
The proof is given in~\cref{proofprop4}. This upper bound suffers from the same drawbacks as (\cref{eq:relstab}) for states that are maximally entangled.

\subsubsection{Magic estimates}\label{section:estimate}
As minimization can be difficult for the relative entropy measure, let us also derive a tighter upper bound based on a computable measure of magic, that is, the stabilizer R\'enyi entropy~\cite{PhysRevLett.128.050402}. 
To do so, we can pick a good estimate that is reasonably close to the minimum. Suppose the entanglement spectrum of the state under the same bipartition $AB$ is $\{\lambda_i\}$, construct a state 
\begin{equation}
    |\psi'\rangle_{AB}=\sum_{i=0}^{2^n-1}\sqrt{\lambda_i} |s_i\rangle|s_i\rangle,
\end{equation}
where $\{|s_i\rangle\}$ are eigenstates of a stabilizer group $\mathcal{S}=\{S_1,S_2, \cdots, S_n\}$ such that for any $S_k$ in $\mathcal{S}$, $S_k\ket{s_i}=\pm\ket{s_i}$. Because the entanglement spectrum is invariant under local unitary $U_A\otimes U_B$, $|\psi'\rangle$ is a reasonable construction such that the reduced density matrix on both $A$ and $B$ are within the \textit{stabilizer polytope}, and hence have vanishing local magic by the relative entropy measure $M_R$.    Note that other choices of the Schmidt basis may yield lower overall magic on $AB$, therefore $M(|\psi'\rangle)$ provides an upper bound of nonlocal magic.

We now present an estimate of $M(|\psi'\rangle)$ using the stabilizer R\'enyi entropy measure.



\begin{proposition}
The nonlocal stabilizer R\'enyi entropy estimate for a state with entanglement spectrum $\{\lambda_i\}$ is 
    \begin{equation}\label{estimate_prop}
        \mathcal{M}_2 (\{\lambda_i\})=\mathcal{M}_2(\sum_{i=0}^{2^n-1}\sqrt{\lambda_i}|s_i\rangle|s_i\rangle), \qquad \lambda_i\geq\lambda_j,\ \ \text{for} \ i<j.
    \end{equation}
\end{proposition}
Note that this nonlocal magic estimate does not depend on the choice of stabilizer group $\mathcal{S}$. However, the ordering of eigenvalues does affect its magnitude. Remarkably, one can obtain an exact expression for $\mathcal{M}_2(\{\lambda_i\})$. {A similar expression has also been obtained by Ref.~\cite{RK_wavefcn} but in a different context. With additional ancillae, it is identical to the one below after applying a global Clifford unitary. }

 \begin{theorem}\label{thm:nlSRE}
     The nonlocal stabilizer Rényi entropy estimate is 
     \begin{equation}\label{eq:analyticalM}
     \begin{split}
    \mathcal{M}_2(\{\lambda_i\})=-\log&\left(\sum_{i_1,i_2,i_3,i_4=0}^{2^n-1}\sqrt{\lambda_{i_1}\lambda_{i_2}\lambda_{i_3}\lambda_{i_4}\lambda_{i_3\wedge i_2\wedge i_1}}\right.\\
    &\left.\times\sqrt{\lambda_{i_4\wedge i_2\wedge i_1}\lambda_{i_1\wedge i_3\wedge i_4}\lambda_{i_2\wedge i_3\wedge i_4}}\right),
    \end{split}
     \end{equation}
where $\wedge$ denotes the bitwise $\XOR$ operation. This expression depends on the ordering of the eigenvalues and reaches its minimum when the eigenvalues are in the descending order, that is, 
$\lambda_i\geq\lambda_j$ for  $i<j$.
 \end{theorem}

In \cref{sec:numerics} we present numerical results of $\mathcal{M}_2(\{\lambda_i\})$ for finite-sized physical system. It is helpful to see that the estimate constitutes a nonlocal magic upper bound. 

\begin{corollary}\label{th:srebound}
Let $\{\lambda_i\}$ be the Schmidt values for $|\psi\rangle_{AB}$ when bipartitioning the system into $A$ and $B$. The nonlocal stabilizer Rényi entropy  is upper bounded by 
\begin{equation}
\begin{split}
\mathcal{M}_2^{NL}(|\psi\rangle_{AB})\leq& \mathcal {M}_{2}(\{\lambda_i\})\\
\leq& \min\{2S_2(A),4(S_{\text{max}}(A)-S_{1/2}(A))\}
\label{eqn:avgnlsre}
\end{split}
\end{equation}
where $S_{max}(A)=n\log 2$, $S_{\alpha}(A) = S_{\alpha}(\rho_A)$ with $\rho_A=\Tr_{B}[|\psi\rangle\langle\psi|]$. 
\end{corollary}

See~\cref{app:estimate} for the proof. Based on this result, we discuss two regimes. One is when the spectrum is almost flat. In this regime, the bipartite nonlocal magic is upper bounded by,
\begin{equation}
    \mathcal{M}_2(\{\lambda_i\})\leq 4(S_{\text{max}}(A)-S_{1/2}(A)).
\end{equation}
This has the interpretation as antiflatness. 
Although the measure of magic is different, we see that this gives a much tighter bound compared to (\cref{th:relstab}) in the near-flat regime.



\begin{remark}\label{rmk:1}
    Haar random states have small bipartite nonlocal magic. 
\end{remark}
We see that $\mathcal{M}_2\sim S_0-S_{1/2}$ whereas the lower bound from relative stabilizer entropy measure in (\cref{th:relstab}) is $S_0-S_1$, both are bounded by a constant for Haar random states\cite{avgrenyi} --- for small $\alpha$, $\dim A\ll \dim B = m$, $S_0-S_{\alpha}\leq \frac{\alpha}{2} +O(1/m^2)$.
This is somewhat surprising because Haar random states are magic rich and have nontrivial total magic~\cite{white2020mana,Liu_2022}. However, the magic sustained by their bipartite entanglement is small even though local magic in any subregion $A$ with $|A|\gg |B|$ can be large.  

Another limit is when $S_0(A)\gg S_{1/2}(A)$, which applies for quantum field theory. In this regime the magic is approximated by the second R\'enyi entropy, 
\begin{equation}
\mathcal{M}_2(\{\lambda_i\})\leq 2S_2(A). 
\end{equation}
As we shall see in~\cref{sec:CFT}, this is consistent with our MERA intuition for conformal field theories.


\subsubsection{Smoothed magic}
The concept of magic and its bound, as discussed earlier, are applicable to systems with finite dimensions. However, in quantum field theory, the Hilbert space has an infinite dimension. In this case, the bounds given by max entropy in~\cref{th:relstab} can easily be divergent. To produce a nontrivial bound, it is imperative to introduce the `smoothed magic', defined as
\begin{equation}
    M_{RS}^{\epsilon}(\rho)\defeq \min_{\Vert\chi-\rho\Vert<\epsilon}M_{RS}(\chi),
\end{equation}
as well as the `smoothed nonlocal magic', defined as
\begin{equation}
\sma{\rho_{AB}}\defeq \min_{\Vert\chi_{AB}-\rho_{AB}\Vert<\epsilon}M_{RS}^{(NL)}(\chi_{AB}).
\end{equation}

For this, a smoothed version of~\cref{th:relstab} holds.

\begin{theorem}\label{th:smoothed}
Let $\rho_{AB}$ be a pure state, then 
\begin{equation}
\begin{split}
     S_{max}^{\epsilon}(\rho_A)-(1-\epsilon)^{-1}S(\rho_A)&\leq \sma{\rho_{AB}}\\
     &\leq \log d \lceil   S_{max}(A)^{\epsilon}/\log d  \rceil.
     \label{eqn:smoothinequality}
\end{split}
\end{equation}
\end{theorem}
where the smoothed maximal entropy is defined as
\begin{equation}
    S_{max}^{\epsilon}(\rho)=\min_{\Vert\chi-\rho\Vert<\epsilon} \ln(\rank(\chi)). 
\end{equation}

The proof can be found in~\cref{proofth3}. As stated by~\cref{th:smoothed}, the magic is bounded from below by the difference between the smoothed maximal entropy and the entanglement entropy which is finite for conformal field theories.


Before we discuss CFTs, let us examine the physical meaning of the lower bound, which is the difference between smoothed max entropy and the von Neumann entropy. In addition to the antiflatness of the entanglement spectrum, this quantifies the \textit{compressibility} of a state~\cite{Akers_2021}. Consider a bipartition of the state followed by a Schmidt decomposition. It is compressible if we can still well approximate it after truncating the less significant singular values, as one is wont to do in DMRG. Here we can show that this compressibility gap which lower bounds smoothed nonlocal magic also quantifies the classical hardness in simulations. 

Let us build up the following argument by recalling that there are states such as random stabilizer states that have high entanglement but are classically easy to simulate. Since magic and entanglement capture two orthogonal perspectives of quantumness, are there quantum states with low entanglement but high magic that are classically hard to simulate? Na\"ively, a state with high magic will have high stabilizer rank, which is hard in the stabilizer simulation. On the other hand, the system will be classically hard using the tensor network method if it has high bond dimensions. However, a folk theorem in tensor network suggests that the small entanglement would permit one to capture the state with a tensor network whose bond dimension only needs scale as $O(e^S)$ where $S$ is the von Neumann entropy of each subsystem. Therefore, it seems that as long as the entanglement is small, there should be a classically easy description. However, one needs to be careful in applying this lore as it is known that there exist states with low entanglement but classically complex\cite{ge_area_2016}. 

More precisely, consider an exact matrix product state (MPS) description of a state with low entanglement such that for any subsystem $A$, $S(A)\ll \log \rank(\rho_A)$ where we have taken the bond dimension $\chi$ to be sufficiently large to reproduce the state exactly. One would be tempted to truncate the singular values and only keep $O(e^S)$ as suggested by the folk theorem. However, we note that this truncation is only justified if there exists $\sigma_A$ with  $||\sigma_A-\rho_A||<\epsilon$ such that  $$\Delta S_{\epsilon}(A)=S_{\rm max}^{\epsilon}(A)-S(A)=\log \rank(\sigma_A)-S(A)$$ is small compared to $S(A)$. In other words, the state is (perfectly) compressible. Such is indeed true for conformal field theory ground states, where $\Delta S_{\epsilon}\sim \sqrt{S\log(1/\epsilon)}$. However, this is not true in general. For example, consider a state $\ket{\psi}=\frac{1}{\mathcal{N}}\sum_{i=1}^r\frac{1}{\sqrt{i}}\ket{i}_{A}\ket{i}_{B}$. The smoothed max entropy $S_{max}^{\epsilon}=\log{r}-\epsilon$, while entanglement entropy is nearly half of it, $S\approx \frac{1}{2}\log{r}$. In holography, Akers and Penington \cite{Akers_2021} argued that certain state mixtures, such as that of a thermal and pure state, can lead to an arbitrarily large $\Delta S_{\epsilon}(A)$. 

Therefore, high incompressibility on the one hand forces high tensor network bond dimension, and on the other necessitates high nonlocal magic from \cref{th:smoothed}. This implies that such states will be classically hard to simulate and sharpens a general empirical observation that relates magic to classical complexity. Furthermore, if $S\ll\Delta S_{\epsilon}\approx S_{\rm max}^{\epsilon}$, then both the lower and upper bounds are approximately saturated. In this case, the smoothed nonlocal magic provides a quantitative measure for the classical hardness of simulating such states. Treating magic as roughly as the log of stabilizer rank and bond dimension, one would expect that classical resource of order $O(\exp(M_{RS}^{(NL,\epsilon)}))$ will be needed. It then follows that \textit{simulating such incompressible states is classical hard} using not only the tensor network method but also the stabilizer and the Monte Carlo method\cite{magicMC} by having large magic\footnote{A careful treatment of this problem should include other formulations of non-magical processes like Gaussian states, matchgates with \cite{bu_stabilizer_2023}.}.

\section{Magic in conformal field theories}\label{sec:CFT}
Having seen a quantitative connection between antiflatness in entanglement spectrum and nonlocal magic, we examine these relations in the context of CFTs.

\subsection{Geometric Interpretation through tensor networks}\label{section:MERA}

To figure out (1) how much nonlocal magic there is in a CFT and (2) how such magic connected with the antiflatness of the entanglement spectrum, it is instructive to first look at an intuitive picture from tensor networks.
For CFTs with small central charges, MERAs have been shown to be good approximations of CFT ground states $|\psi\rangle_{AB}$. By extension, it also holds for products of CFTs with small central charge. Let us assume that the tensor network structure remains valid for arbitrary degree of accuracy, perhaps at the cost of increasing the bond dimension, which is supported by empirical observations. Using this as a heuristic, we deduce that local unitary deformations $U_A\otimes U_B$ unitarily ``distills'' an entangled state\footnote{For simplicity, we will refer to such a process as distillation from now on. However, one should note that it is distinct from the usual entanglement distillation of perfect Bell pairs unless otherwise specified.} between $A$ and $B$ with log Schmidt rank that is upper bounded by the number of edge cuts (green triangle in \cref{fig:mera}). As such cuts scale linearly with the size of the RT surface, i.e. the boundary of the triangle in the bulk, the log of Schmidt rank must be bounded by the number of edge cuts which scale the same way as entanglement entropy in this case. This implies that the nonlocal Magic in CFTs should scale linearly with the area of the Ryu-Takayanagi surface.

In fact, we can almost identify the optimal distilled state that has the same Schmidt rank but removes the unnecessary zero eigenvalues by just acting mostly unitaries and disentanglers. Let the blue rectangles at the bottom layer be the CFT ground state but at a more coarse-grained scale. As the ground state is an IR fixed point, we can simply use it as an input in the MERA to generate the more fine-grained state on the top layer. We can decompose the IR state by Schmidt decomposition, and the Schmidt rank is upper bounded by the bond dimension (here the bond is represented as three edges on each side of the blue rectangle on the bottom assuming the worst case volume law upper bound in the central region). By acting disentanglers and isometries in B followed by global unitaries on the subsystems represented by the blue rectangles on two sides of the bottom layer in the IR ground state, we ``pushed'' the subregion $B$ on the top layer to the red boundary by acting $U_B$, which now lives on $\partial B$. Similarly, acting $U_A$ by running unitaries and isometries in $A$, we remove the bulk dof and push $A$ to $\partial A$, marked by the orange lines. The qubits on $\partial A$ and $\partial B$ are entangled and their entanglement spectrum is unchanged since we only applied unitaries $U_A\otimes U_B$. 

Let $|\partial A|,|\partial B|$ be the number of edges in $\partial A,\partial B$. The distilled state $|\chi\rangle_{AB}$ is not optimal as $ \log rank(\rho_A)\leq |\partial A|<|\partial B|$, where we would have hoped that $|\partial A|=|\partial B|=\log rank(\rho_A)$, but this is close enough as $|\partial A|$ and $|\partial B|$ both scale as $\sim \log |A|$ as the blue region that contributed to suboptimality in the edge cuts is only constant (AdS) radius away from the true minimal surface. The number of edge cuts on the bottom layer is always bounded as the width of the MERA past causal cone is bounded. This means that $|\partial A|+const = |\partial B|$ where the constant depends on the network discretization. For binary MERA it stabilizes at four to six sites. 

As a consequence, after the removal of local magic in each wedge, the remaining magic is tied up into the interface between $A$ and $B$ marked by the region shaded in blue. Since the amount of magic generically scales linearly with the number of tensors, for a contiguous subregion $A$, the size of the interface region scales as $\log |A|$, which is proportional to the size of the RT surface up to subleading corrections. Note that while it may be possible to lower the size of this interface region further by local unitary transformations, the number of sites it involves must be lower bounded by the minimum number of edges connecting $A$ and $B$, which is given by $|\partial A|$. Heuristically, consider a case where all the bipartite entanglement between $A$ and $B$ have been ``distilled'' into imperfect Bell pairs connecting the two complementary regions. Then for any additive measure of magic, the nonlocal magic should scale linearly with the number of such imperfect Bell states, which is again proportional to the length of the minimal surface.

\begin{figure}
    \centering
    \includegraphics[width=\linewidth]{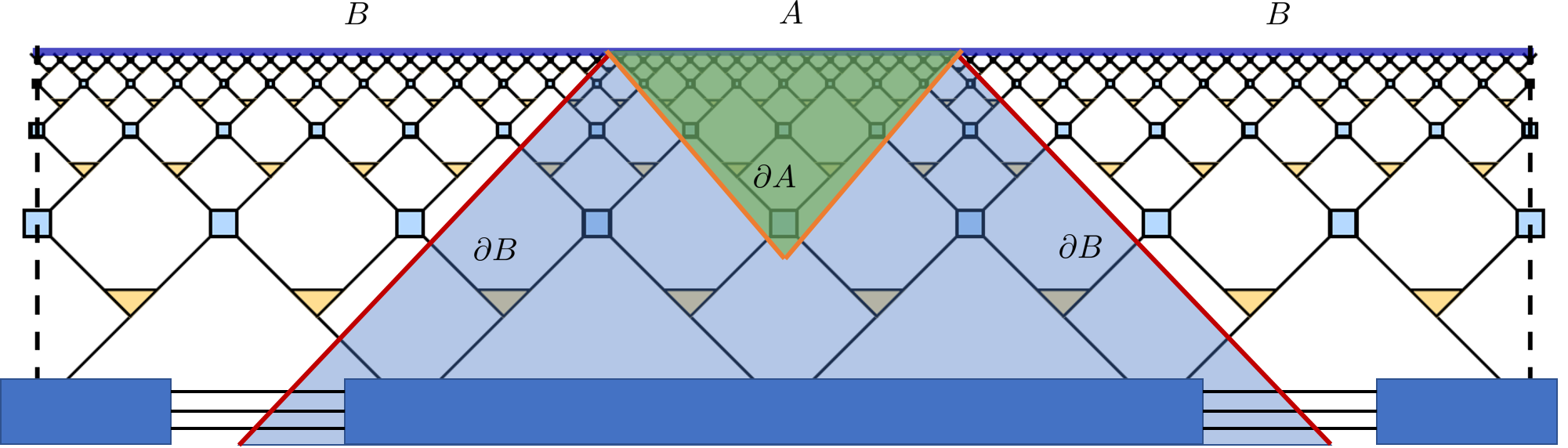}
    \caption{Green: past causal domain of dependence of $A$, and union of blue and green: past causal cone of $A$. Time runs upward.}
    \label{fig:mera}
\end{figure}

More precisely, we observe that the tensor network of the interface region is a MPS (\cref{figD1b}) by removing the local unitaries. The remaining structure contributes to the nonlocal magic is shown in \cref{figD1a}. Each matrix in the chain consists of two isometries and one disentangler. 

\begin{equation}\label{eq:MPS}
\begin{split}
\ket{\chi}_{AB}=M_1^{(s_1r_1)}M_2^{(s_2r_2)}\cdots &M_n^{(s_nr_n)}  \ket{s_1s_2\cdots s_n}_{\partial A}\\
&\otimes \ket{r_1r_2\cdots r_n}_{\partial B}.
\end{split}
\end{equation}

We expect the magic of this state to scale linearly with the number of matrices, namely the size of the light-cone, $\min\{|\partial A|, |\partial B|\}$. Indeed, we verify that magic scales as volume of the MPS, which is $\sim \log|A|$.

\tikzset{pics/matrix/.style args={#1#2#3}{
    code={
       \def\s{0.4}
    \draw[fill=orange] (-#1,0)--(-#1/2,-\s)--(-#1*1.5,-\s)--cycle;
   \draw[fill=orange] (#1,0)--(#1/2,-\s)--(#1*1.5,-\s)--cycle;
  \draw[fill=green] (-#1/2,-2*\s)--(#1/2,-2*\s)--(#1/2,-3*\s)--(-#1/2,-3*\s)--cycle;
  \draw[thick] (-#1/2,-\s)--(-#1/2,-2*\s);
  \draw[thick] (#1/2,-\s)--(#1/2,-2*\s);
  \draw[thick] (-#1*1.5,-\s)--(-#1*1.5,-2*\s);
  \draw[thick] (#1*1.5,-\s)--(#1*1.5,-2*\s);
  \draw[thick] (#1,0)--(#1,\s);
    \draw[thick] (-#1,0)--(-#1,\s);
    \draw[thick] (#1/2,-3*\s)--(#1/2,-4*\s);
    \draw[thick] (-#1/2,-3*\s)--(-#1/2,-4*\s);
    \draw (-#1*1.5,-2*\s) node[below]{#2};
    \draw (#1*1.5,-2*\s) node[below]{#3};
    }
},
    pics/square/.style args={#1#2}{
        code={
            \draw[fill=yellow] (-0.3,0) rectangle (0.3,-0.6);
            \draw[thick] (0,0)--(0,0.5);
            \draw[thick] (0,-0.6)--(0,-1.1);
            \draw[thick] (-0.3,-0.3)--(-1.1,-0.3);
            \draw[thick] (0.3,-0.3)--(1.1,-0.3);
            \draw (-1.1,-0.3) node[left] {#1};
            \draw (1.1,-0.3) node[right] {#2};
        }
    }
}

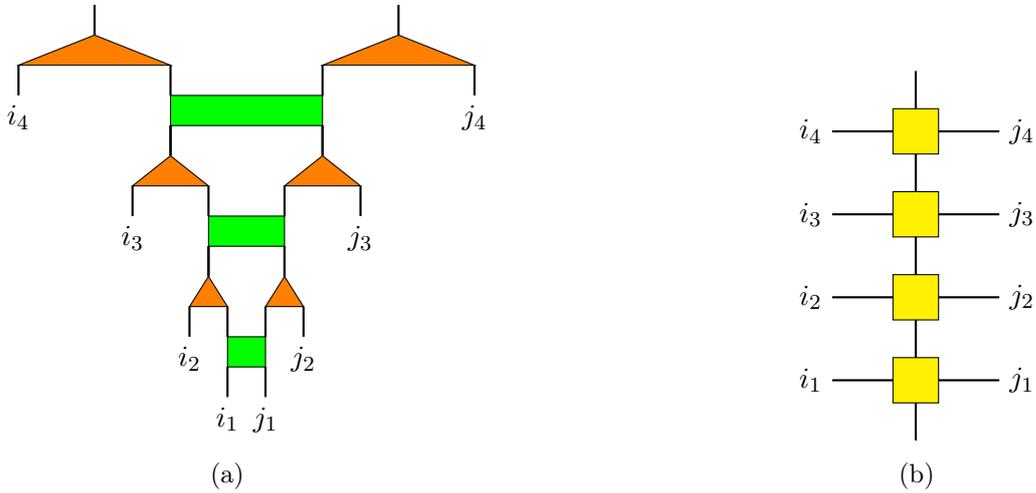
\begin{figure}
    \centering
    \begin{subfigure}[b]{0.4\textwidth}
        \begin{tikzpicture}
        \def\s{0.4}
        \pic at (0,0) {matrix={2}{$i_4$}{$j_4$}};
        \pic at (0,-4*\s) {matrix={1}{$i_3$}{$j_3$}};
        \pic at (0,-8*\s) {matrix={0.5}{$i_2$}{$j_2$}};
        \draw (-0.5/2,-12*\s) node[below]{$i_1$};
        \draw (0.5/2,-12*\s) node[below]{$j_1$};
    \end{tikzpicture}
    \caption{}
    \label{figD1a}
    \end{subfigure}
    \hfill
    \begin{subfigure}[b]{0.4\textwidth}
        \centering 
        \begin{tikzpicture}
          \pic at (0,0) {square={$i_1$}{$j_1$}};
          \pic at (0,1.1) {square={$i_2$}{$j_2$}};
          \pic at (0,2.2) {square={$i_3$}{$j_3$}};
          \pic at (0,3.3) {square={$i_4$}{$j_4$}};
        \end{tikzpicture}
        \caption{}
        \label{figD1b}
    \end{subfigure}
    \caption{The MERA tensor network with local unitaries removed produces a tensor network (a) that contributes to nonlocal magic. It can be written as an MPS (b) for which its stabilizer R\'enyi entropy can be computed numerically.}
    \label{fig:MPS}
\end{figure}

For the numerics, we pick a random realization of the disentangler and isometry and use them for each layer, in accordance of the scaling invariance. Then we present two estimations of the nonlocal magic of this MPS state. The first estimation we calculate the lower bound of the stabilizer relative entropy, given in \eqref{eq:relstab}. We present the result in~\cref{figD2a}. Both max entropy and the von Neumann entropy scale linearly with the number of matrices, and thus linearly with respect to the RT surface area and the entanglement entropy $S(A)$ of the boundary theory of subregion $A$. 
In the second estimation we calculate the entanglement spectrum of this state, denoting the set of eigenvalues as $\{\lambda_i\}$. Then we construct a state with the same entanglement spectrum and compute its nonlocal magic estimate using Eq.~\eqref{estimate_prop}. 

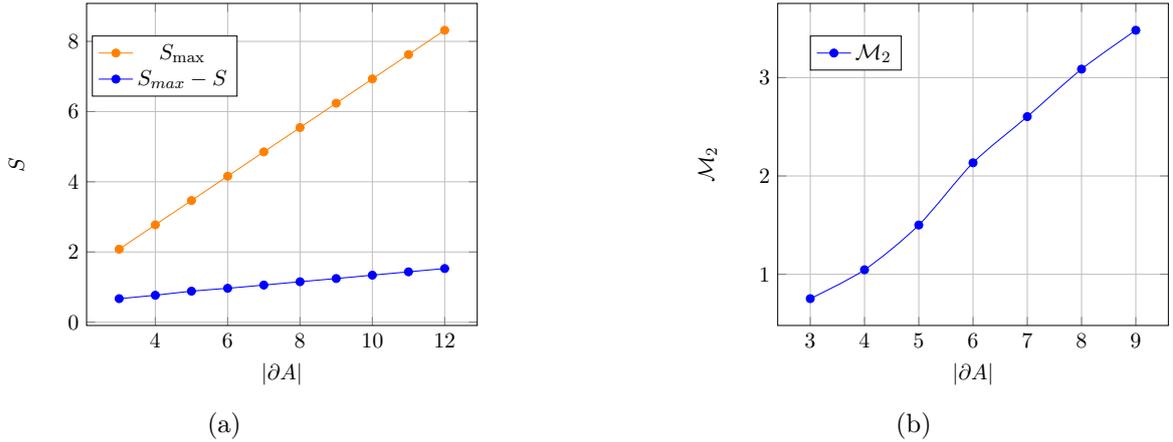
\begin{figure}[H]
    \centering
    \begin{subfigure}[b]{0.4\textwidth}
    \scalebox{0.75}{
    \begin{tikzpicture}
        \begin{axis}[
            xlabel={$|\partial A|$}, 
            ylabel={$S$}, 
            tick align=inside, 
            legend style={at={(0.2,0.9)},anchor=north}, 
            grid={both}
        ]
        \addplot[
            color=orange,
            mark=*, 
            smooth 
        ] coordinates {
            (3,2.079)
            (4,2.773)
            (5,3.466)
            (6,4.159)
            (7,4.852)
            (8,5.545)
            (9,6.238)
            (10,6.931)
            (11,7.624)
            (12,8.317)
        };
        \addlegendentry{\(S_{\text{max}}\)}
        \addplot[
            color=blue,
            mark=*, 
            smooth 
        ] coordinates {
            (3, 0.67)
             (4, 0.767)
             (5, 0.883)
             (6, 0.965)
             (7, 1.057)
             (8, 1.152)
             (9, 1.243)
             (10, 1.339)
             (11, 1.434)
             (12, 1.528)
        };
        \addlegendentry{\(S_{max}-S\)}
        \end{axis}
    \end{tikzpicture}
    }
    \caption{}
    \label{figD2a}
\end{subfigure}
\hfill
\begin{subfigure}[b]{0.4\textwidth}
    \centering
    \scalebox{0.75}{
    \begin{tikzpicture}
        \begin{axis}[
            xlabel={$|\partial A|$}, 
            ylabel={$\mathcal{M}_2$}, 
            tick align=inside, 
            legend style={at={(0.2,0.9)},anchor=north}, 
            grid={both}
        ]
        \addplot[
            color=blue,
            mark=*, 
            smooth 
        ] coordinates {
            (3,0.752299)
            (4,1.04593)
            (5,1.50105)
            (6,2.13376)
            (7,2.603420)
            (8,3.086418)
            (9,3.481793)
        };
        \addlegendentry{\(\mathcal{M}_{2}\)}
        \end{axis}
    \end{tikzpicture}
    }
    \caption{}
    \label{figD2b}
\end{subfigure}
\caption{(a) Maximal entropy and von Neumann entropy of the MPS as a function of the number of sites in the state. (b) stabilizer R\'enyi entropy of the state with small local magic.}
\end{figure}


The above intuition is also apparent when we think of the holographic QECC perspective of AdS/CFT where it is given by a code that corrects erasures approximately. In this case, complementary \cite{HarlowRT} approximate erasure correction promises the existence of recover unitaries supported on each subregion, such that

\begin{equation}
    U_A U_{A^c}|\tilde{\psi}\rangle_{AA^c} U_A^{\dagger} U_{A^c}^{\dagger} \approx |\psi\rangle |\chi\rangle,
\end{equation}
where $|\psi\rangle$ captures the bulk encoded information while $|\chi\rangle$ is the entanglement mediating erasure correction\footnote{We note that this heuristic argument is only expected to hold approximately in the leading order $N$ for holographic CFTs as that is when they function as approximate erasure correction codes with recovery errors suppressed by $1/N$. }. This is the case for certain of holographic QECC toy models, such as instances of approximate holographic Bacon-Shor codes\cite{ABSC} (see, e.g. Fig. 41 or generally when the skewing is small,) and \cite{Hayden_2016} when imperfectly entangled pairs are used in place of maximally entangled states when building the tensor network. The latter is known to be able to produce the correct single-interval CFT entanglement entropy but fails at the multi-interval level.

The states $|\chi\rangle$ 
now play the role of the interface tensor in MERA. To leading order, the R\'enyi entropies associated with $|\chi\rangle$ again scale as the area of the extremal surface where it is explicitly given by the number of entangled states across the bulk cut. Therefore, by (\cref{eqn:avgnlsre}), its magic as measured by the stabilizer R\'enyi entropy should scale as the area of the RT surface for any additive magic measure as one simply has to count the number of such approximate Bell pairs.
Strictly speaking, this again yields an upper bound as we do not optimize over all basis choices.

We also expect this linear dependence between nonlocal magic and entanglement to extend to noncritical states with translational invariance. For example, it is well-known that the truncated MERA (or more simply an MPS) can describe ground state of gapped phases where entanglement can increase slightly as the system grows, but plateaus at sufficiently large $|A|$\footnote{See \cite{Molina-Vilaplana:2012rmg} for example.}.  The nonlocal magic again resides on the edges connecting $A$ and $B$. From the bond counting argument, we again arrives at $\mathcal{M}^{NL}\sim S(A)$ provided $\mathcal{M}$ is an additive measure of magic.

Having established that the nonlocal magic should scale as the entropy, let us now examine how it should be connected with the antiflatness of the entanglement spectrum. Consider again the distilled state $|\chi\rangle$ which we represent as an MPS shared between $A$ and $B$ with local magic removed. Recall that since $\chi_A=\tr_B[|\chi\rangle\langle\chi|] = U_A\rho_A U_A^{\dagger}$, their entanglement spectrum and antiflatness are identical, i.e., $\mathcal{F}(\chi_A)=\mathcal{F}(\rho_A)$. 

For simplicity, let us approximate the MPS as entangled states $|\phi\rangle^{\otimes n}$ where each state $|\phi\rangle$ can be thought of as imperfect entangled pairs\footnote{For concreteness, one can think of them as imperfect Bell pairs. More generally, they do not have to be qubits, but a pair of qudits that are not maximally entangled. }. These states have volume law entanglement across $A$ and $B$.  We expect this to be a reasonable approximation because MPS with constant bond dimension limits the amount of correlation to be short-ranged, making them close to the tensor products which one can think of as a mean field approximation. We further support this claim with numerical evidence in Appendix~\ref{app:MPS}. 

It is known from Ref.~\cite{flatness} that for a typical state $|\phi\rangle_{ab}$ with stabilizer linear entropy $M_{\rm lin}(|\phi\rangle)$ chosen from its Clifford orbit $\{\Gamma_{ab}|\phi\rangle, \forall \Gamma\in \mathcal{C}_2\}$, the antiflatness of the entanglement spectrum of $\phi=\Tr_{b}[|\phi\rangle\langle\phi|]$ when cutting the state in half is given by
\begin{equation}
    \mathcal{F}(\phi) = c(d,d_a) M_{\rm lin}(\phi),
    \label{eqn:Cliffordtypical}
\end{equation}
where $c(d_a,d)=\frac{(d^2-d_a^2)(d_a^2-1)}{(d^2-1)(d+2)}$. Note that the second stabilizer R\'enyi entropy is related to the stabilizer linear entropy $\mathcal{M}_2=-\log (1-M_{\rm lin})$. Here $M_{\rm lin}\leq 1-2(d+1)^{-1}$ with $d$ being the dimension of the Hilbert space of $|\phi\rangle$ and $d_a=\sqrt{d}$ the Hilbert space dimension of subsystem $a$. Applying \eqref{eqn:Cliffordtypical} to each pair, we would have 
\begin{equation}
\begin{split}
    \mathcal{M}_2(|\phi\rangle) \approx &\frac{\mathcal{F}(\phi_a)}{c(d,d_a)}\\
    \approx &- \frac{\pur(\phi_a)^2}{c(d,d_a)}\left.\frac{\partial\tilde{S}_m(\phi_a)}{\partial m}\right\vert_{m=1},
\end{split}
\end{equation}
where we have applied the approximation \eqref{eqn:approxrenyid} to rewrite the right-hand side. in terms of additive antiflatness measure. 
Based on the assumption of distillation $|\chi\rangle\approx |\phi\rangle^{\otimes n}$ (see~\cref{app:MPS} for discussion) and additivity of $\mathcal{M}_2$ we conclude that
\begin{equation}\label{eqn:minsurfNLM}
\begin{split}
 \mathcal{M}_2(\ket{\chi})&\approx n\frac{\pur(\phi_a)^2}{c(d,d_a)}|\partial_m\tilde{S}_m(\phi_a)||_{m=1}\\
 &=\kappa |\partial_m\tilde{S}_m(\chi_A)||_{m=1}.
\end{split}
\end{equation}
where $\kappa = \pur(\phi_a)^2/c(d,d_a)$  is some coefficient that depends on the details of $|\phi\rangle$. Note that $\partial_m\tilde{S}_m$ is negative in our convention.

Since we argued that $\mathcal{M}_2(|\chi\rangle) \approx \mathcal{M}_2^{NL}(|\psi\rangle_{AB})$, 
\begin{equation}
    \mathcal{M}_2^{NL}(|\psi\rangle_{AB}) \approx\kappa |\partial_m\tilde{S}_m(\chi_A)||_{m=1}.
    \label{eqn:NLMflatness}
\end{equation}
for a CFT ground state. Therefore, if one uses the computable stabilizer R\'enyi entropy $\mathcal{M}_2$, we predict that the nonlocal magic scales linearly with both entanglement entropy and the additive antiflatness $\partial_m\tilde{S}_m(\chi_A)|_{m=1}$ of the entanglement spectrum across $A$ and $B$. 
In \cref{sec:numerics}, we numerically verify that this is indeed the case for an Ising CFT.

\begin{remark}
    Notably, our reasoning for the \emph{area-law} scaling of exact magic, i.e. $\mathcal{M}^{(NL)}\sim S(A)$, and the nonflatness relation (\ref{eqn:NLMflatness}) does not rely on any particular properties of the CFT. Indeed, as long as one can concentrate the magic from $A$ to $\partial A$ using the kind of unitary distillation procedure, this area law would also hold for gapped system with entanglement area law. The antiflatness relation is also similar to the area law scaling of entanglement spread\cite{harrow_2002_tight,Coudron_2019_entanglement,harrow_2010_entanglement,Anshu_2022_entanglement,Bennett_2014_reverse}.
\end{remark}

\subsection{Nonlocal magic in Ising model}\label{sec:numerics}
 
In this section we provide numerical computations to support our prior conjectures. We begin with the $1 + 1D$ transverse-field Ising model, with Hamiltonian given by
\ba\label{IsingHamiltonian}
H_{\rm Ising}=-\cos(\theta)\sum_i Z_i Z_{i+1}-\sin(\theta)\sum_i X_i.
\ea
We particularly consider \eqref{IsingHamiltonian} near its critical point, when $\theta = \pi/4$.

This model is described by an Ising CFT in the thermodynamic limit at criticality, that is when $\theta=\frac{\pi}{4}$. For our analysis, we perform exact diagonalization to determine the ground state of a 26-site spin chain with periodic boundary condition. Subsequently, the state is partitioned into two contiguous segments: $A$ and $\bar{A}$. To numerically estimate the nonlocal magic related to this bipartition, we use the Stabilizer Rényi Entropy measure $\mathcal{M}_2(\{\lambda_i\})$, as defined in~\cref{section:estimate}. Importantly, this measure relies solely on the entanglement spectrum, which we obtain through singular value decomposition.

At the critical point, we compute the  $\mathcal{M}_2(\{\lambda_i\})$ measure while progressively increasing the size of the subsystem $|A|$. The plot of $\mathcal{M}_2$ is present in \cref{fig:magic_CFTa}. 

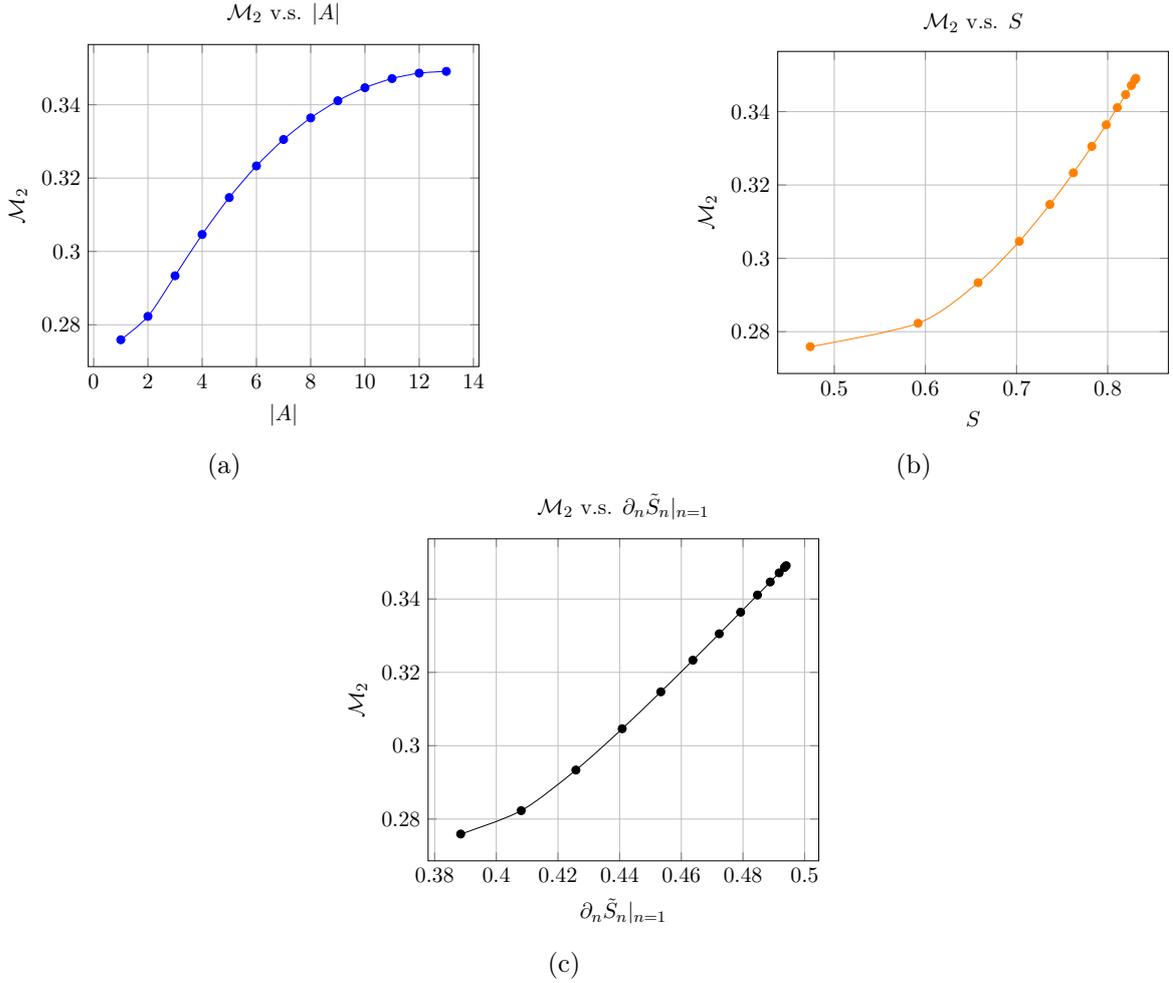
\begin{figure}
    \centering
    \begin{subfigure}[b]{0.4\textwidth}
    \scalebox{0.75}{
    \begin{tikzpicture}
        \begin{axis}[
            xlabel={$|A|$}, 
            ylabel={$\mathcal{M}_2$}, 
            tick align=inside, 
            title={$\mathcal{M}_2$ vs $|A|$},
            legend style={at={(0.2,0.9)},anchor=north}, 
            grid={both}
        ]
        \addplot[
            color=blue,
            mark=*, 
            smooth 
        ] coordinates {
            (1, 0.2759145827825882)
             (2, 0.28228738186896263)
             (3, 0.29336691658661357)
             (4, 0.30462444579473225)
             (5, 0.3146928627439401)
             (6, 0.32331311085491654)
             (7, 0.3305184956722919)
             (8, 0.3364093242384839)
             (9, 0.34108791009708955)
             (10, 0.3446403128445661)
             (11, 0.3471323166256153)
             (12, 0.34860965777546515)
             (13, 0.3490991613456804)
        };
        \end{axis}
    \end{tikzpicture}
    }
    \caption{}
   \label{fig:magic_CFTa}
\end{subfigure}
\hfill
\begin{subfigure}[b]{0.4\textwidth}
    \centering
    \scalebox{0.75}{
    \begin{tikzpicture}
        \begin{axis}[
            xlabel={$S$}, 
            ylabel={$\mathcal{M}_2$}, 
            tick align=inside, 
            title={$\mathcal{M}_2$ vs $S$ },
            legend style={at={(0.2,0.9)},anchor=north}, 
            grid={both}
        ]
        \addplot[
            color=orange,
            mark=*, 
            smooth 
        ] coordinates {
            (0.47365496896115133, 0.2759145827825882)
             (0.5918813439084066, 0.28228738186896263)
             (0.6578352071417765, 0.29336691658661357)
             (0.7030228571069249, 0.30462444579473225)
             (0.7365380200480853, 0.3146928627439401)
             (0.7623514890669422, 0.32331311085491654)
             (0.782554333578475, 0.3305184956722919)
             (0.7983723766168666, 0.3364093242384839)
             (0.8105729284602428, 0.34108791009708955)
             (0.8196538174273882, 0.3446403128445661)
             (0.8259401351911264, 0.3471323166256153)
             (0.8296366915149734, 0.34860965777546515)
             (0.8308567814379734, 0.3490991613456804)
        };
        \end{axis}
    \end{tikzpicture}
    }
    \caption{}
    \label{fig:magic_CFTb}
\end{subfigure}
\hfill
\begin{subfigure}[b]{0.4\textwidth}
    \centering
    \scalebox{0.75}{
    \begin{tikzpicture}
        \begin{axis}[
            xlabel={$\partial_n\tilde{S}_n|_{n=1}$}, 
            ylabel={$\mathcal{M}_2$}, 
            tick align=inside, 
            title={$\mathcal{M}_2$ vs $\partial_n\tilde{S}_n|_{n=1}$ },
            legend style={at={(0.2,0.9)},anchor=north}, 
            grid={both}
        ]
        \addplot[
            color=black,
            mark=*, 
            smooth 
        ] coordinates {
            (0.38844044873524697, 0.2759145827825882)
             (0.4080570661929262, 0.28228738186896263)
             (0.4258068254534285, 0.29336691658661357)
             (0.44083017326285073, 0.30462444579473225)
             (0.4533650556446711, 0.3146928627439401)
             (0.46375639975735045, 0.32331311085491654)
             (0.4722992639562782, 0.3305184956722919)
             (0.47922104982083263, 0.3364093242384839)
             (0.4846911425107271, 0.34108791009708955)
             (0.4888330874178818, 0.3446403128445661)
             (0.4917343044615911, 0.3471323166256153)
             (0.493452921810052, 0.34860965777546515)
             (0.49402218734097664, 0.3490991613456804)
        };
        \end{axis}
    \end{tikzpicture}
    }
    \caption{}
    \label{fig:magic_CFTc}
\end{subfigure}
\label{fig:magic_CFT}
\caption{(a) Plot of nonlocal stabilizer R\'enyi entropy $\mathcal{M}_2$ vs subsystem size |A|;  (b) Plot of  $\mathcal{M}_2$ vs Entropy $S$. (c) Plot of  $\mathcal{M}_2$ vs the antiflatness based on entanglement capacity. Model is at critical point, with 26 lattice sites.}
\end{figure}

In \cref{fig:magic_CFTb}, we observe that the nonlocal magic scales similarly to entropy when we increase the size of the subregion $|A|$, particularly beyond 3 qubits. This indicates that the nonlocal magic in the CFT scales logarithmically with $|A|$, in agreement with our analysis presented in the MERA framework in~\cref{section:MERA}. Additionally, \cref{fig:magic_CFTc} demonstrates the proportional relationship between nonlocal magic and antiflatness, supporting the estimation in   \eqref{eqn:NLMflatness}. 

A similar analysis is applied to study the model away from the critical point, as illustrated in~\cref{fig:offcritical}.  We define the parameter $g=\theta-\frac{\pi}{4}$, where $g$ quantifies the deviation from criticality. In this regime, we observe that the nonlocal magic reaches a plateau at a certain point, mirroring the behavior observed in entropy. 

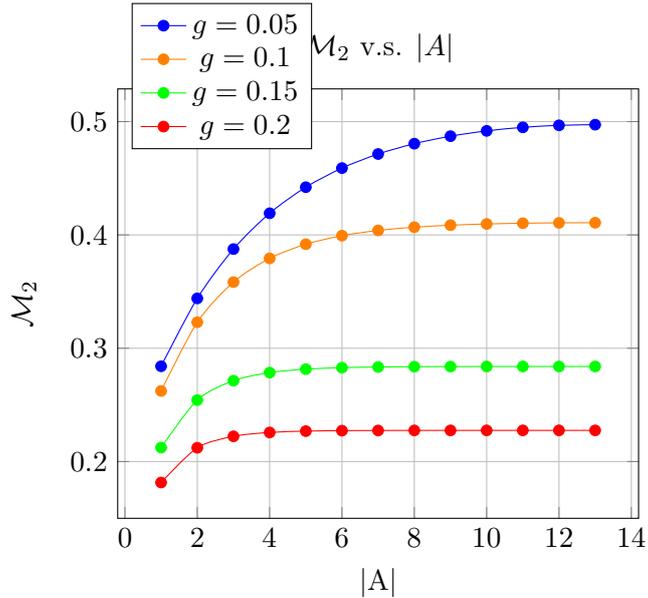
\begin{figure}
    \centering
    \begin{tikzpicture}
        \begin{axis}[
            xlabel={|A|}, 
            ylabel={$\mathcal{M}_2$}, 
            tick align=inside, 
            title={$\mathcal{M}_2$ vs $|A|$},
            legend style={at={(0.2,1.2)},anchor=north}, 
            grid={both}
        ]
        \addplot[
            color=blue,
            mark=*, 
            smooth 
        ] coordinates {
            (1, 0.28407617773515886)
              (2, 0.34397774537887427)
              (3, 0.38751332633761937)
              (4, 0.4190970491310785)
              (5, 0.4421015965056568)
              (6, 0.45896702985427884)
              (7, 0.4713795896484691)
              (8, 0.4804991524617711)
              (9, 0.48712558430690295)
              (10, 0.4918089362542348)
              (11, 0.4949206945520738)
              (12, 0.49669922942297173)
              (13, 0.49727781169858104)
        };
        \addlegendentry{\(g=0.05\)}
        \addplot[
            color=orange,
            mark=*, 
            smooth 
        ] coordinates {
            (1, 0.2622746594000573)
              (2, 0.32291548379551155)
              (3, 0.35835308837480184)
              (4, 0.37928016995541014)
              (5, 0.3917771326269049)
              (6, 0.3993340830642571)
              (7, 0.4039561283143258)
              (8, 0.4068079872130534)
              (9, 0.4085740871637637)
              (10, 0.40965959874648833)
              (11, 0.41030367773714116)
              (12, 0.41064345911210953)
              (13, 0.41074948178076964)
        };
        \addlegendentry{\(g=0.1\)}
        \addplot[
            color=green,
            mark=*, 
            smooth 
        ] coordinates {
            (1, 0.21230120003858938)
              (2, 0.2542448918104177)
              (3, 0.27134898072257646)
              (4, 0.2784874066332622)
              (5, 0.28152527488291956)
              (6, 0.2828413759502512)
              (7, 0.28342048270374093)
              (8, 0.28367879657447115)
              (9, 0.2837954188472098)
              (10, 0.2838485595687859)
              (11, 0.28387271597683295)
              (12, 0.2838830717667294)
              (13, 0.28388594897480607)
        };
        \addlegendentry{\(g=0.15\)}
        \addplot[
            color=red,
            mark=*, 
            smooth 
        ] coordinates {
            (1, 0.18148373123245282)
              (2, 0.21215783087646567)
              (3, 0.22223636904171473)
              (4, 0.22565129476412513)
              (5, 0.2268364651797759)
              (6, 0.22725618915490875)
              (7, 0.22740733176914482)
              (8, 0.22746252974583842)
              (9, 0.22748294302369282)
              (10, 0.22749058018764026)
              (11, 0.22749345708531876)
              (12, 0.2274945078106672)
              (13, 0.2274947725650576)
        };
        \addlegendentry{\(g=0.2\)}
        \end{axis}
    \end{tikzpicture}
    \caption{ Plot of nonlocal stabilizer R\'enyi entropy $\mathcal{M}_2$ vs subsystem size |A|. The model parameter $g=\theta-\frac{\pi}{4}$ is adjusted to position the model away from its critical point.}
    \label{fig:offcritical}
\end{figure}

In our final analysis, we keep the size of the subregion $|A|$ constant and track the changes in nonlocal magic as the model approaches and passes through the critical point. As depicted in \cref{fig:transition}, a distinct peak in nonlocal magic is observed. Notably, this peak shifts closer to the critical point ($g=0$) and becomes increasingly sharp as the total system size ($n$) is enlarged. These observations suggest the potential presence of a phase transition in the nonlocal magic measure. \cref{fig:comparison} presents a comparison of nonlocal magic and antiflatness against the model parameter $g$, revealing a consistent trend as $g$ changes. 

\begin{figure}[H]
    \centering
    \begin{tikzpicture}
        \begin{axis}[
            xlabel={$g$}, 
            ylabel={$\mathcal{M}_2$}, 
            title={$\mathcal{M}_2$ vs $g$ },
            tick align=inside, 
            legend style={at={(0.2,1.2)},anchor=north}, 
            grid={both}
        ]
        \addplot[
            color=blue,
            mark=*, 
            smooth 
        ] coordinates {
            (-0.2324768366025517, 0.010179027166607601)
             (-0.2174768366025518, 0.012700677518416869)
             (-0.20247683660255178, 0.015914929866794172)
             (-0.18747683660255177, 0.020024212927263287)
             (-0.17247683660255175, 0.02528791586737733)
             (-0.15747683660255174, 0.032034550069390655)
             (-0.14247683660255173, 0.04067294569286466)
             (-0.1274768366025517, 0.05169871325524065)
             (-0.11247683660255181, 0.06568911429733332)
             (-0.09747683660255169, 0.08327519508063366)
             (-0.08247683660255178, 0.10507537532282413)
             (-0.06747683660255177, 0.13157250926463754)
             (-0.05247683660255176, 0.16292287590693638)
             (-0.037476836602551744, 0.19870953869174304)
             (-0.02247683660255173, 0.2376998361713844)
             (-0.0074768366025518285, 0.27772582130124907)
             (0.007523163397448296, 0.31583345964325205)
             (0.022523163397448198, 0.348777714878471)
             (0.03752316339744832, 0.37375895572844275)
             (0.052523163397448225, 0.3891051689553718)
             (0.06752316339744824, 0.3945830900350399)
             (0.08252316339744825, 0.39122799220737287)
             (0.09752316339744815, 0.38084991622703784)
             (0.11252316339744828, 0.36548868544182306)
             (0.12752316339744818, 0.3470148330281814)
             (0.1425231633974483, 0.32693104318822336)
             (0.1575231633974482, 0.3063316700703259)
             (0.17252316339744822, 0.2859505784494)
             (0.18752316339744823, 0.26624121060662875)
             (0.20252316339744825, 0.247456486000428)
             (0.21752316339744826, 0.22971483114393784)
        };
        \addlegendentry{\(n=10\)}
        \addplot[
            color=orange,
            mark=*, 
            smooth 
        ] coordinates {
            (-0.2324768366025517, 0.008465978246463156)
             (-0.2174768366025518, 0.010407193006500729)
             (-0.20247683660255178, 0.012875327118831839)
             (-0.18747683660255177, 0.016039921641527576)
             (-0.17247683660255175, 0.0201294554803197)
             (-0.15747683660255174, 0.025450769720905188)
             (-0.14247683660255173, 0.03241275236797553)
             (-0.1274768366025517, 0.04155263056466259)
             (-0.11247683660255181, 0.053559580056458625)
             (-0.09747683660255169, 0.06928326208862791)
             (-0.08247683660255178, 0.08970284247116354)
             (-0.06747683660255177, 0.11581559301769835)
             (-0.05247683660255176, 0.14839133812423738)
             (-0.037476836602551744, 0.1875538019511966)
             (-0.02247683660255173, 0.23223365361545356)
             (-0.0074768366025518285, 0.27971406882100586)
             (0.007523163397448296, 0.3256706714548758)
             (0.022523163397448198, 0.3650407584940913)
             (0.03752316339744832, 0.39354680139542364)
             (0.052523163397448225, 0.4090658305660494)
             (0.06752316339744824, 0.41201429206413853)
             (0.08252316339744825, 0.40466397043774915)
             (0.09752316339744815, 0.39001853899237493)
             (0.11252316339744828, 0.37090936249006595)
             (0.12752316339744818, 0.3495543491942982)
             (0.1425231633974483, 0.32748593479147553)
             (0.1575231633974482, 0.30566298236122613)
             (0.17252316339744822, 0.2846304790946891)
             (0.18752316339744823, 0.2646616263994785)
             (0.20252316339744825, 0.24586326932071415)
             (0.21752316339744826, 0.22824671738457047)
        };
        \addlegendentry{\(n=12\)}
        \addplot[
            color=green,
            mark=*, 
            smooth 
        ] coordinates {
            (-0.2324768366025517, 0.007933413987301907)
             (-0.2174768366025518, 0.009643632352382655)
             (-0.20247683660255178, 0.011791929928751497)
             (-0.18747683660255177, 0.014519609408522049)
             (-0.17247683660255175, 0.018021681370604066)
             (-0.15747683660255174, 0.02256839312394047)
             (-0.14247683660255173, 0.028534509077290136)
             (-0.1274768366025517, 0.03643744218320802)
             (-0.11247683660255181, 0.04698286638780856)
             (-0.09747683660255169, 0.06110965787415222)
             (-0.08247683660255178, 0.08001058854894082)
             (-0.06747683660255177, 0.10507569246146076)
             (-0.05247683660255176, 0.13766195830438846)
             (-0.037476836602551744, 0.17856395871997527)
             (-0.02247683660255173, 0.22713655856876372)
             (-0.0074768366025518285, 0.28033989294322587)
             (0.007523163397448296, 0.33248637767314865)
             (0.022523163397448198, 0.37655598026422266)
             (0.03752316339744832, 0.40684951518439183)
             (0.052523163397448225, 0.4212065461968845)
             (0.06752316339744824, 0.4211640128304521)
             (0.08252316339744825, 0.41036777854896495)
             (0.09752316339744815, 0.39276294067321815)
             (0.11252316339744828, 0.37155015442930295)
             (0.12752316339744818, 0.34892560405388956)
             (0.1425231633974483, 0.3262306498393502)
             (0.1575231633974482, 0.3042075712158339)
             (0.17252316339744822, 0.283222806200963)
             (0.18752316339744823, 0.2634239992990375)
             (0.20252316339744825, 0.24483906960710441)
             (0.21752316339744826, 0.22743469159561436)
        };
        \addlegendentry{\(n=14\)}
        \addplot[
            color=red,
            mark=*, 
            smooth 
        ] coordinates {
            (-0.2324768366025517, 0.007637717074741366)
             (-0.2174768366025518, 0.009174338805064993)
             (-0.20247683660255178, 0.011052959394144994)
             (-0.18747683660255177, 0.013365130617364369)
             (-0.17247683660255175, 0.01623280968600742)
             (-0.15747683660255174, 0.01982129786435423)
             (-0.14247683660255173, 0.024359537181350557)
             (-0.1274768366025517, 0.03017316361630484)
             (-0.11247683660255181, 0.037740180797763534)
             (-0.09747683660255169, 0.047786889149565595)
             (-0.08247683660255178, 0.061453076176694446)
             (-0.06747683660255177, 0.08056229243796553)
             (-0.05247683660255176, 0.1079877179114993)
             (-0.037476836602551744, 0.14782590237321785)
             (-0.02247683660255173, 0.20418188817273544)
             (-0.0074768366025518285, 0.2762491075509873)
             (0.007523163397448296, 0.351294063986569)
             (0.022523163397448198, 0.4087064418576)
             (0.03752316339744832, 0.437567104391568)
             (0.052523163397448225, 0.4421015965056568)
             (0.06752316339744824, 0.4316509592969976)
             (0.08252316339744825, 0.41346368643960796)
             (0.09752316339744815, 0.3917771326269049)
             (0.11252316339744828, 0.36885259515848734)
             (0.12752316339744818, 0.3458776892433747)
             (0.1425231633974483, 0.32347725558175994)
             (0.1575231633974482, 0.3019745954831864)
             (0.17252316339744822, 0.28152527488291956)
             (0.18752316339744823, 0.262188711010384)
             (0.20252316339744825, 0.24396854935148515)
             (0.21752316339744826, 0.2268364651797759)
        };
        \addlegendentry{\(n=26\)}
        \end{axis}
    \end{tikzpicture}
    \caption{ Plot of nonlocal stabilizer R\'enyi entropy $\mathcal{M}_2$ vs parameter $g=\theta-\frac{\pi}{4}$, at $|A|=5$ with increasing total spins $n$. }
    \label{fig:transition}
\end{figure}
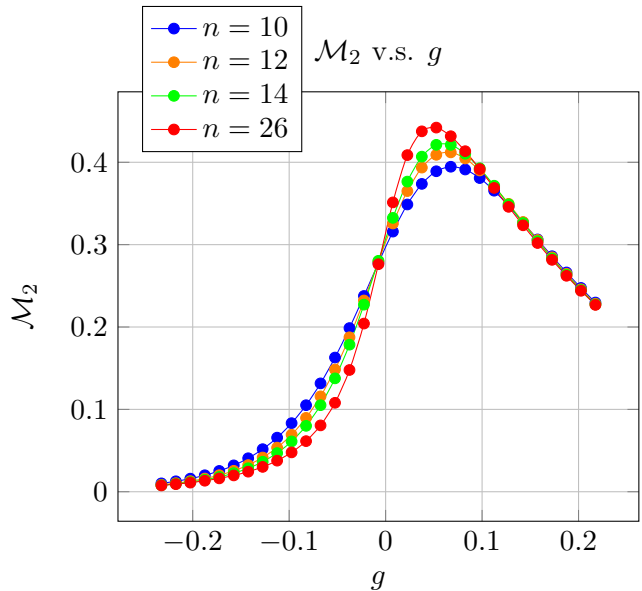

\begin{figure}[H]
    \centering
    \begin{subfigure}[b]{0.4\textwidth}
        \scalebox{0.7}{
        \begin{tikzpicture}
            \begin{axis}[
            xlabel={$g$}, 
            ylabel={$\mathcal{M}_2$}, 
            tick align=inside, 
            legend style={at={(0.2,1.2)},anchor=north}, 
            grid={both}
        ]
        \addplot[
            color=blue,
            mark=*, 
            smooth 
        ] coordinates {
            (-0.2324768366025517, 0.007145304321954218)
             (-0.2174768366025518, 0.008487960605390173)
             (-0.20247683660255178, 0.010101485878992943)
             (-0.18747683660255177, 0.012053517023108064)
             (-0.17247683660255175, 0.014435402927989302)
             (-0.15747683660255174, 0.017374563446235013)
             (-0.14247683660255173, 0.02105582840800511)
             (-0.1274768366025517, 0.02575961508929479)
             (-0.11247683660255181, 0.03193201207252851)
             (-0.09747683660255169, 0.04031495948275232)
             (-0.08247683660255178, 0.05218588101359276)
             (-0.06747683660255177, 0.06977904131169534)
             (-0.05247683660255176, 0.09693602789547702)
             (-0.037476836602551744, 0.13974677692867024)
             (-0.02247683660255173, 0.20574693233055968)
             (-0.0074768366025518285, 0.29755202131258474)
             (0.007523163397448296, 0.3990203245055871)
             (0.022523163397448198, 0.47501942013468546)
             (0.03752316339744832, 0.5045213886126098)
             (0.052523163397448225, 0.49727781169858104)
             (0.06752316339744824, 0.4722337165709027)
             (0.08252316339744825, 0.44161526289047964)
             (0.09752316339744815, 0.41074948178076964)
             (0.11252316339744828, 0.3814622520861778)
             (0.12752316339744818, 0.354203227070393)
             (0.1425231633974483, 0.3289563794953849)
             (0.1575231633974482, 0.30557374117264996)
             (0.17252316339744822, 0.28388594897480607)
             (0.18752316339744823, 0.2637345106763872)
             (0.20252316339744825, 0.24497865343083947)
             (0.21752316339744826, 0.2274947725650576)
        };
        \end{axis}
        \end{tikzpicture}
        }
        \caption{}
    \end{subfigure}
    \hfill
    \begin{subfigure}[b]{0.4\textwidth}
       \scalebox{0.7}{
        \begin{tikzpicture}
            \begin{axis}[
            xlabel={$g$}, 
            ylabel={$\frac{\mathcal{F}(\rho)}{\mathrm{Pur}(\rho)^2}$}, 
            tick align=inside, 
            legend style={at={(0.2,1.2)},anchor=north}, 
            grid={both}
        ]
        \addplot[
            color=green,
            mark=*, 
            smooth 
        ] coordinates {
            (-0.2324768366025517, 0.0017906550155215046)
             (-0.2174768366025518, 0.002128430358248331)
             (-0.20247683660255178, 0.002535271954815998)
             (-0.18747683660255177, 0.0030292921303933847)
             (-0.17247683660255175, 0.003635877087823563)
             (-0.15747683660255174, 0.004392447420408408)
             (-0.14247683660255173, 0.005357636636801908)
             (-0.1274768366025517, 0.006629761987891159)
             (-0.11247683660255181, 0.008385239763439722)
             (-0.09747683660255169, 0.010960203971579534)
             (-0.08247683660255178, 0.015024513163008176)
             (-0.06747683660255177, 0.021939845122460465)
             (-0.05247683660255176, 0.0343972961772237)
             (-0.037476836602551744, 0.05700884023110594)
             (-0.02247683660255173, 0.09447793318209577)
             (-0.0074768366025518285, 0.1426692425225171)
             (0.007523163397448296, 0.18137932190436926)
             (0.022523163397448198, 0.1939601185649182)
             (0.03752316339744832, 0.18574544217475258)
             (0.052523163397448225, 0.16953391029563253)
             (0.06752316339744824, 0.15267413044338285)
             (0.08252316339744825, 0.1375274912018791)
             (0.09752316339744815, 0.12439700438451141)
             (0.11252316339744828, 0.11302272776847992)
             (0.12752316339744818, 0.10307817688079396)
             (0.1425231633974483, 0.09429054631087523)
             (0.1575231633974482, 0.08645135784048381)
             (0.17252316339744822, 0.0794032334518217)
             (0.18752316339744823, 0.07302573383562753)
             (0.20252316339744825, 0.06722470772402801)
             (0.21752316339744826, 0.061924963221174356)
        };
        \end{axis}
        \end{tikzpicture}
        }
        \caption{}
    \end{subfigure}
     \vskip\baselineskip
     \begin{subfigure}[b]{0.4\textwidth}
        \scalebox{0.7}{
         \begin{tikzpicture}
            \begin{axis}[
            xlabel={$g$}, 
            ylabel={$|\partial_n \tilde{S}_n||_{n=1}$}, 
            tick align=inside, 
            legend style={at={(0.2,1.2)},anchor=north}, 
            grid={both}
        ]
        \addplot[
            color=black,
            mark=*, 
            smooth 
        ] coordinates {
            (-0.2324768366025517, 0.005280910267068737)
             (-0.2174768366025518, 0.006510627625492099)
             (-0.20247683660255178, 0.00803596046963862)
             (-0.18747683660255177, 0.009939502333841627)
             (-0.17247683660255175, 0.012334930862737499)
             (-0.15747683660255174, 0.015384434278946342)
             (-0.14247683660255173, 0.019329900398019965)
             (-0.1274768366025517, 0.02455095011158324)
             (-0.11247683660255181, 0.031675969945498386)
             (-0.09747683660255169, 0.04179745283457983)
             (-0.08247683660255178, 0.056886945731048945)
             (-0.06747683660255177, 0.08056106196909345)
             (-0.05247683660255176, 0.11931511666420035)
             (-0.037476836602551744, 0.18365927393584316)
             (-0.02247683660255173, 0.2854831656147101)
             (-0.0074768366025518285, 0.4227007543631155)
             (0.007523163397448296, 0.5574088985771578)
             (0.022523163397448198, 0.6380384975874119)
             (0.03752316339744832, 0.6538151915333436)
             (0.052523163397448225, 0.629612387418702)
             (0.06752316339744824, 0.5898810310801145)
             (0.08252316339744825, 0.5471976045328252)
             (0.09752316339744815, 0.5062319411694376)
             (0.11252316339744828, 0.4682415484170781)
             (0.12752316339744818, 0.43330194413223755)
             (0.1425231633974483, 0.4011593535907676)
             (0.1575231633974482, 0.37151086480571227)
             (0.17252316339744822, 0.3440817036848223)
             (0.18752316339744823, 0.3186383619697557)
             (0.20252316339744825, 0.29498445375401916)
             (0.21752316339744826, 0.27295368302457707)
        };
        \end{axis}
         \end{tikzpicture}
    }
         \caption{}
     \end{subfigure}
     \hfill
     \begin{subfigure}[b]{0.4\textwidth}
        \scalebox{0.7}{
         \begin{tikzpicture}
         \begin{axis}[
            xlabel={$g$}, 
            ylabel={$S$}, 
            tick align=inside, 
            legend style={at={(0.2,1.2)},anchor=north}, 
            grid={both}
        ]
        \addplot[
            color=red,
            mark=*, 
            smooth 
        ] coordinates {
            (-0.2324768366025517, 0.7075195993413678)
             (-0.2174768366025518, 0.7098622966766157)
             (-0.20247683660255178, 0.7126093020806886)
             (-0.18747683660255177, 0.7158476059617306)
             (-0.17247683660255175, 0.7196900786617996)
             (-0.15747683660255174, 0.7242859044010616)
             (-0.14247683660255173, 0.7298358887171571)
             (-0.1274768366025517, 0.7366142225423622)
             (-0.11247683660255181, 0.7449965966544939)
             (-0.09747683660255169, 0.7554868402188267)
             (-0.08247683660255178, 0.7687086181565803)
             (-0.06747683660255177, 0.7852568413515487)
             (-0.05247683660255176, 0.8051400532946233)
             (-0.037476836602551744, 0.8263312358056276)
             (-0.02247683660255173, 0.8423208036318784)
             (-0.0074768366025518285, 0.8413255742297105)
             (0.007523163397448296, 0.8131684153975448)
             (0.022523163397448198, 0.7604474831216399)
             (0.03752316339744832, 0.6968134434919527)
             (0.052523163397448225, 0.634554879455417)
             (0.06752316339744824, 0.5791007799488848)
             (0.08252316339744825, 0.5311801664511724)
             (0.09752316339744815, 0.48976840976322106)
             (0.11252316339744828, 0.45356850203610183)
             (0.12752316339744818, 0.4214828918603129)
             (0.1425231633974483, 0.39268085746758397)
             (0.1575231633974482, 0.366551789226895)
             (0.17252316339744822, 0.34264488394343834)
             (0.18752316339744823, 0.32062101358587114)
             (0.20252316339744825, 0.30021898829499505)
             (0.21752316339744826, 0.2812328944584596)
        };
        \end{axis}
         \end{tikzpicture}
         }
         \caption{}
     \end{subfigure}
    \caption{(a) Plot of nonlocal magic $\mathcal{M}_2$ vs $g$. (b) Plot of antiflatness vs $g$. (c) Plot of antiflatness $|{\partial_n\tilde{S}_n}|$ vs $g$. (d) Plot of entropy $S$ vs $g$. All of these plots are based on data for fixed subregion size $|A|=13$. } 
    \label{fig:comparison}
\end{figure}
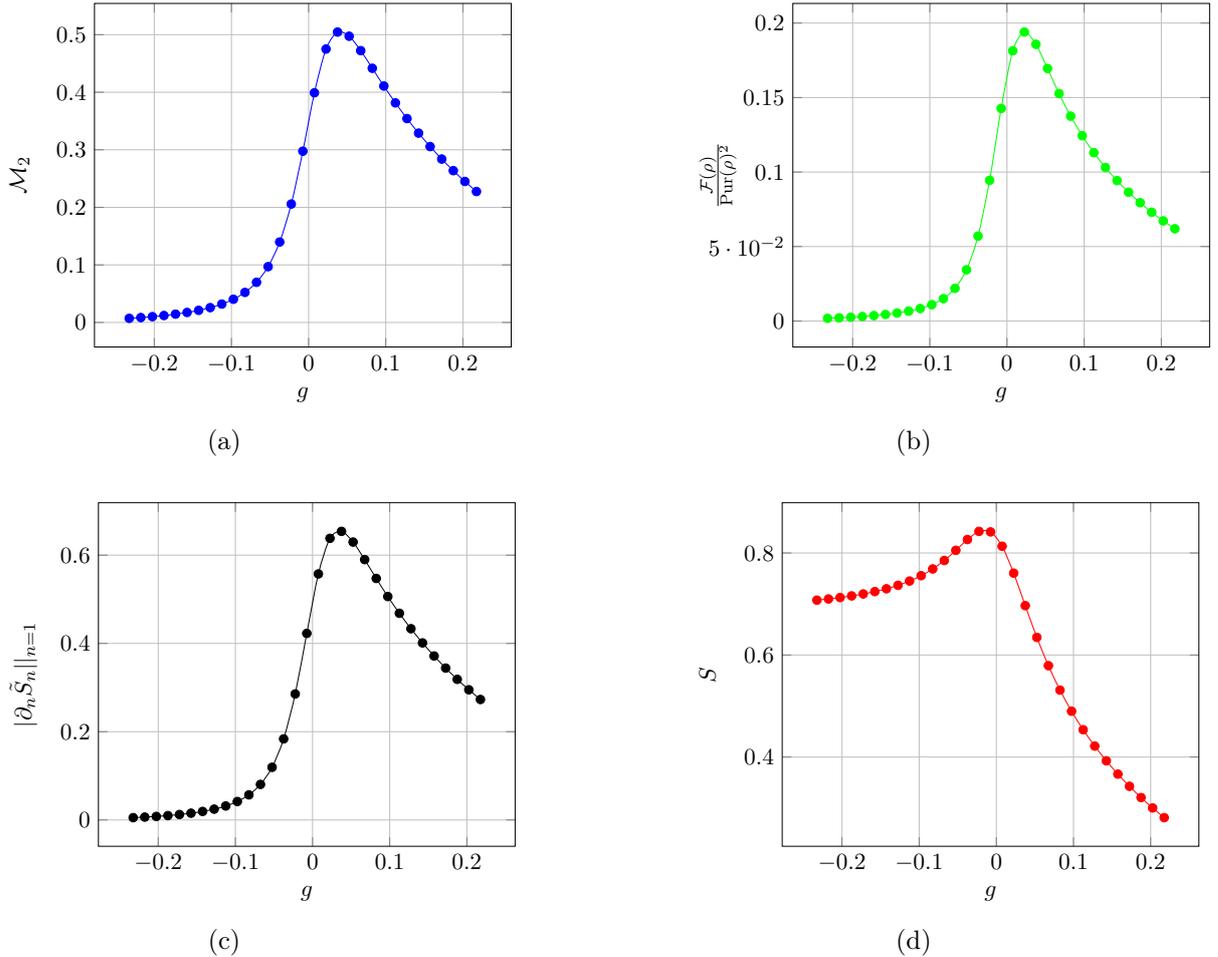

However, it is important to point out that nonlocal magic is not simply the entanglement entropy despite their similarity in this example. For instance, the ratio between nonlocal magic and entanglement depends on $g$. \cref{fig:NonLocalMagicSurface} gives a complete picture of $\mathcal{M}_2/S$ for a $14$-qubit Ising chain, as we vary both the parameter $g$ and the subsystem cardinality $|A|$. We observe that $\mathcal{M}_2/S$ maximizes for angles slightly above the critical point ($g = 0$) due to finite-size effect, in agreement with \cref{fig:transition,fig:comparison}. 
\begin{figure}[H]   
\begin{center}
		\begin{overpic}[width=9cm]{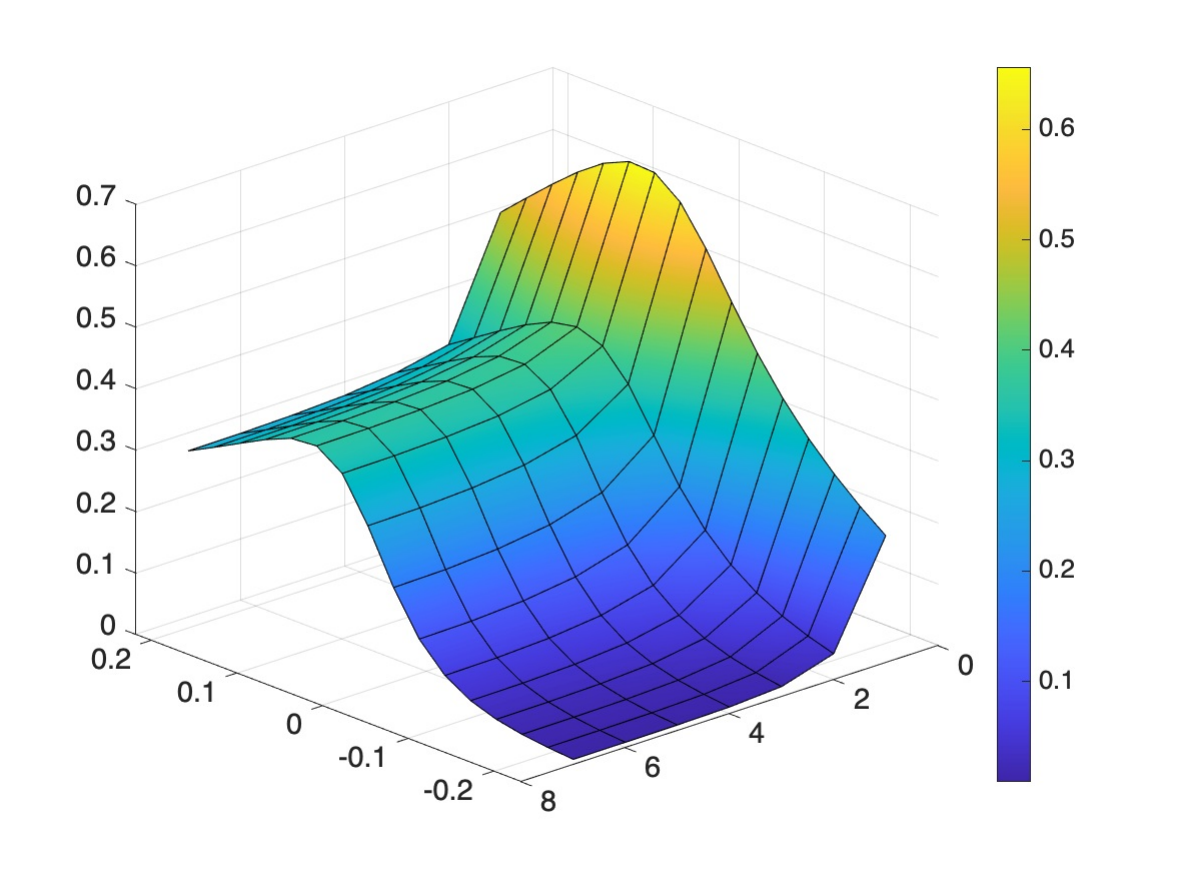}
		\put (-1,40) {$\frac{\mathcal{M}_2}{S}$}
		\put (20,10) {$g$}
		\put (66,7) {$|A|$}
        \end{overpic}
        \caption{Surface illustrating the ratio of nonlocal magic $\mathcal{M}_2$ to entanglement entropy $S$ in $n=14$ Ising CFT. We plot $\mathcal{M}_2/S$ as a function of parameter $g = \theta - \pi/4$ and subsystem size $|A|$. The value $\mathcal{M}_2/S$ reaches a maximum just above criticality ($g = 0$), before decreasing and ultimately plateauing.}
        \label{fig:NonLocalMagicSurface}
\end{center}
\end{figure}

The plateau in  \cref{fig:NonLocalMagicSurface} suggests a linear scaling between $M_2$ and $S$, as subsystem $|A|$ grows large. As we see that the linear behavior is already apparent at $n=14$. Recall from the tensor network picture, the linear scaling between nonlocal magic and entanglement entropy is expected, however, the density of nonlocal magic can vary depending on the shape of the spectrum. This is reflected in the figure as the asymptotic proportionality constant between $M_2$ and $S$ depends on $\theta$. 

Another instance where nonlocal magic distinguishes itself from entanglement can be found in the context of symmetry breaking. For $g<0$, the Ising model enters the symmetry-breaking phase in the thermodynamic limit where the nonlocal magic further displays a transition. We refer interested readers to~\cref{app:symmbreak} for details of this discussion.

\subsection{Smoothed Magic from Entropic Bounds}
Beyond tensor network and finite-size numerics, we recognize that many of the entropic quantities we have examined so far are generally infinite in conformal field theories and need regularization. It makes more sense to look at smoothed magic, which can be bounded by smoothed max-entropies. On the one hand, it generally leads to finite quantities. On the other hand, for any reasonable simulation of a CFT, it is far more relevant to produce approximations of a target state up to a small precision parameter $\epsilon$ instead of the exact state defined by the theory.

In Ref.~\cite{Bao:2019}, it was shown that under the assumption that the R\'enyi entropies satisfy $S_{n} = \frac{s_n}{G_N}$, the smoothed maximal entropy is directly proportional to the following expression:

\begin{equation}
    S_{max}^{\epsilon}=S+\sqrt{\log{\frac{1}{\epsilon}}S}+O(c^0), 
\end{equation}
where $S$ denotes the von Neumann entropy of the state and the central charge $c$ is taken to be very large, $c\rightarrow \infty.$ A similar expression is obtained by Ref.~\cite{maxminentropy} using the explicit spectrum for a 1+1D CFT by Calabrese and Lefevre~\cite{CFTspec}. This entropy is proportional to the central charge $c$ of holographic CFT, which is assumed to be  large. The leading-order correction to this expression is at $O(1)$, making it negligible relative to the primary term. 

With this in mind, we can estimate the lower bound for magic as follows:
\begin{equation}
    M^{(NL,\epsilon)}_{RS}(\rho_{AA^c})\geq S_{max}^{\epsilon}(A)-S(A)=\sqrt{S(A)\log\frac{1}{\epsilon}}+O(\epsilon c). 
\end{equation}

We assume the parameter $\epsilon$ to fall within the range $e^{-c}\ll\epsilon\ll c^{-1}$. 

Recall that for a given bipartition $A$ and $A^c$ in a holographic CFT, the von Neumann entropy of subregion $A$ to leading order is equal to the area $\mathcal{A}$ of the extremal surface anchored to the entangling boundary $\partial A$ divided by $4G_N$ according to the Ryu-Takayanagi formula\cite{Ryu_2006}. Thus, we can formally represent the lower bound of nonlocal magic as:
\begin{equation}
    M^{(NL,\epsilon)}_{RS}(\rho_{AA^c})\geq \sqrt{\log{\frac{1}{\epsilon}}}\sqrt{\frac{\mathcal{A}}{4G_N}},
\end{equation}
where $G_N$ denotes the bulk gravitational constant, which is related to the central charge of the CFT through the equation $c=\frac{3R}{2G_N}$ for 1+1 d CFT, and $c\sim \frac{R^{d-1}}{G_N}$ for general dimensions. $R$ is the AdS radius.

\subsubsection*{Exact and Smoothed Magic in CFTs}
Having obtained a lower bound, we now examine the smoothed magic upper bound. let us pause for a moment and make an interesting observation about exact versus smoothed magic. Consider $n$ copies of $|\psi\rangle=a|00\rangle+b|11\rangle$ which is not maximally entangled. For any additive magic measure, the total magic $M\sim n$. The same can be deduced from the entropy bounds as both the lower and upper bounds pick up a constant multiple of $n$ compared to that of a single copy. 

However, smoothed entropies are not additive. If we allow for approximations, then it is known that \cite{Hayden_2003} for any state $|\psi\rangle$ there exist local unitaries $U_A\otimes U_B$ such that
\begin{equation}
    F(U_A\otimes U_B |\psi\rangle^{\otimes n} , |\Phi^+\rangle^{S_n-O(\sqrt{n})} \otimes |\chi\rangle)\geq 1-\epsilon
    \label{eqn:distill}
\end{equation}
for some $\epsilon$, where $F(\sigma,\rho)= (\Tr[\sqrt{\sigma^{1/2}\rho\sigma^{1/2}}])^2$ is the Uhlmann fidelity and $|\chi\rangle$ is a state that is entangling $O(\sqrt{n})$ qubits. Because the perfect Bell pairs $|\Phi^+\rangle$ contain zero magic, the smoothed nonlocal magic of such a system must be upper bounded by $O(\sqrt{n})$ with implicit $\epsilon$ dependence. From this, we can derive a tighter upper bound of $O(\sqrt{n})\sim O(\sqrt{S})$. This agrees with the lower bound up to constant factors. Hence assuming the distillation argument, the smoothed nonlocal magic $M_{RS}^{(NL,\epsilon)}(\rho_{AA^c})\sim O(\sqrt{S(A)})$. This is contrasted with magic scaling without smoothing, which has shown to scale linearly with $S(A)$ in tensor networks and small size numerics without smoothing.



A similar argument can be applied to CFTs by taking an $n$-fold tensor product. Let $|\psi\rangle_{AB}$ now be a CFT ground state with some fixed bipartition. Under such an $n$-fold tensor product, $c\rightarrow nc$ and the magic lower bound scales as $O(\sqrt{c})\rightarrow O(\sqrt{cn})$ where we take identical bipartitions $A,B$ for all copies of the CFT. Although the magic scaling is $O(n)$ according to the smoothed max entropy upper bound, by \cref{eqn:distill}, a tighter bound from Bell pair counting yield $O(\sqrt{n})$ scaling again. On the surface, an $n$-fold copy of CFTs should have $n$-fold increase of the nonlocal magic if the measure is additive, however, we see that smoothing in fact always brings about a quadratic reduction in the amount of required magic in producing an approximation of the target state. 

It is natural to ask whether the square root scaling of smoothed magic persists for $SU(N)$ gauge theories like holographic CFTs in the large $N$ limit and when the upper-lower bounds in~\cref{th:smoothed} are tight. Here we conjecture that the lower bound (\cref{eqn:smoothinequality}) is essentially saturated by smoothed magic whereas the nonsmoothed magic can scale linearly with the R\'enyi entropies $S_n(A)$. In other words, the upper bounds (\cref{prop:RelativeEnt}) and (\cref{eqn:smoothinequality}) are approximately saturated up to constant multiplicative factors.






\begin{conjecture}
\label{conj:cft}
Let $|\psi\rangle_{AB}$ be a low energy state of any conformal field theory. Assuming a UV cut off to render entropies finite, let $S(A)$ be the von Neumann entropy of the state on a contiguous subregion $A$. For any additive measure of magic,

\begin{enumerate}[label=(\alph*)]
    \item the smoothed nonlocal magic evaluated at any fixed precision $\epsilon$ is of $O(\sqrt{S(A)})$.
\item If the exact nonlocal magic is well-defined, then it scales as $O(S(A))$. \label{conjb}
\end{enumerate}

\end{conjecture}

A simple reasoning is as follows. Suppose the bipartite entanglement across $AB$ are distillable such that for each Planck area of the RT surface, we can obtain a Bell-like state $|\chi\rangle_{AB}$ which need not be maximally entangled; suppose these states are near identical by the conformal symmetries of the CFT ground state, then we must have $O(S(A))$ copies of such states. Following a distillation like~\cref{eqn:distill}, we obtain at most $O(\sqrt{S})$ states that are imperfectly entangled, in which nonlocal magic can reside. Note that if no smoothing is allowed, and the magic measure is additive, then the $O(S)$ number of entangled pairs simply contain $O(S)$ amount of magic, consistent with our MERA intuitions and CFT numerics. 

This conjecture, if true, has a wider implication for quantum simulations of conformal field theories. Although our na\"ive expectation is that the nonlocal magic should increase as the volume of the minimal surface, as indicated by holographic tensor networks, the magic needed to produce a good approximation allows a quadratic reduction. In terms of non-Clifford resources, it implies that an practical preparation of a CFT ground state may permit a quadratic reduction of $T$ gates compared to na\"ive expectations with moderate scaling with increasing precision $\epsilon$. 
However, the actual state preparation has to take into account local magic, which is volume law, and multipartite nonlocal magic, which is not covered by our bipartite analysis. Therefore, although a state isospectral to $\rho_A$ may consume less non-Clifford resource, we make no claim as to how it alters the total resource scaling for the preparation of $\rho_A$.


\subsubsection*{antiflatness and smoothed magic}
We now comment on a key relation between smoothed magic and entanglement in the CFT. It was suggested in \cite{white_conformal_2021} that magic nonlocally distributed would be needed to reproduce the antiflatness of the CFT entanglement spectrum. We have seen a version of it for exact magic in \cref{sec:CFT}. We can also verify this relation precisely for smoothed magic --- the spectral antiflatness $\mathcal{F}_R(\rho_A)$ is proportional to the amount of smoothed nonlocal magic $M_{RS}^{NL}(\rho_{AB})$ to leading order. However, the scaling with entropy is different. {\color{black} By combining the lower bound where nonlocal magic scales as $O(\sqrt{S})$ from Theorem 4 and an upper bound based on the approximate unitary distillation argument in Eq.\ \ref{eqn:distill} and Appendix~\ref{app:MPS}, we arrive at the following proposition.}
 
\begin{proposition}\label{prop:flatsmoothNLM}
    For any bipartition $A$ and $A^c$ of the CFT ground state, the antiflatness of the CFT entanglement spectrum necessitates the existence of smoothed nonlocal magic of at least $O(\sqrt{S(A)\log(1/\epsilon})$. If the distillation argument holds, then 
    \begin{equation}
        \mathcal{F}_R(\rho_A) \sim M_{RS}^{(NL,\epsilon)}(\rho_{AA^c})= O(\sqrt{S(A)}).
        \label{eqn:smoothentBound}
    \end{equation}
\end{proposition}

\section{Holographic Magic and Gravity} \label{sec:gravity}

Heuristically, antiflatness of the entanglement spectrum is critical in emerging gravity. Various approaches for (entanglement) entropic derivations of the Einstein's equations make use of entanglement first law in both AdS/CFT, e.g., Ref.~\cite{Blanco_2013,Faulkner_2014,swingle2014universality,Czech_2017}, and beyond~\cite{Jacobson_2016,Cao_2017,Cao_2018}. 
This simple relation connects the stress energy by way of modular Hamiltonian $H_A=-\log \rho_A$. Under a perturbation $\rho_A\rightarrow \rho_A+\delta \rho$ such that $\delta S \equiv S(\rho_A+\delta \rho) - S(\rho_A)$ and $\delta \langle H_A\rangle \equiv \Tr[H_A \delta{\rho}]$, then to linear order  $\delta S  = \delta \langle H_A\rangle$. As entropy is linked to the area of an extremal surface and $H_A$ can be linked to functions of the stress energy tensor in quantum field theories, $\delta\langle H_A\rangle$ is connected to perturbation in stress energy caused by the perturbation $\delta\rho$ while $\delta S$ can be linked to the area and hence metric perturbation. The combination of these relations produce the Hamiltonian constraint, where a covariantized version leads to (linearized) Einstein's equations. It is clear that if the spectrum was flat, i.e. the system has zero nonlocal magic and the modular Hamiltonian is proportional to the identity, then no state perturbation can ever incur entropy and therefore metric perturbations, let alone Einstein gravity. Therefore, it is natural to link nonlocal magic to the emergence of gravity by way of entanglement spectrum.

In this section, we examine nonlocal magic in CFTs with dual gravity theories. Although it is speculated that nonlocal magic should play an important role in the dual theory \cite{white_conformal_2021,nogo}, the precise relation has not been made clear. We now provide a holographic dual of nonlocal magic: nonlocal magic in the CFT is backreaction in the bulk. {\color{black} Note that in this section we always take the large central charge limit}.

\subsection{Brane tension and magic}\label{section:brane}
We now make a more precise statement from the point of view of R\'enyi entropies. Recall that  the R\'enyi entropies  in holographic CFTs are computed by the replica geometries which insert a conical singularity that correspond to cosmic branes at various tensions \cite{Dong:2018lsk,Dong1}. Therefore, antiflatness in the entanglement spectrum can be naturally interpreted as the difference between minimal surfaces areas in different backreacted geometries caused by the addition of some stress energy in the form of a cosmic brane with tension $\mathcal{T}$. 

More precisely, the derivative of brane area is related to antiflatness (\cref{def:braneflatness}),
\ba
\frac{\partial_n A_n}{4G}=\partial_n\tilde{S}_n.
\ea

The brane tension $\mathcal{T}$ is related to $n$ by 
\ba
\mathcal{T}_n=\frac{n-1}{4n G}
\ea

Hence for $n=1$, or tension $\mathcal{T}=0$, we have that 
$4G\partial_n A_n|_{n=1} = \partial \mathcal{A}/ \partial\mathcal{T}|_{\mathcal{T}=0}$.
Applying (\cref{eqn:NLMflatness}) we arrive at a linear relation between $\partial\mathcal{A}/\partial\mathcal{T}\sim \mathcal{M}_2(|\phi\rangle)$, specifically
\be\label{branemagicNL}
\begin{split}
\left\vert\frac{\partial \mathcal{A}}{\partial \mathcal{T}}\right\vert_{\mathcal{T}=0} = (4G)^2|\partial_n\tilde{S}_n||_{n=1}\approx \frac{(4G)^2}{\kappa} \mathcal{M}_2^{NL}(|\psi\rangle_{AB})
\end{split}
\ee
which then provides an estimate for the nonlocal magic $M_{\rm dist}^{(NL)}$ across the bipartition from~\cref{th:magicdist}. Note that the bipartition is arbitrary and each subregion $A$ need not be connected.

{\color{black}By showing that a modified entanglement capacity evaluated at $n=1$ serves as a measure of nonlocal magic, we also have 
\begin{equation}
    \frac{\partial \mathcal{A}}{\partial \mathcal{T}}\simeq-M_{C_E}^{NL}(\psi).
\end{equation}
}


That is, nonlocal magic controls the level of geometric change in response to adding mass energy in the bulk, where the zero magic limit indeed recovers the trivial response function in stabilizer holographic tensor networks. As we showed earlier in \cref{lemmaNL}, antiflatness is  zero if and only if the nonlocal magic vanishes. Then, through \cref{branemagicNL},  there is no backreaction in the zero magic limit. This is consistent with results from \cite{nogo}.


\begin{remark}
    Recall the flatness problem of the entanglement spectrum is also present in random tensor networks even though they are not stabilizer codes. This is because nonlocal magic is also low for Haar random states (\cref{rmk:1}), even though they are not stabilizer codes. Therefore the same type of gravitational backreaction is also ``turned off'' in Ref.~\cite{Hayden_2016}.
\end{remark}

A more rigorous bound relating nonlocal magic and the R\'enyi entropy derivatives $\partial_n \mathcal{A}$  can also be proven. 

\begin{proposition}\label{pp:branebound}
Assuming the distillation argument where $U_A\otimes U_B\ket{\psi}_{AB} \approx \otimes_i\ket{\phi_i}_{a_ib_i}$ for the state with local magic removed, then the nonlocal stabilizer R\'enyi entropy for a CFT under bipartition $AB$ is bounded by
    \begin{equation}
    \frac{1}{2}\left\vert\frac{\partial_n\mathcal{A}_n|_{n=2}}{4G}(\ket{\psi}_{AB})\right\vert\leq\mathcal{M}_2(\ket{\psi}_{AB})\leq\left\vert\frac{\partial_n\mathcal{A}_n|_{n=1}}{4G}(\ket{\psi}_{AB})\right\vert
\end{equation}
\end{proposition}

See proof in~\cref{app:branebound} and justification of the distillation assumption for CFT in~\cref{app:MPS}. 
We elaborate the regime of validity for various magic bounds and antiflatness relations in \cref{app:bound}.

\subsection{Magic in Holographic CFT}
Note that magic in quantum many-body systems is generally difficult to compute as the cost can grow exponentially with the system size\cite{white_conformal_2021,PhysRevA.106.042426}. This scaling is much improved for measures like stabilizer R\'enyi entropy where the nonlinear function of the state can be computed using MPS\cite{haug_efficient_2023,tarabunga2023manybody,tarabunga2024nonstabilizerness} or enumerator-based tensor networks \cite{cao2023quantum}. However, the computation remains costly at high bond dimensions and for other measures.  
On the other hand, the bounds of magic from~\cref{sec:magicbounds} offer an entropic perspective into this otherwise hard-to-compute quantity by leveraging existing results. 

We now study nonlocal magic in CFTs in light of the general relations derived in~\cref{sec:magicbounds}.
Using the holographic dictionary and applying \cref{conj:cft}, we can predict the behavior of nonlocal magic in CFTs that are otherwise difficult to compute. Although the following examples essentially amounts to putting square roots on known holographic entanglement entropies, it is instructive to review their behaviors and analyze their implications for magic and, by extension, classical complexity and quantum resource needed for state preparation. At the same time, holographic calculations enable us to study magic dynamics under quantum quenches, for which existing results have been sparse and size-limited \cite{Sewell} due to prohibitive computational costs.

\subsubsection*{Static Configurations}
We now apply (\ref{eqn:smoothentBound}) to estimate the smoothed nonlocal magic in the CFT state. To illustrate, consider the thermal state $\rho_{AA^c}$ of a (1+1)D CFT which is purified by $B$, e.g. in a thermal field double state.
\begin{equation}
    |TFD\rangle \propto \sum_n \exp(-\beta E_n/2) |E_n\rangle_{AA^c} |E_n\rangle_{B}
\end{equation}

Bipartitioning the system into $A$ and $A^c\cup B$, the behavior of the nonlocal magic is given by
\begin{equation}
   \sma{|TFD\rangle_{A{A}^cB}}\sim \sqrt{\frac{c}{3}\log\left(\frac{\beta}{\pi\delta_{UV}}\sinh\left(\frac{\pi l}{\beta}\right)\right)},
\end{equation}
where $l$ is size of subregion $A$. The magic increases logarithmically with the subregion size 
$l$ for $l\ll \beta$. However, when the size surpasses the thermal correlation length, represented by $\beta=\frac{1}{T}$, it becomes proportional to $\sqrt{l}$. A similar result holds for a small subsystem $A$ of a pure state $|\psi\rangle_{AA^c}$ with that thermalizes under ETH such that $A$ has fixed temperature $T=1/\beta$.

Now instead consider the bipartition of the system in to $AA^c$ and $B$. It is known that for holographic CFTs, the system undergoes a confinement-deconfinement phase transition which corresponds to the Hawking-Page transition in the bulk at a critical temperature $T_c$\footnote{This is a simplified account of the transition, which for different theories there can be different phases as one dial up the temperature\cite{Aharony_2004}.}. 

It is known that 
\begin{equation}
    S(B)=S(AA^c) \sim \begin{cases}
        O(N^0)\quad T<T_c\\
        O(N^2)\quad T>T_c
    \end{cases}
\end{equation}
In the same way, we predict a magic phase transition where $M_{RS}^{(NL,\epsilon)}/N$ is discontinuous across $T_c$ in the $N\rightarrow \infty$ limit.

\subsubsection*{Local quench}
In the following sections, we consider several time-dependent scenarios and analyze their implications on the system dynamics.

For our first scenario, let us examine a CFT ground state that is been perturbed by a smeared local operator $O_{\alpha}(x,0)$ at $t=0$.  This is then subjected to time evolution governed by the CFT Hamiltonian. We can express the state as:
\begin{equation}
    |\psi(t)\rangle=\mathcal{N}e^{-iHt}e^{-\delta H}O_{\alpha}(x,t)|\Omega\rangle. 
\end{equation}

In the corresponding bulk dual, this equates to introducing an in-falling particle with mass $m$ into the initially vacuum anti-de Sitter-Vaidya spacetime. The energy-momentum tensor for this scenario can be characterized as:  
\begin{equation}
    T_{uu}=\frac{mR\alpha^2}{8\pi (u^2+\alpha^2)^2}. 
\end{equation}

Here, $\alpha$ denotes the size of the smeared operator. As 
$\alpha$ approaches 0, this converges to a delta function in $u$. The subsequent effect on the bulk spacetime is encapsulated by a shock-wave geometry, as illustrated below in \figref{figholo1}.

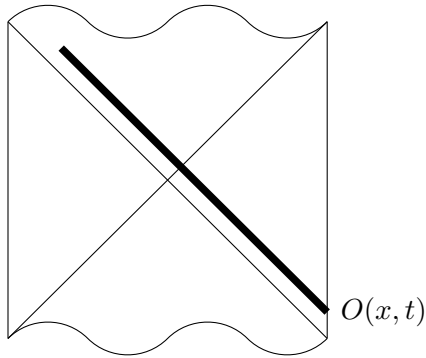
\begin{figure}[H]
\center
\begin{tikzpicture}[scale=0.7]
\draw (-3,-3)--(-3,3);
\draw (3,-3)--(3,3);
 \draw (-3,3) to[out=45,in=135] (-1.5,3) to[out=-45,in=-135] (0,3) to[out=45,in=135] (1.5,3) to[out=-45,in=-135] (3,3);
\draw (-3,-3) to[out=45,in=135] (-1.5,-3) to[out=-45,in=-135] (0,-3) to[out=45,in=135] (1.5,-3) to[out=-45,in=-135] (3,-3);
\draw (-3,3)--(3,-3);
\draw (-3,-3)--(3,3);
\draw[line width=1mm] (3,-2.5)node[right]{$O(x,t)$} --(-2,2.5);
\end{tikzpicture}
\caption{Penrose diagram depicting a shock wave in global coordinates. }
\label{figholo1}
\end{figure}

We aim to investigate the nonlocal magic of subsystem 
$A$ in relation to $A^c$. These subsystems are separated by the boundary $\partial A=\partial A^c$, a 
$d-2$ sphere of radius $l$. By solving the Einstein equation, Nozaki \textit{et al.}~\cite{Nozaki:2013wia} derived the leading-order change in entanglement entropy due to the injected energy. Specifically, for a (1+1)D holographic CFT, this change is expressed as:
\begin{equation}
\begin{split}
    \Delta S(t)=&\frac{2mRl\alpha+mR(l^2-\alpha^2-t^2)\arctan(\frac{2\alpha l}{t^2+\alpha^2-l^2})}{8l\alpha}\\
    &+O\left((mR)^2\right).
\end{split}
\end{equation}

We can then employ the lower bound to estimate the growth of the nonlocal magic as follows:
\begin{equation}
\begin{split}
    M^{(NL)}_{RS}(|\psi(t)\rangle)\sim& \sqrt{S(0)+\Delta S(t)}\\
    \approx& \sqrt{S(0)}+\frac{1}{2}\frac{\Delta S(t)}{\sqrt{S(0)}}. 
\end{split}
\end{equation}

In the early-time regime,  $t\ll \sqrt{l^2-\alpha^2}$, the magic exhibits quadratic growth with time, independent of the spacetime dimension. This can be expressed as: 
\begin{equation}
    \Delta M^{(NL)}_{RS}(t)\sim \kappa_d\frac{mR}{\sqrt{S_0}}(\frac{\alpha l}{l^2-\alpha^2})^2\frac{t^2}{l^2-\alpha^2}+O(\frac{t^4}{(l^2-\alpha^2)^2}).
\end{equation}

At $t=\sqrt{l^2-\alpha^2}$, the magic reaches its peak value of $\Delta M^{(NL)}_{RS}=\kappa_d\frac{mR}{\sqrt{S_0}}$, after which it declines to zero. In the long-term regime, it decays following a power-law pattern: 
\begin{equation}
    \Delta M^{(NL)}_{RS}(t)\sim \frac{mR}{\sqrt{S_0}}\left(\frac{\alpha l}{t^2}\right)^d\left(1+O(\frac{l^2-\alpha^2}{t^2})\right).
\end{equation}

For the (1+1)D CFT, another intriguing scenario arises when subsystem $A$ encompasses half of the space, signifying $l\rightarrow \infty$. In this context, there exists a range in which the magic grows logarithmically with $t$ \cite{Caputa:2014vaa}, specifically when $l\ll t\ll D^{1/mR}\alpha$,
\begin{equation}
    \Delta M^{(NL)}_{RS}(t)\sim \frac{mR}{\sqrt{S_0}}\log\frac{t}{\alpha},
\end{equation}
where $D$ is quantum dimension of the quench operator $O$. The value reaches a constant late-time limit of 
$\Delta M^{(NL)}_{RS}=\frac{\log D}{\sqrt{S_0}}$. This logarithmic growth can only be observed in system with large central charge due to the otherwise small value of $D$.
Note that holographic methods are at a distinct advantage here because magic dynamics for large systems over long periods of time is numerically intractable using existing methods.

\subsubsection*{Global Quench}
We also explore the global quench scenario wherein the perturbation is not confined to a localized region but influences the entire CFT state. Within the bulk dual, this corresponds to a spherically symmetric in-falling mass shell.

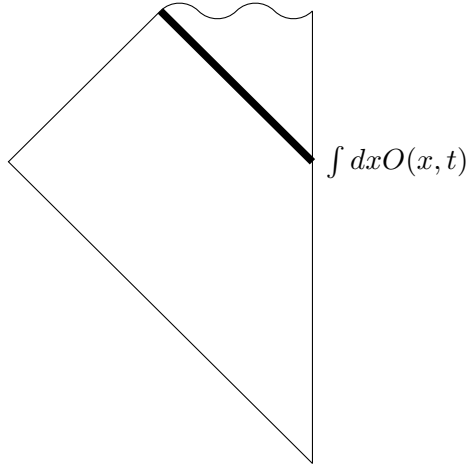
\begin{figure}[H]
\center
\begin{tikzpicture}
\draw (1,-5)--(1,1);
\draw (-1,1) to[out=45,in=135] (-0.5,1) to[out=-45,in=-135] (0,1) to[out=45,in=135] (0.5,1) to[out=-45,in=-135] (1,1);
\draw[line width=1mm] (1,-1)node[right]{$\int dx O(x,t)$} --(-1,1);
\draw (-1,1)--(-3,-1)--(1,-5);
\end{tikzpicture}
\caption{Vaidya geometry. Right side boundary denotes the asymptotic boundary of the anti-de Sitter-Vaidya spacetime. Outside the mass shell is the black hole geometry. Inside the mass shell is  Vacuum anti-de Sitter in Poincaré patch.}
\label{figholo2}
\end{figure}

The geometry impacted by the mass shell is characterized by the Vaidya metric. This is essentially the integration of pure AdS with an AdS-Schwarzschild black hole, aligned along the mass shell, as illustrated in~\cref{figholo2}.

The shell's descent into the bulk parallels the boundary CFT's thermalization following the global perturbation. The state transitions from the ground state and progressively thermalizes to a certain finite temperature. The entanglement entropy of subregion $A$ serves as a quantitative measure, increasing during this process. Correspondingly, in the bulk perspective, this entropy surge is represented by the expanding area of the minimal surface anchored to the boundary of $A$. 

In a (1+1)-dimensional CFT, it is feasible to precisely solve for the minimal surface \cite{Balasubramanian:2011ur}. The entropy at 
$t=0$ is equivalent to the CFT ground state entropy, given by $S(0)=\frac{c}{3}\log{\frac{l}{\delta_{UV}}}$. This aligns with the length of the geodesic fully contained within the pure AdS. Following the onset of the quench, the geodesic begins to intersect with the in-falling mass shell, causing its length to increase over time. Initially, this growth is quadratic with respect to $t$,
\begin{equation}
    \mathcal{L}(t)=2\log{\frac{l}{\delta_{UV}}}+2\frac{\pi^2 t^2}{\beta^2}+O(t^3). 
\end{equation}

As thermalization progresses, the geodesic's intersection with the mass shell delves deeper into the bulk. Once the subregion completes its thermalization at time $t=\frac{l}{2}$, the geodesic no longer intersects the in-falling shell, stabilizing its length to an equilibrium value,
\begin{equation}
    \mathcal{L}(t>\frac{l}{2})=2\log{\frac{ \beta}{\pi\delta_{UV}}\sinh{\frac{\pi l}{\beta}}}. 
\end{equation}

We also detail the behavior of the geodesic length in the late stages, prior to reaching full thermalization, as outlined below:
\begin{equation}
\begin{split}
    \mathcal{L}(t\lesssim\frac{l}{2})=&2\log{\left(\frac{ \beta}{\pi\delta_{UV}}\sinh{\frac{\pi l}{\beta}}\right)}-\frac{2}{3}\sqrt{2\tanh{\frac{\pi l}{\beta}}}\left(\frac{l}{2}-t\right)^{\frac{3}{2}}\\
    &+O\left(\left(\frac{l}{2}-t\right)^{2}\right).
\end{split}
\end{equation}

Based on the aforementioned results, the evolution of the smoothed nonlocal magic for a subregion in a 2D CFT can be characterized as follows: it increases according to, 
\begin{equation}
    \sma t \sim\sqrt{c\abs{\log\epsilon}}\left(\sqrt{S_0/c}+\frac{\pi^2 t^2}{6\sqrt{S_0/c}}+O(t^3)\right),
\end{equation}
during the initial stages, and as,
\begin{equation}
\begin{split}
    \sma{t}\sim&\sqrt{c\abs{\log\epsilon}}\left(\sqrt{S_T/c}-\frac{1}{18}\frac{\sqrt{2\tanh{\frac{\pi l}{\beta}}}}{\sqrt{S_T/c}}\left(\frac{l}{2}-t\right)^{\frac{3}{2}}\right.\\
    &+\left.O\left(\left(\frac{l}{2}-t\right)^{2}\right)\right),
\end{split}
\end{equation}
during the latter phases when the subregion is nearing full thermalization. This can be contrasted with the dynamics of total subsystem magic under thermalization\cite{Sewell} which decays after a quick initial rise.

\subsubsection*{Wormhole}
Lastly, we examine a thermalization process involving two copies of CFT states. This dynamic process corresponds to the evolution of an expanding wormhole in the bulk dual. 

\begin{figure}[H]
    \centering
    \begin{tikzpicture}[scale=0.6]
        \def\s{10}
        \draw (0,0) ellipse (1cm and 2cm);
        \draw (\s,0) ellipse (1cm and 2cm);
        \draw[red,line width=0.5mm] (-1,0) arc(180:360:1cm and 2cm) ;
        \draw[red,line width=0.5mm] (\s-1,0) arc(180:360:1cm and 2cm);
        \draw (0,2) ..controls (\s/2,1).. (\s,2);
        \draw (0,-2).. controls (\s/2,-1).. (\s,-2);
    \end{tikzpicture}
    \caption{Wormhole geometry}
    \label{figholo3}
\end{figure}
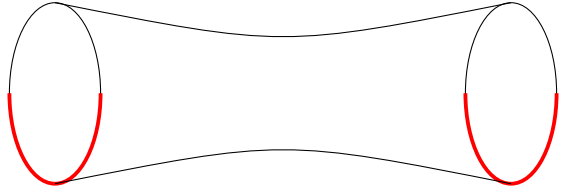

Let us revisit the thermal field double (TFD), 
\begin{equation}
    \ket{\text{TFD}}=\frac{1}{\sqrt{Z(\beta)}}\sum_n e^{-(\frac{\beta}{2}+2it) E_n}\ket{E_n}_L\ket{E_n}_R
\end{equation}

We designate our region of interest to encompass a section from both the left and right CFT states (illustrated in~\cref{figholo3}). The entanglement entropy of this composite region is probed by the extremal surface spanning the wormhole, connecting the left and right segments.

In this setup, we assume symmetry when exchanging the two CFT sides. Specifically, we mandate that the subregion $A$ on one side mirrors its counterpart on the other side. See red region in~\cref{figholo3}. Given this symmetry, the extremal surface occupies a plane defined by constant transverse spatial coordinates and is characterized solely by the relationship between time and the radial direction.

At the boundary time $t=0$, the area of extremal surface is given by
\begin{equation}
    \mathcal{A}(0)=\frac{\beta r_{\infty}}{\pi}V_{d-2}.
\end{equation}
where $r_{\infty}$ is the UV cutoff of radial coordinates. This extremal area is proportional to the volume of subregion boundary $\partial A$, reminiscent of the area law entanglement observed in gapped systems.  As time progresses, the extremal surface accrues additional contributions from regions beyond the horizon. As highlighted in Ref.~\cite{Hartman:2013qma}, this contribution exhibits a straightforward linear relationship with the boundary time, as illustrated below:
\begin{equation}
    \mathcal{A}(t)= \frac{4\pi t}{\beta}\alpha_d V_{d-2},   \qquad \text{for $t\gg \beta$}.
\end{equation}

The linear growth eventually ceases when the extremal surface traversing the wormhole is surpassed by another, more minimal configuration. A different set of competing extremal surfaces, anchored to the same entangling boundary but bypassing the wormhole, emerges. These surfaces are essentially combinations of the extremal surfaces corresponding to subregions within each individual thermal CFT. Their area is given by
\begin{equation}
    \mathcal{A}(\infty)-\mathcal{A}(0)=\frac{2\pi}{\beta}V_{d-1}.
\end{equation}

The transition of dominant extremal surface occurs around $t\sim R$, which corresponds to the size of the subregion under consideration. Consequently, we anticipate the nonlocal magic in this TFD state to scale as follows:
\begin{equation}
\begin{split}
\sma{t}&\sim\sqrt{\abs{\log{\epsilon}}}\sqrt{S_0+\frac{4\pi t}{\beta}\alpha_dV_{d-2}}, \qquad \text{for $\beta\ll t<R$}\\
&\sim \sqrt{\abs{\log{\epsilon}}S_T},  \qquad \text{for $t\geq R$}.
\end{split}
\end{equation}

\section{Discussion}
In this work, we explored the question: what dual boundary quantity enables gravitational backreaction in the bulk? The celebrated formula of Ryu and Takayanagi provides a fundamental observation of the AdS-CFT conjecture by showing that areas in AdS correspond to entanglement entropies in the CFT. 
In the greater context of spacetime and gravity emerging from quantum information, we ask: If entanglement builds geometry, then what builds gravity?
In this work we show that the strength of gravitational backreaction is connected to (nonlocal) magic in CFT. In other words, gravity is magical! Accordingly, both defining properties of quantumness admit holographic counterparts in AdS. 

To obtain this result, we 
studied the interplay between nonlocal magic and entanglement. We show that for any quantum state in a finite dimensional Hilbert space, this form of nonstabilizerness that can only live in the bipartite correlations is lower bounded by the antiflatness of the entanglement spectrum and upper bounded by the amount of entanglement in the system as defined by R\'enyi entropies. We then apply these results to CFTs and conclude that both the exact and smoothed nonlocal magic is proportional to various notions of antiflatness. However, they scale differently with entropy --- the exact nonlocal magic scales linearly with the von Neumann entropy of a CFT subregion while the smoothed magic only scales as the square root. Numerically we verify that nonlocal magic is sensitive to quantum phase transition in a way that is different from entanglement. We also examined its behavior under symmetry breaking. 

Finally, in the context of holographic CFTs, we derive a quantitative relation between nonlocal magic and the level of gravitational backreaction. Using the bulk gravity theory, smoothed nonlocal magic in the CFT can also be estimated holographically. As nonstabilizerness in quantum systems are generically hard to compute, our work also provides an important estimate on the practical level and constrain magic distributions using existing data and well-founded methods like tensor networks and DMRG. \textcolor{black}{An alternative connection was established by showing that a modified entanglement capacity, at $n=1$, is likewise a faithful nonlocal magic measure. This measure is related to entanglement capacity within constant offset, and therefore holographic calculations can be utilized directly for its computation. While general entanglement capacity calculations can be difficult, as they rely on identifying extremal surfaces in the backreacted geometries with cosmic branes, perturbative calculations are readily computable for small $\delta n =n-1$, e.g., Ref.~\cite{Nakaguchi:2016zqi}. }

{\color{black} While our core findings pertain to quantum information, they also carry important implications for the design of holographic toy models and the resource estimates necessary to simulate conformal field theories. Although it is expected that Clifford processes, such as those used to prepare holographic stabilizer codes~\cite{Pastawski_2015,ABSC,Pollack_2022}, are insufficient for realizing the emergence of gravity, our results indicate that it is likewise insufficient to arbitrarily distribute magic in the CFT. In fact, to the contrary, magic need be injected by the correct amount and nonlocally smeared across the state. This observation can be realized in random tensor networks which, while possessing magic in abundance, have virtually no nonlocal magic. Our finding is consistent with a separate line of investigation, based in operator algebra quantum error correction, that has shown a nonlocal form of nonstabilizerness is required to reproduce a nontrivial area operator, which is expected for physical spacetimes and emergent gravity~\cite{nogo}. 

Furthermore, our findings provide a quantitative estimate for precisely where, and how much, non-stabilizer resources need to be injected, thus constituting an important constraint for future modifications of holographic toy models that seek to realize emergent gravity. This same resource estimate also indicates that the nonstabilizerness needed for simulating conformal field theories is not as high as one might na\"ively expect. Our results therefore imply a lower limit on efficient quantum simulations of quantum field theories, which promise long term significance in the study of resource estimation and state preparation.}

There are several directions that are of interest for future work. The key constraints for nonlocal magic here are given in terms of inequalities. Part of the reason for bounds instead of a precise equality is that nonlocal magic requires extremization while the computation of magic itself is already nontrivial. However, given the universal behavior of nonlocal magic across multiple distinct measures of antiflatness, there is reason to believe that a unifying statement or even a precise equality exists between entanglement spectral properties and magic. In the particular case of quantum field theory, it is also crucial to generalize our observations to definitions of magic that is native to the infinite dimensional system, e.g. non-Gaussianity, as well as other different measures of spectral antiflatness. 

Approaching nonlocal magic from a different perspective, we can start with state $\rho_A$ from the usual stabilizer polytope and construct a purified state $\psi_{AB}$. One can also define a nonlocal magic as the minimal magic among all possible purifications.
In the same vein of connecting magic with entanglement, we ask whether it is possible to define instead magical entanglement, i.e., the entanglement that cannot be removed by any Clifford operation\footnote{We thank Kaifeng Bu for this suggestion.}. In this case, one can easily show from our entropy bounds that magical entanglement is an upper bound of nonlocal magic. However, it is yet unknown whether the two definitions are equivalent. Finally, recall that nonlocal magic can be generalized to systems with multipartite entanglement. This will be crucial in understanding the behavior of e.g. Haar random states, random tensor networks, and holographic states. As the type of multipartite entanglement is quite constrained for stabilizer states, nonlocal magic may be crucial in the classification of multipartite entanglement.

For CFTs specifically, several of our results rely on the assumption that the bipartite entanglement across a subregion $A$ and its complement $B$ in a pure state can be approximately converted into a tensor product of entangled pairs through unitaries that only act on the respective subregions. Although this assumption is well-supported by numerics and well-motivated by holographic tensor networks models, it is unclear the extent to which this holds for a single copy of (holographic) CFT in general. This assumption may also admit further modification in the case where $A$ consists of multiple disjoint regions. It is important that we understand the regime of validity for such assumptions and pave the way for proving \cref{conj:cft} and extending the generality of \cref{prop:flatsmoothNLM}. 

Just as various types of entanglement can admit different holographic interpretations, a similar situation may hold for magic. While we take a first step toward addressing the open question of what is the holographic dual of magic, much remains unknown. For instance, the connection we identify with antiflatness signals a link between nonlocal magic and gravitational backreaction. However, because we lack a systematic understanding of how the bulk duals should deform under a sequence of boundary theories that have increasing flat spectrum, the physical meaning of how the removal of magic turns off backreaction is unclear. {\color{black} Additionally, we only consider backreaction generated by the stress energy of a cosmic brane, instead of a planet or a star in the bulk. While it is possible that other gravitational phenomena such as a planet or star generates backreaction that emerge from a different magical origin on the boundary, it is important to recognize that this backreaction will similarly be absent without antiflatness (and thereby nonlocal magic) as discussed in Sec~\ref{sec:gravity}. Nevertheless, it is possible that important differences exist in how distinct backreactions manifest in the boundary theory.} If that is the case, we eventually wish to distinguish them from the consequence of bipartite nonlocal magic in the boundary theory. {\color{black} More broadly still, while the holographic connections established above provide many useful clues for understanding quantum many-body magic, the large variety among magic measures render the reverse implications difficult as it can be challenging to discern which provides the most natural gravitational probes. We do, however, expect multipartite magic to carry a lot more information about bulk physics, where existing entropy data is limited and where magic is required to sustain nontrivial bulk connectivity~\cite{Hayden:2021gno,Akers:2019gcv} without violating the holographic entropy inequalities.}

Although it has been suggested that boundary states with flat entanglement spectrum are dual to peculiar bulk states of fixed areas\cite{Akers:2018fow,Dong_2019}, exactly how these bulk states should be interpreted holographically remains to be understood. To this end, a more precise relation between magic and emergent gravity \cite{Faulkner_2013} in the bulk, one which does not rely on the distillation assumptions used in this work, is highly desirable. Furthermore, a connection between magic and a local function of curvature generated by more physical forms of stress energy instead of an extended conical singularity such as a cosmic brane may provide a more natural link with the Einstein's equations or the Hamiltonian constraint. A more comprehensive understanding of holographic magic through the lens of dynamics such as quantum chaos\cite{chaosbymagic} and (classical) complexity can also provide another unique perspective that is not captured by our current work.

Finally, this work calls for several important lines of investigation as we move toward establishing nonlocal magic as a key metric for characterizing quantum many-body systems. For instance,  the tensor product of random single-qubit states, the ground states of physical quantum many-body systems, and the Haar random states all have volume law magic scaling. Purely from the point of view of entanglement entropy, they can also be mimicked by stabilizer states. However, their nonlocal magic behaves very differently. Thus it provides a distinct indicator for the properties of the underlying quantum systems that are invisible to entanglement entropy or total nonstabilizerness alone. It would also be intriguing to study the role of nonlocal magic in quantum phase transition, in symmetry breaking, and in nonequilibrium systems.





\acknowledgments
We would like to thank Chris Akers, Vijay Balasubramanian, Ning Bao, Ed Barnes, Kaifeng Bu, Xi Dong, Sophia Economou, Monica Kang, Cynthia Keeler, Nick Mayhall, Jason Pollack, Howard Schnitzer, Brian Swingle, Christopher White, and Tianci Zhou for helpful comments, resource, and discussions. We are especially grateful to Christopher White and Daniele Iannotti for identifying the mistakes in the earlier version of this manuscript. C.C. and A.H. would like to thank the organizers of the Quantum Information and Quantum Matter Conference at NYU Abu Dhabi during which this work was first conceived.
 C.C.\ acknowledges the support by the National Science Foundation (No. PHY-1733907) and the Commonwealth Cyber Initiative. The Institute for Quantum Information and Matter is an NSF Physics Frontiers Center.
A.H. acknowledges support from PNRR MUR Project (No. PE0000023-NQSTI) and
PNRR MUR Project (No.~CN $00000013$-ICSC). W.M. is supported by the U.S. Department of Energy under Grant No.~DE-SC0019470 and by the Heising-Simons Foundation ``Observational Signatures of Quantum Gravity'' Collaboration Grant No.~2021-2818. S.F.E.O. acknowledges support from PNRR MUR Project (No.~PE0000023-NQSTI). L.L. is funded through the Munich Quantum Valley project (No.~MQV-K8) by Bayerisches Staatsministerium für Wissenschaft und Kunst and DFG (No.~CRC 183).

\section{Data availability}
The data that support the findings of this article are
openly available \cite{cheng2025dataset}.

\appendix
\section{Invariance of $\stab_0$}\label{app:stab0invariance}
In this section, we prove that $\stab_0$ is invariant under the following operations
\begin{enumerate}
	\item Clifford unitaries. $\rho\rightarrow U\rho U^{\dagger}$ with $U\in \mathcal{C}(d^n)$.
	\item Composition with stabilizer states, $\rho \rightarrow \rho \otimes \sigma $ with $\sigma $ a stabilizer state.
 	\item Partial trace of the first qudit, $\rho \rightarrow \Tr_{1}(\rho)$
	\item Computational basis measurement on the first qudit, $\rho \rightarrow (\st{i}\otimes \bbbone_{n-1}) \rho (\st{i}\otimes \bbbone_{n-1} )/\Tr(\rho \st{i}\otimes \bbbone_{n-1})$ with probability $\Tr(\rho \st{i}\otimes \bbbone_{n-1})$

\end{enumerate}
\begin{proposition}{Clifford Invariance.}
Given $\sigma\in\stab_0$ and $C\in\mathcal{C}(d^n)$, then $C\sigma C^{\dagger}\in\stab_0$
\begin{proof}
    \begin{equation}
    C\sigma C^\dagger = \frac{1}{d^n}\sum_{P\in G} C PC^\dagger = \frac{1}{d}\sum_{\tilde{P}\in\tilde{G}} \tilde{P}
\end{equation}
the latter is an element of $\stab_0$  since it is the equal-weighted sum of Pauli operators of a commuting set. This is since $C:P\mapsto \tilde{P}\in\tilde{\mathcal{P}}$ and the action of a unitary on a subgroup $G$ does not modify the commutation relations.
\end{proof}
\end{proposition}
\begin{proposition}
    Given $\rho\in\stab_0$ and $\tau\in\stab_0$ then $\rho\otimes\tau\in\stab_0$
\end{proposition}
\begin{proof}
    \begin{equation}
        \rho\otimes\sigma = \frac{1}{d^{2n}}\sum_{P\in G_1, Q\in G_2}P\otimes Q=\frac{1}{d^2}\sum_{P\otimes Q\in G_1\times G_2} P \otimes Q. 
    \end{equation}
where the latter is an element of $\stab_0$ since the tensor product of Pauli operators is still a Pauli operator and the Cartesian product of a group is still a group, and since the tensor product does not affect the commutation relations of the $G_1$ or $G_2$, then $G_1\times G_2$ is a commuting group and so $\rho\otimes\sigma\in\stab_0$.
\end{proof}
\begin{proposition}
    Given a state $\rho\in\stab_0$ then $\Tr_1{\rho}\in\stab_0$
\end{proposition}
\begin{equation}
    \Tr_1(\rho)=\frac{1}{d^n}\sum_{P\in G}\Tr(P_1)P_{2\ldots n}=\frac{1}{d^{n-1}}\sum_{P_{2\ldots n}\in\Tr_1(G)}P_{2\ldots n}
\end{equation}
where $P_1$ labels the Pauli operator on the first qudit of $P$. It is easy to observe that the only elements whose partial trace is different from $0$ are the ones with $P_1=\bbbone$. These elements that were in $G$ are still commuting Pauli operator in the traced group $\Tr_1{G}$. 
\begin{proposition}
    Given a state $\rho$ and $\{|i\rangle\}$ the 1-qudit computational basis, then $(\st{i}\otimes \bbbone_{n-1}) \rho (\st{i}\otimes \bbbone_{n-1} )/\Tr(\rho \st{i}\otimes \bbbone_{n-1})\in\stab_0$
\begin{proof}
For the sake of simplicity, let us consider the case for a multi-qubit system and $i=0$, it can be easily generalized for $i\neq 0$ and qudits.
    \begin{equation}
    \begin{split}
    &\frac{(\st{0}\otimes \bbbone_{n-1}) \rho (\st{0}\otimes \bbbone_{n-1} )}{\Tr(\rho \st{0}\otimes \bbbone_{n-1})}\\
    =&\frac{\sum_{P\in G}\Tr(\st{0}P_1) \st{0}\otimes P_{2\ldots n}}{\sum_{P\in G}\Tr(\st{0}P_1)\Tr(P_{2\ldots n})}\\
    =&\frac{1}{2^{n-1}}\frac{\sum_{P\in G}\Tr(P_1\st{0})\st{0}\otimes P_{2\ldots n}}{\sum_{P_1\in\Tr_{2\ldots n}G} \Tr(P_1\st{0})}\\
    =&\frac{1}{2^{n-1}}\frac{\sum_{P\in G|P_1\in\{\bbbone,Z\}}\st{0}\otimes P_{2\ldots n}}{\sum_{P_1\in{\bbbone,Z}\cap\Tr_{2\ldots n}G}}\\
    =&\st{0}\otimes\frac{1}{2^{n-1}}\frac{\sum_{P\in G|P_1\in\{\bbbone,Z\}}P_{2\ldots n}}{\sum_{P_1\in{\bbbone,Z}\cap\Tr_{2\ldots n}G}}\\
    \end{split}
    \end{equation}
    note that $\sum_{P_1\in{\bbbone,Z}\cap\Tr_{2\ldots n}G}$ can be either $1$ or $2$, due to the terms $\bbbone_{n}$ $Z\bbbone_{n-1}$. While on the numerator the only terms surviving have on the first qubit $\bbbone$ or $Z$. Now it is not difficult to see that for $G$ to be a commuting group if $\sum_{P_1\in{\bbbone,Z}\cap\Tr_{2\ldots n}G}=1$ then there will be no multiplying factor to the numerator, while in the other case, there will be a $2$ since each nonzero $P_{2\ldots n}$ has to repeat twice. Then it is not difficult to see that one has a stabilizer state, because $P_{2\ldots n} $ is still summing on a commuting Pauli subgroup. 
\end{proof}

\end{proposition}
\subsection{Proof of~\cref{prop:qrf}}\label{qrf}
\begin{proof}
   Let us start by expanding the relative entropy, we have 
   \begin{align}
   \mathcal{F}_R(\rho)=-\min_{\sigma\in \mathrm{FLAT}^{(n)}}\Tr[\rho\log\sigma]-S(\rho).
   \end{align}
   Since the elements of FLAT are all proportional to projection operators,  it is possible to choose $\sigma$ such that $\sigma=\bbbone_{ r}/r\oplus 0_{d- r}$ is diagonal in the same basis as $\rho$ where $r\equiv \rank(\rho)$, $\bbbone_r$ is the identity on a subspace of dimension $r$, and $0_{d-r}$ is the zero-matrix of dimension $(d-r)\times (d-r)$. Hence the first term becomes $\log r\,\Tr\rho = \log r=S_{max}(\rho)$ and $\mathcal F_R(\rho)\leq S_{max}(\rho)-S(\rho)$. 

    To show that the minimum is attained for a rank $r$ density operator $\rho$ when $\mathcal F_R(\rho)= S_{max}(\rho)-S(\rho)$, suppose on the contrary that there exists $\sigma=\Pi_k/k$, where $\Pi_k$ is a projection operator of rank $k$ such that the first term is less than $\log r$ and making $\mathcal{F}_R(\rho)<S_{\rm max}(\rho)-S(\rho)$. Let $M=\rho\log\sigma$; in the diagonal basis of $\sigma$, where we denote the diagonal element by $\lambda_i$, one has  
    
  \begin{align}
  \Tr M=\sum_i M_{ii} = \sum_i \rho_{ij}\delta_{ji}\log\lambda_i = \sum_i \rho_{ii}\log\lambda_i.
  \end{align}
Let us note that when $i>k$, $\log\lambda_{i>k}=-\infty$. Then in order for the trace of $M=\rho\log\sigma$ to be finite, we need $\rho_{ii}=0$ for any $\lambda_i\ne 0$ to be $0$.
    On the other hand, we know that if the diagonal of a positive semidefinite matrix has a zero on the diagonal, then the corresponding rows and columns must be all $0$s. Then it implies that up to rearranging the rows and columns for the sake of clarity, 
    \begin{align}
    \rho = \begin{pmatrix}
        A & 0\\
        0 & 0
    \end{pmatrix}
    \end{align}
    where $A$ is a $k\times k$ block matrix of rank at most $k$. Therefore, $-\Tr M = \log k\, \Tr[AI] = \log k<\log r$ by assumption, we must have $k<r$. Because $\dim A\geq \rank(A)$, it follows that $\rank(A)=\rank(\rho)\leq k<r$, which is a contradiction. 
\end{proof}
\section{Proof of~\cref{th:magicdist}}\label{app:proofthmd}
In this section, we prove~\cref{th:magicdist}. Let us start from the upper bound. Being defined through two minima, we can arbitrarily choose a state $\psi$ and a stabilizer $\sigma$ to upper bound $M_{\text{dist}}^{(NL)}$. Consider the state $\ket{\psi_{AB}}$, whose Schmidt decomposition can be written as $\ket{\psi_{AB}}=\sum_{i}^{D}\lambda_i\ket{\lambda_i^{A}\lambda_i^{B}}$, where $\lambda_i$ are the Schmidt coefficients and $D$ its Schmidt rank. Due to the minimization over $U=U_{A}\otimes U_B$, the basis $\ket{\lambda_{i}^{A/B}}$ can be brought in the computational basis (or another complete stabilizer basis)
\be 
U\ket{\lambda_{i}^{A/B}}=\ket{s_i^{A/B}}
\ee
Then, we choose $\ket{\sigma}=\sum_{i}^{d^{\lfloor \log_d D \rfloor}}\frac{1}{d^{\lfloor \log_d D \rfloor/2}} \ket{s_i^As_i^B}$, where $\lfloor \cdot \rfloor$ labels the floor function. Let us then compute the upper bound to $M_{\text{dist}}^{(NL)}(\psi_{AB})$. 
\begin{equation}
\begin{split}
M_{\text{dist}}^{(NL)}(\psi_{AB})\le &\frac{1}{2}\left \Vert \sum_{ij}^{D}\lambda_i\lambda_j\ketbra{s_i^A s_i^B}{s_j^A s_j^B}\right.\\
&\left.-\sum_{ij}^{d^{\lfloor \log D \rfloor}}d^{-\lfloor \log_d D \rfloor}\ketbra{s_i^A s_i^B}{s_j^A s_j^B}\right \Vert\\ 
=\Big(1-\sum_{ij}^{D}\sum_{kl}^{d^{\lfloor \log D \rfloor}}&d^{-\lfloor \log D \rfloor}\lambda_{i}\lambda_{j}\braket{s_l^A s_l^B}{s_i^A s_i^B}\braket{s_j^A s_j^B}{s_k^As_k^B}\Big)^{\frac{1}{2}}
\end{split}
\end{equation}
where we first used the Fuchs-Van der Graaf inequality, where the equality comes by $\psi_{AB}$ and $\sigma$ being pure states, and then rewritten the states in their Schmidt decomposition. Now without loss of generality, since $d^{\lfloor \log_d D \rfloor}\le D$, due to our degrees of freedom in the choice of $\sigma$ and $\psi_{AB}$ we can order the basis states such that only the first $d^{\lfloor \log_d D \rfloor}$ have nonzero overlap, and so it follows: 
\begin{equation}
\begin{aligned}
M_{\text{dist}}^{(NL)}(\psi_{AB})&\le \sqrt{1-\sum_{ij}^{d^{\lfloor \log_d D \rfloor}}d^{-\lfloor \log_d D \rfloor}\lambda_{i}\lambda_{j}}\\ 
&\le \sqrt{1-\sum_{ij}^{d^{\lfloor \log_d D \rfloor}}D^{-1}\lambda_{i}\lambda_{j}}\\
&\le \sqrt{1-\sum_{i}^{d^{\lfloor \log_d D \rfloor}}D^{-1}\lambda_{i}^2}\\ 
&=\sqrt{1-\frac{1}{D}+\sum_{i=d^{\lfloor \log D \rfloor}}^{D}D^{-1}\lambda_{i}^2}\\ 
&=\sqrt{1-\frac{1}{D}+\lambda_{\max}^{2}
\left(1-\frac{d^{\lfloor \log_d D \rfloor}}{D}\right)}\\ 
=\Big(1-e^{S_{max}(A)}\Big.&+e^{S_{\infty}(A)}\left(1-\frac{e^{\log d \lfloor S_{max}(A)/\log d \rfloor}}{e^{S_{max}(A)}}\right)\Big.\Big)^{1/2}
\end{aligned}
\end{equation}
where we first utilized the inequality $D^{-1}\leq d^{-\log_d D}$, and since $\lambda_i>0$ by definition, we can upper bound $M_{\text{dist}}^{(NL)}(\psi_{AB})$ by simply considering the diagonal terms. Next, we employed $\sum_{i}^{D}\lambda_i^2=1$ to rewrite our inequality, and finally, we utilized $S_{\max}(A)=\log{D}$ and $S_{\infty}=\log \lambda_{\text{max}}^2$.

Let us now focus on the lower bound. To prove it let us first provide a bound between $\mathcal{F}(\psi)$ and $ M_{dist}(\psi)$.
\begin{lemma}\label{lemmaflatmdist}
Let $\psi$ be a state then its flatness $\mathcal{F}(\psi)$ is upper bounded by $M_{\text{dist}}$ as follows
\begin{equation}
    \mathcal{F}(\psi)\le 8M_{\text{dist}}(\psi).
\end{equation}
\begin{proof}
Starting from the flatness one can add a zero term to it; take a flat state $\sigma\in\stab_0$
\begin{equation}
  \mathcal{F}(\psi)=\mathcal{F}(\psi)-\mathcal{F}(\sigma).
  \label{eq:flatsigma}
\end{equation}
We can then bound the flatness as follows:
\begin{equation}
\begin{split}
  \mathcal{F}(\psi)&=\Tr(\psi^3-\sigma^3)-\Tr\left((\psi^2)^{\otimes 2}-(\sigma^2)^{\otimes 2}\right)\\
                     &=\left| \Tr(\psi^3-\sigma^3) \right| + \left| \Tr\left((\psi^2)^{\otimes 2}-(\sigma^2)^{\otimes 2}\right) \right| \\
                     &\le \left| \Tr(\psi^3-\sigma^3) \right|+ 2 \left|\Tr\left((\psi^2)-(\sigma^2)\right) \right|\\
                     &\le 1- (1-T)^3 +2 - 2(1-T)^2\\
                     &\le T^3+ 7T\le 8T
\end{split}
\end{equation}
where $T=1/2 \left\| \psi -\sigma \right\|_1 $. In the second line we made use of the triangular inequality, in the third line, we used the following inequality
\begin{equation}
\begin{split}
    &|\Tr(\psi^2)\Tr(\psi^2)-\Tr(\sigma^2)\Tr(\sigma^2)|\\
    \le&|\Tr(\psi^2)\left(\Tr(\psi^2)-\Tr(\sigma^2)\right)|+|\left(\Tr(\psi^2)-\Tr(\sigma^2)\right)\Tr(\sigma^2)|\\
    \le& 2|\Tr(\psi^2)-\Tr(\sigma^2)|
\end{split}
\end{equation}
while in the fourth line we used \cite[ Lemma 1.2]{chen_sharp_2016} and then $T^3\le T$, since $0\le T\le 1$. 
By minimizing over $\sigma\in\stab$ we prove the lower bound with $M_{dist}(\psi)$. 
\end{proof}
\end{lemma}
Using Lemma~\ref{lemmaflatmdist}, we can thus write,  
\begin{equation}
    \mathcal{F}(\psi_A)\le 8 \min_{U_A}M_{\text{dist}}(U_A\psi_A U_A^{\dag})\label{cscsc}
\end{equation}
where $\psi_A=\tr_B{\psi_{AB}}$ and we used that $\mathcal{F}(\psi_A)$  is invariant under the action of global unitaries. 
Now, let us show that $ \min_{U_A}M_{\text{dist}}(U_A\psi_A U_A^{\dag})\le M_{\text{dist}}^{(NL)}(\psi_{AB})$. First recall that given $\psi_A=\Tr_B(\psi_{AB}) $ due to the monotonicity of $M_{\text{dist}}$ one has $M_{\text{dist}}(\psi_A)\le M_{\text{dist}}(\psi_{AB})$. Now let us prove the statement by contradiction. First, let $U_A$ be the unitary attaining the minimum in Eq.~\eqref{cscsc}. Let us suppose that there exists a bipartite unitary $U\equiv V_A\otimes V_B$ obeying
\begin{equation}
    M_{\text{dist}}(U_A \psi_A U_A^{\dag})> M_{\text{dist}}(U \psi_{AB} U^{\dag})
\end{equation}
Then we have the following chain of inequalities
\begin{equation}
\begin{split}
    M_{\text{dist}}(U_A \psi_A U_A^{\dag})&> M_{\text{dist}}(U \psi_{AB} U^{\dag})\\
    &\ge  M_{\text{dist}}(\tr_B U \psi_{AB} U^{\dag}) =  M_{\text{dist}}(V_A \psi_{A} V_A^{\dag}) 
\end{split}
\end{equation}
and this is a contradiction to the statement that $U_A$ attains the minimum. Therefore, one obtains that $ \min_{U_A}M_{\text{dist}}(U_A\psi_A U_A^{\dag})\le M_{\text{dist}}^{(NL)}(\psi_{AB})$. This result combined with ~\cref{lemmaflatmdist} concludes the proof.

\section{Stabilizer relative entropies}\label{app:nlentropies}
\subsection{Proof of~\cref{th:relstab}}\label{proofth2}
Let us start by proving the upper bound to $M_{RS}^{(NL)}$. We choose $\sigma_{AB}=\bbbone_{d^{\lceil \log_d  D  \rceil}}/d^{\lceil\log_d  D  \rceil}\oplus 0_{n-\lceil \log_d D  \rceil}$ where $D$ is the Schmidt rank of $\rho_A$.  Then expanding the relative entropy expansion one obtains the following bound
\begin{equation}
\begin{split}
M_{RS}^{(NL)}(\psi_{AB})&\le-\Tr[\psi_{AB} \log\sigma_{AB}] \\
&= \Tr[\sum_{i,j=1}^D\lambda_i\lambda_j |s_i\rangle\langle s_j| \times \\
&\qquad \sum_{k=1}^{d^{\lceil\log_d  D  \rceil}} |s_k\rangle\langle s_k|\log d \lceil \log_d D  \rceil] \\
&= \lceil \log_d  D  \rceil \log d\,   \Tr[\psi_{AB}]\\
&=\log d \lceil   S_{max}(A)/\log d  \rceil,
\end{split}
\end{equation}
%
Concluding the proof for the upper bound. Shifting our focus on the lower bound instead, let us note that for any $\rho$, $M_{RS}(\rho)\geq \mathcal{F}(\rho)$. This is a simple consequence of $\stab_0^{(n)}\subset \mathrm{FLAT}^{(n)}$. Because $\mathcal{F}(\rho)$ is isospectral under any unitary conjugation, it must follow that $M_{RS}(U \rho U^{\dagger})\geq F(U\rho U^{\dagger})=F(\rho)$. Therefore, 
\begin{align}
\mathcal F_R(\rho_A)\leq \min_{U_A} M_{RS}(U_A\rho_A U_A^{\dagger}).
\end{align}

    On the other hand, for any $\rho_A$, from monotonicity it follows that $M_{RS}(\rho_{AB})\geq M_{RS}(\rho_A)$ where $\rho_A=\Tr_B[\rho_{AB}]$. Therefore, we must have 
    \begin{equation}
    \begin{split}
       \min_{U_A} M_{RS}(U_A\rho_A U_A^{\dagger})&\leq \min_{U=V_A\otimes V_B}M_{RS}(U\rho_{AB}U^{\dagger}) \\
       &\equiv M_{RS}^{(NL)}(\rho_{AB}).
    \end{split}
    \end{equation}
    We can see that this is true from a proof by contradiction. Suppose there exists some $U_{A},U$ which attains the respective minima but has $$M_{RS}(U_A\rho_A U_A^{\dagger})> M_{RS}(U\rho_{AB} U^{\dagger}),$$ then from monotonicity, we must have 
    \begin{equation}
    \begin{split}
    M_{RS}(U_A\rho_A U_A^{\dagger})&>M_{RS}(U\rho_{AB}U^{\dagger})\\
    &\geq M_{RS}(\Tr_B [U\rho_{AB} U^{\dagger}]) \\
    &= M_{RS}(V_A\rho_A V_A^{\dagger})
    \end{split}
    \end{equation}
    for some local unitary $V_A$ which yields a lower distance than $U_A$. Since we assumed that $U_A$ attains the minimum, this violates our assumption, concluding the proof for the lower bound. 

    \subsection{Proof of~\cref{prop:RelativeEnt}}\label{proofprop4}

To bound the nonlocal magic of a pure state $\rho_{AB}$, consider a pure state $\psi_{AB}= U \rho_{AB}U^{\dagger}$ where $U=U_A\otimes U_B$ and $\psi_A$, $\psi_B$ are isospectral (up to truncation of 0 eigenvalues) to that of a subsystem $\rho_A,\rho_B$. Suppose they are states where we have removed the local magic such that both $\psi_A,\psi_B$ are diagonal in the computational (or another complete stabilizer basis). Again, this can be done by first rewriting the state $\rho_{AB}$ in the Schmidt basis, which is orthonormal. Then we replace the Schmidt basis with an orthonormal stabilizer basis to get $\psi_{AB}$. Since the mixture of stabilizer states is in the convex hull of stabilizer states, each $\psi_A,\psi_B$ must have zero local magic. Note that there are also other basis choices such that the basis state need not be a stabilizer, such states can also be in the convex hull of the stabilizer group as long as they are not pure states. 

By definition, $M^{(NL)}_{R}(\psi_{AB})\leq M_{R}(\psi_{AB})$ because we have chosen a particular instance of the local unitary $U_A\otimes U_B$ on the right-hand side whereas the left hand side is minimized over all possible instances.  Now we evaluate the relative entropy of magic $M_R(\psi_{AB})=-S(\psi_{AB})-\min_{\sigma \in \mathrm{STAB}}\Tr[\psi_{AB}\log \sigma_{AB}]$. Since $\psi_{AB}$ is pure, $S(\psi_{AB})=0$. If $\sigma_{AB}$ is pure, then the relative entropy is either $0$ when $\sigma=\psi$ or $\infty$ for any other $\sigma$ that is mixed.

We pick a stabilizer state $\sigma_{AB} = \sum_i\lambda_i^2 |s_i\rangle\langle s_i|_{AB}$ where $\lambda_i$ are the Schmidt coefficients of $\psi_{AB} = \sum_i\lambda_i |s_i\rangle_{AB}$ where $|s_i\rangle$ are the stabilizer basis we chose. Then for the second term, we write

\begin{align}
   M_R(\psi_{AB})&= -\Tr[\psi_{AB}\log\sigma_{AB}] \\
   &= -\Tr[\sum_{ij}\lambda_i\lambda_j |s_i\rangle\langle s_j|\sum_k \log(\lambda_k^2)|s_k\rangle\langle s_k|]\\
    &=-\sum_{i,j,k}\delta_{ij}\delta_{ik} \lambda_i\lambda_j \log(\lambda_k^2) \\
    &=-\sum_k p_k \log p_k
\end{align}
where we have set $\lambda_k^2=p_k$ because each Schmidt coefficient is real. $\delta_{ij}$ are Kronecker deltas because we have chosen the basis $\{|s_k\rangle\}$ to be orthonormal. Note that $\sum_k p_k=1$.
In this case, the second term is nothing but $S(A)=S(B)$ which is the von Neumann entropy of a subsystem. 

Since we have chosen a particular stabilizer state $\sigma_{AB}$, this serves as an upper bound of the relative entropy of magic. Hence 
\begin{align}
M^{(NL)}_{R}(\rho_{AB})\leq M^{(NL)}_{R}(\psi_{AB})\leq S(A)=S(B).
\end{align}

\subsection{Proof of~\cref{th:smoothed}}\label{proofth3}

    Let $\rho_{AB}^{\epsilon}$ represent the state that minimizes the nonlocal magic.  Therefore $\sma{\rho_{AB}}=M_{RS}^{(NL)}(\rho^{\epsilon}_{AB})$.  Drawing from \cref{th:relstab}, we understand that: 
\begin{equation}
\begin{split}
    \sma{\rho_{AB}}\geq & S_{max}(\rho^{\epsilon}_A)-S(\rho^{\epsilon}_A)\\
    \geq & \min_{\Vert\chi-\rho_A\Vert<\epsilon}\left(S_{max}(\chi)-S(\chi)\right).
\end{split}
\end{equation}

On the right-hand side, our goal is to identify a state $\chi$ within the $\epsilon$-ball of $\rho_A$ that minimizes the difference between $S_{max}(\chi)$ and $S(\chi)$. Interestingly, the state that minimizes this difference also reduces $S_{max}(\chi)$ to its lowest value $S_{max}^{\epsilon}$.  To illustrate, denote $\chi_A^{\epsilon}$ as the state that minimizes $S_{max}$ within the $\epsilon$-ball. Then consider increasing $S_{max}$ by modifying one eigenvalue of $\chi_A^{\epsilon}$ from zero to $\delta$.  This adjustment results in an increase  $\Delta S_{max}=e^{-S_{max}}$, while the change in entropy is capped at  $\Delta S\leq \delta\abs{\log\delta}$. Such a modification invariably elevates the entropy gap, i.e. $\Delta (S_{max}-S) \geq e^{-S_{max}}-\delta\abs{\log\delta}>0$, since $\delta$ can be arbitrarily small. 

To evaluate the von Neumann entropy of the state  $\chi_A^{\epsilon}$, as a modification from $\rho_A$ by dropping some eigenvalues whose total contribution to the trace is smaller than $\epsilon$. Let’s denote their contribution to the von Neumann entropy as  $S_{\epsilon}$. Then the entropy of the new state $\chi_A^{\epsilon}$ is given by $S(\chi_A^{\epsilon})= \frac{S(\rho_A)-S_{\epsilon}}{1-\epsilon}\leq \frac{S(\rho_A)}{1-\epsilon}$. Therefore, we get the following inequality:

\begin{equation}
    \begin{split}
         \sma{\rho_{AB}}\geq & S_{max}(\chi_A^{\epsilon})-S(\chi_A^{\epsilon})\\ 
        \geq & S_{max}^{\epsilon}(\rho_A)-(1-\epsilon)^{-1}S(\rho_A).
    \end{split}
\end{equation}

Regarding the upper bound, since $\chi_A^{\epsilon}$ minimizes the maximal entropy, it satisfies the following condition:
\begin{equation}
    S_{max}^{\epsilon}(\rho_A)=S_{max}(\chi_A^{\epsilon}).
\end{equation}

While this condition specifies the spectrum of $\chi_A^{\epsilon}$, we retain the flexibility to select a purification  $\chi_{AB}^{\epsilon}$, ensuring its deviation from $\rho_{AB}$ remains within an $\epsilon$ bound. Consequently, the process of minimizing the nonlocal magic leads us to the following inequality: 
\begin{equation}
\begin{split}
    \sma{\rho_{AB}}\leq & M_{RS}^{(NL)}(\chi_{AB}^{\epsilon})\\
    \leq &(\log d) \lceil \log_d\rank{\chi_{AB}^{\epsilon}} \rceil\\
    =&\log d \lceil   S_{max}^{\epsilon}(A)/\log d  \rceil.
\end{split}
\end{equation}
where the second step is a result from~\cref{th:relstab}. 


\section{Estimate by stabilizer R\'enyi entropy}\label{app:estimate}

\subsection{Proof of~\cref{thm:nlSRE}}

In this section, we provide an estimation of the second stabilizer R\'enyi entropy measure of the nonlocal magic. It is defined in \cite{stabrenyi} as the second R\'enyi entropy of a probability distribution, $p_a=\frac{1}{d}|\langle\psi|P_a|\psi\rangle|^2$, over all the Pauli-string basis $P_a$.
\begin{equation}
    \mathcal{M}_2(\ket{\psi}):=-\log(\sum_a p_a^2)-\log{d}.
\end{equation}

Given the entanglement spectrum $\{\lambda_i\}$, we construct a state $|\psi'\rangle$ with small local magic,   
\begin{equation}
    |\psi'\rangle_{AB} = \sum_{i=0}^{r-1}\sqrt{\lambda_i}|s_i\rangle_A|s_i\rangle_B. 
\end{equation}
where the rank $r$ is taken to be $2^n$ for integer $n$. The Pauli operators on the Hilbert space $\mathcal{H}_{AB}=\mathcal{H}_A\otimes\mathcal{H}_B$ can be  factorized as product of Pauli operators on $\mathcal{H}_A$ and $\mathcal{H}_B$ respectively, $P^{ab}=P^a\otimes P^b$. We denote their matrix elements as $P^{a,b}_{ij}:=\bra{s_i}P^{a,b}\ket{s_j}$, and compute the magic measure $\mathcal{M}_2$ as follows, 
\begin{equation}
\begin{split}
    &\mathcal{M}_2(\ket{\psi'})=-\log\left(\sum_{a=1}^{r^2}\sum_{b=1}^{r^2}\left|\sum_{i,j=0}^{r-1}\sqrt{\lambda_i}\sqrt{\lambda_j}P^a_{ij}P^b_{ij}\right|^4\right).
\end{split}
\end{equation}

The result is complicated and depends on specific choice of the basis $|s_i\rangle$'s. We simplify the analysis by assuming that the orthonormal basis $|s_i\rangle$'s are common eigenstates of a stabilizer group $\mathcal{S}=\{S_1,S_2,\cdots,S_n\}$.  This condition allows us to write the Pauli matrices $P^a_{ij}$ in computational basis. Substituting the matrix representation  of Pauli operators, we find that 

\begin{equation}\label{M2Estimate}
\begin{split}
    \mathcal{M}_2=&-\log\left(\sqrt{\lambda_{i_1}}\sqrt{\lambda_{i_2}}\sqrt{\lambda_{i_3}}\sqrt{\lambda_{i_4}}\sqrt{\lambda_{i_5}}\sqrt{\lambda_{i_6}}\sqrt{\lambda_{i_7}}\sqrt{\lambda_{i_8}}\right.\\
    &\left.\times (\sum_a P^a_{i_1i_2}P^a_{i_3i_4}P^a_{i_5i_6}P^a_{i_7i_8})^2\right)\\
    =&-\log\left(\sum_{i_1,i_2,i_3,i_4=0}^{r-1}\sqrt{\lambda_{i_1}\lambda_{i_2}\lambda_{i_3}\lambda_{i_4}\lambda_{i_3\wedge i_2\wedge i_1}}\right.\\
    &\left.\times \sqrt{\lambda_{i_4\wedge i_2\wedge i_1}\lambda_{i_1\wedge i_3\wedge i_4}\lambda_{i_2\wedge i_3\wedge i_4}}\right).
\end{split}
\end{equation}
where $\wedge$ denotes the bitwise $\XOR$ operation. This expression depends on order of eigenvalues and takes minimum when the eigenvalues are ordered, $\lambda_i>\lambda_j$ for  $i<j$. If we take all the eigenvalues to be the same, then each term in the summation is equal to $\frac{1}{r^4}$. The number of terms is $r^4$ since we are summing over four indices. The argument is equal to 1 in this case. Therefore, the nonlocal stabilizer R\'enyi entropy vanishes when the spectrum is flat. 

\cref{M2Comparison} gives a comparison of the direct SRE calculation against the estimation given by  \eqref{M2Estimate}. As can be observed in the plots, the approximation in  \eqref{M2Estimate} is correct up to numerical imprecision.
\begin{figure}[h]
\begin{center}
    \begin{subfigure}[b]{0.49\textwidth}
        \includegraphics[width=8cm]{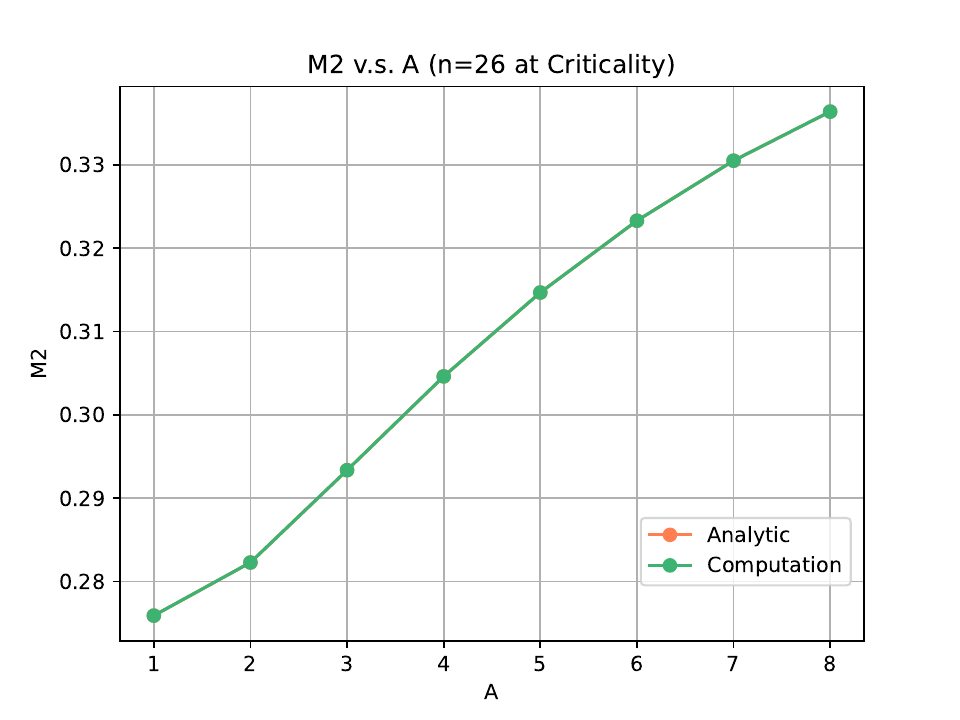}
        \caption{$\mathcal{M}_2$ computed from stabilizer R\'enyi entropy, compared to the estimation in \eqref{M2Estimate}.}
    \end{subfigure}
    \hfill
    \begin{subfigure}{0.49\textwidth}
        \includegraphics[width=8cm]{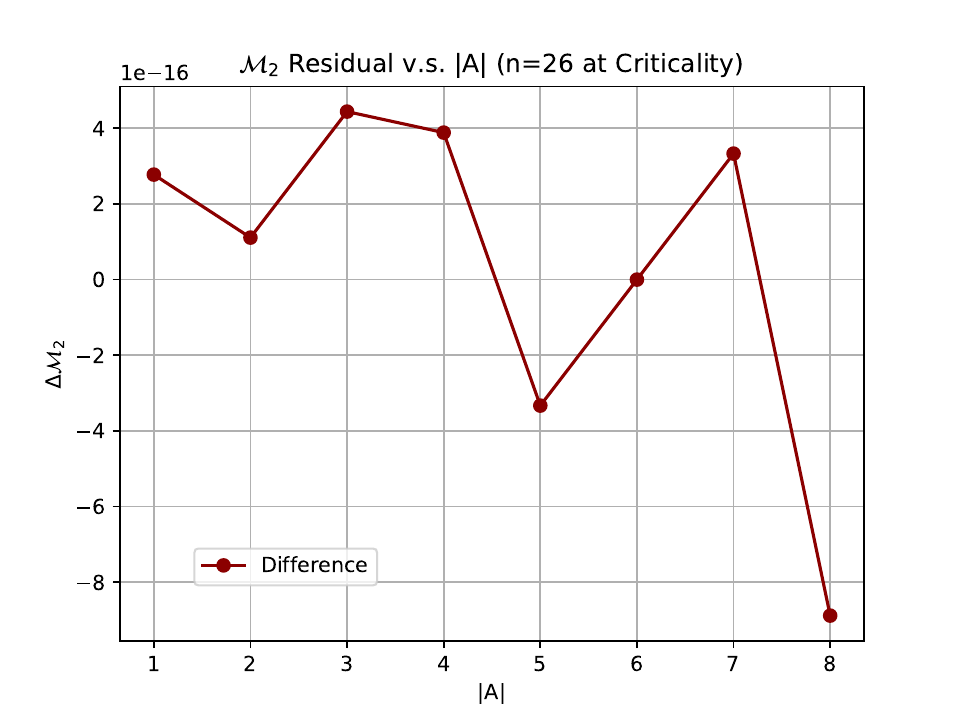}
    \caption{Residual from computational and analytic calculations of $\mathcal{M}_2$, accurate to one part in $10^{16}$.}
    \end{subfigure}
    \caption{}
    \label{M2Comparison}
\end{center}
\end{figure}
%


We can derive an upper bound for $\mathcal{M}_2$, by averaging over the permutations of eigenvalues, this gives us the expression,
\begin{equation}\label{M2UpperBound}
\begin{split}
    \mathcal{M}_2\leq\overline{\mathcal{M}}_2=-\log\left(\sum_{i=0}^{r-1}\lambda_i^4+7\sum_{0\leq i\neq j\leq r-1}\lambda_i^2\lambda_j^2\right.\\
    +\frac{7}{r-3}\sum_{0\leq i\neq j\neq k\neq l\leq r-1}\lambda_i\lambda_j\lambda_k\lambda_l\\
    \left. +\frac{\sum_{0\leq i_1\neq i_2\neq \cdots\neq i_8\leq r-1}\prod_{a=1}^8\sqrt{\lambda_{i_a}}}{(r-3)(r-5)(r-6)(r-7)}\right).
\end{split}
\end{equation}

In \eqref{M2UpperBound}, the sum inside the logarithm is taken over products of distinct eigenvalues. Computing this sum explicitly, and expressing the result in terms of different R\'enyi entropies $S_{\alpha}$, we obtain
\begin{equation}\label{eq:aveM}
\begin{split}
    \overline{\mathcal{M}}_2=&-\log \left(7e^{-2S_2}-6e^{-3S_4}
    +7e^{-S_0}(1-6e^{-S_2}\right.\\
    &+8e^{-2S_3}+3e^{-2S_2})+ e^{-4S_0}(e^{4S_{1/2}}+105e^{-3S_4}\\
    &\left.-420e^{S_{1/2}}+\cdots)\right)\\
    =& -\log\left(7e^{-2S_2}-6e^{-3S_4}+e^{4S_{1/2}-4S_0}\right)+O(e^{-S_{1/2}}).
\end{split}
\end{equation}
The averaged magic  $\overline{\mathcal{M}}_2$ is a complicated combination of R\'enyi entropies, ranging from $S_{1/2}$ to $S_4$. However, in the large Hilbert dimension limit, where $S_{1/2}\gg 1$, the averaged magic $\overline{\mathcal{M}}_2$ simplifies to the final expression in \eqref{eq:aveM}. It provides a straightforward estimate of $\mathcal{M}_2$ based on a few Rényi entropy terms. 

To establish a rigorous bound for $\mathcal{M}_2$, we start with \eqref{M2UpperBound}, leading to:
\begin{equation}\label{eq:bound1}
\begin{split}
    \overline{\mathcal{M}}_2&\leq -\log\left(\sum_{i=0}^{r-1}\lambda_i^4+7\sum_{0\leq i\neq j\leq r-1}\lambda_i^2\lambda_j^2\right)\\
    &\leq -\log\left(\left(\sum_{i}\lambda_i^2\right)^2\right)=2S_2.
\end{split}
\end{equation}

This holds for all spectrum distributions. Rewriting \eqref{M2UpperBound} in terms of the entropy difference $\delta_{1/2}=S_0-S_{1/2}$, we obtain the following expansion; 
\begin{equation}
\begin{split}
     \overline{\mathcal{M}}_2=&-\log\left( e^{-4\delta_{1/2}}+e^{-S_0}(7+21e^{-4\delta_{1/2}}-28e^{-3\delta_{1/2}})\right.\\
     &\left.+O(e^{-2S_0})\right)\leq 4\delta_{1/2}.
\end{split}
\end{equation}

Note that the coefficient associated with $e^{-S_0}$ in the expansion remains non-negative for any value of  $\delta_{1/2}$ and vanishes when $\delta_{1/2}=0$. Verifying that these coefficients are non-negative for every order of  $e^{-S_0}$ supports the inequality. Finally, combining this with the previously established bound finishes our proof that:
\begin{equation}
    \mathcal{M}_2(\{\lambda_i\})\leq \overline{\mathcal{M}}_2\leq \min\{2S_2,4(S_0-S_{1/2})\}.
\end{equation}

\subsection{Proof of~\cref{pp:branebound}}\label{app:branebound}

Let $\ket{\phi}$ denotes an entangled pair of qubits, with the entanglement spectrum given by $\{\lambda,1-\lambda\}$. We show that the nonlocal stabilizer R\'enyi entropy $\mathcal{M}_2(\lambda)$ of $\ket{\phi}$ is bounded by the nonflatness $\partial_n\tilde{S}_n$. 

From \eqref{eq:analyticalM}, we find that $\mathcal{M}_2(\lambda)$ is equal to, 
\begin{equation}
    \mathcal{M}_2(\lambda)=-\log\left(1-4\lambda+20\lambda^2-32\lambda^3+16\lambda^4\right).
\end{equation}

By definition  \eqref{eqn:holononflat}, the nonflatness is 
\begin{equation}
    -\partial_n\tilde{S}_n=n\frac{\lambda^n(1-\lambda)^n\left(\log\frac{\lambda}{1-\lambda}\right)^2}{\left(\lambda^n+(1-\lambda)^n\right)^2}.
\end{equation}

Both functions are zero at $\lambda=0,\ \frac{1}{2},\ 1$. So let us make a Taylor expansion around these value. Around $\lambda=\frac{1}{2}$, we have that 

\begin{equation}
    \begin{split}
    \mathcal{M}_2(\lambda)&=4(\lambda-1/2)^2-8(\lambda-1/2)^4+O((\lambda-1/2)^5)\\
    -\partial_n\tilde{S}_n\vert_{n=1}&=4(\lambda-1/2)^2-\frac{16}{3}(\lambda-1/2)^4+O((\lambda-1/2)^5)\\
    -\frac{1}{2}\partial_n\tilde{S}_n\vert_{n=2}&=4(\lambda-1/2)^2-\frac{160}{3}(\lambda-1/2)^4+O((\lambda-1/2)^5).
\end{split}
\end{equation}

Therefore for $\lambda$ close to $1/2$, the following inequality holds:
\begin{equation}
   -\frac{1}{2}\partial_n\tilde{S}_n\vert_{n=2} \leq \mathcal{M}_2(\lambda)\leq -\partial_n\tilde{S}_n\vert_{n=1}.
\end{equation}

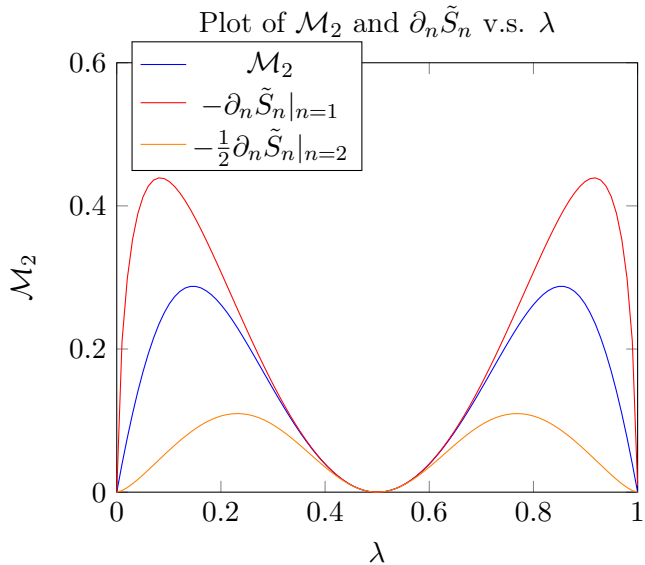
\begin{figure}
    \centering
    \begin{tikzpicture}
    \begin{axis}[
        title={Plot of $\mathcal{M}_2$ and $\partial_n\tilde{S}_n$ vs $\lambda$},
        xlabel={$\lambda$},
        ylabel={$\mathcal{M}_2$},
        xmin=0, xmax=1,
        ymin=0, ymax=0.6,
        legend style={at={(0.25,1.05)},anchor=north},
        grid style=dashed,
    ]
    
    \addplot[
        domain=0:1, 
        samples=100, 
        color=blue,
    ]{-ln(1-4*x+20*x^2-32*x^3+16*x^4)};
    \addlegendentry{$\mathcal{M}_2$}
    
    \addplot[
        domain=0:1, 
        samples=100, 
        color=red,
    ]{x*(1-x)*(ln(x/(1-x)))^2};
    \addlegendentry{$-\partial_n\tilde{S}_n\vert_{n=1}$}

    \addplot[
        domain=0:1, 
        samples=100, 
        color=orange,
    ]{x^2*(1-x)^2*(ln(x/(1-x)))^2/(x^2+(1-x)^2)^2};
    \addlegendentry{$-\frac{1}{2}\partial_n\tilde{S}_n\vert_{n=2}$}
    
    \end{axis}
    \end{tikzpicture}
    \caption{The nonlocal stabilizer R\'enyi entropy $\mathcal{M}_2$ is bounded by the antiflatness $\partial_n\tilde{S}_n$ }
    \label{fig:singlequbit}
\end{figure}

Similarly, one can show that this inequality holds for $\lambda$ close to $0$ and $1$, where the functions are,

\begin{equation}
\begin{split}
    \mathcal{M}_2(\lambda)&=4\lambda-12\lambda^2+O(\lambda^3)\\
    -\partial_n\tilde{S}_n\vert_{n=1}&=\lambda\log^2 \lambda+(2\log \lambda-\log^2 \lambda)\lambda^2+O(\lambda^3)\\
    -\frac{1}{2}\partial_n\tilde{S}_n\vert_{n=2}&=\lambda^2\log^2 \lambda+2(\log \lambda+\log^2 \lambda)\lambda^3+O(\lambda^4).
\end{split}
\end{equation}

For other value of $\lambda$, we justify this inequality by the plot in \cref{fig:singlequbit}.

Both the stabilizer R\'enyi entropy and  antiflatness are additive. Therefore for state $\ket{\psi}$ that can be distilled into product of entangled pairs $U_A\otimes U_B\ket{\psi}_{AB}=\otimes_{i=1}^k\ket{\phi}_{a_ib_i}$, we have,
\begin{equation}\label{eq:branebound}
    \frac{1}{2}\left\vert\frac{\partial_n\mathcal{A}_n|_{n=2}}{4G}(\ket{\psi}_{AB})\right\vert\leq\mathcal{M}_2(\ket{\psi}_{AB})\leq\left\vert\frac{\partial_n\mathcal{A}_n|_{n=1}}{4G}(\ket{\psi}_{AB})\right\vert.
\end{equation}

\subsection{Distillation of Matrix Product State}\label{app:MPS}
We further elaborate our discussions from~\cref{section:MERA}. Building on the MERA representation of CFT, we transform the state on the boundary of the past light-cone, $\partial A$, into a MPS using local unitaries, as defined in \eqref{eq:MPS} and illustrated in \cref{fig:MPS}. In this section, we further contend that this MPS state can approximately be distilled into a tensor product of entangled pairs:
\begin{equation}
    \ket{\chi}_{AB}\approx U_A\otimes U_B\left(\otimes_{i=1}^k \ket{\phi_i}_{a_ib_i}\right).
\end{equation}
where $k$ is the size of MPS state. it is clear that this approximation does not hold in general due to the disparity in the number of free parameters between the most general entanglement spectrum (contains $2^{k-1}$ parameters) and that of the tensor product of entangled pairs ($k$ parameters). However, for translationally invariant MPS states characterized by short correlation lengths, this approximation is valid.

\begin{figure}
    \centering
    \begin{tikzpicture}
    \begin{axis}[
        title={Plot of $\log\epsilon_k$ vs $k$},
        xlabel={$k$},
        ylabel={$\log{\epsilon_k}$},
        xmin=3, xmax=12,
        ymin=-6, ymax=0,
        xtick={4,5,6,7,8,9,10,11},
        ytick={-6,-5,-4,-3,-2,-1,0},
        legend pos=north east,
        grid={both},
    ]
    
    \addplot[
        color=red,
        mark=square,
        ]
        coordinates {
        (4, -2.1892564076870427)
        (5, -2.24431618487007)
        (6, -2.5383074265151158)
        (7, -2.645075401940822)
        (8, -2.9187712324178627)
        (9, -3.101092789211817)
        (10, -3.3242363405260273)
        (11, -3.506557897319982)
        };
        \addlegendentry{Identical EPR}
    
    \addplot[
        color=blue,
        mark=o,
        ]
        coordinates {
        (4, -3.3524072174927233)
        (5, -3.619353391465326)
        (6, -4.074541934925921)
        (7, -4.509860006183766)
        (8, -4.815891217303744)
        (9, -5.149897361429764)
        (10, -5.322610059117081)
        (11, -5.572754212249797)
        };
        \addlegendentry{Non-Identical EPR}
    
    \end{axis}
    \end{tikzpicture}
    \caption{Scaling of error with system size.}
    \label{fig:error}
\end{figure}
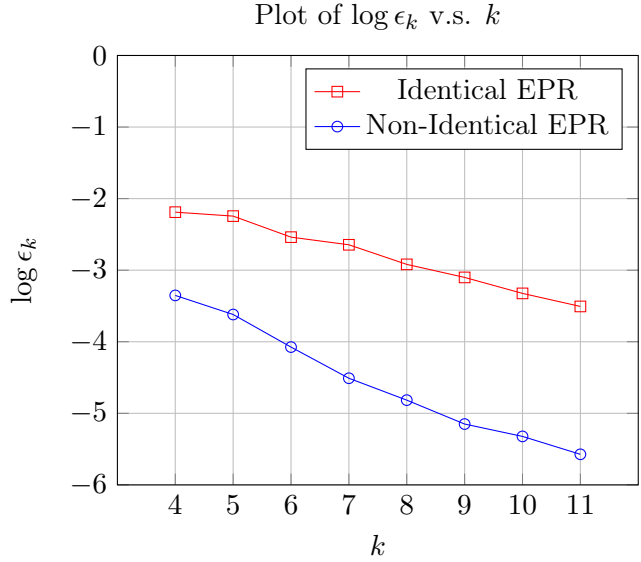

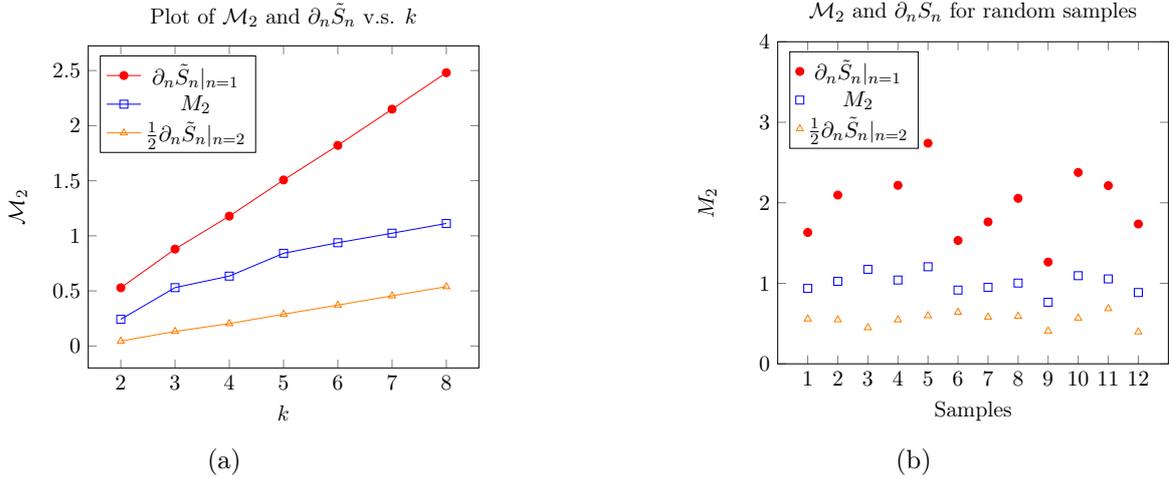
\begin{figure}
    \centering
    \begin{subfigure}[b]{0.4\textwidth}
           \scalebox{0.75}{
        \begin{tikzpicture}
        \begin{axis}[
            title={Plot of $\mathcal{M}_2$ and $\partial_n\tilde{S}_n$ vs $k$},
            xlabel={$k$},
            ylabel={$\mathcal{M}_2$},
            legend pos=north west,
            grid style=dashed,
        ]

        \addplot[
            color=red,
            mark=*,
            ]
            coordinates {
            (2, 0.5284476869093326)
            (3, 0.8796597152606531)
            (4, 1.1786243687213358)
            (5, 1.5067020734480323)
            (6, 1.8209208299973332)
            (7, 2.149153547543458)
            (8, 2.4804444445207428)
            };
            \addlegendentry{$\partial_n \tilde{S}_n\vert_{n=1}$}

            \addplot[
            color=blue,
            mark=square,
            ]
            coordinates {
            (2, 0.2420997928053625)
            (3, 0.5296291745870216)
            (4, 0.6341573657488669)
            (5, 0.8412419311548289)
            (6, 0.9369987551747572)
            (7, 1.0234907616064521)
            (8, 1.1120380681237716)
            };
            \addlegendentry{$M_2$}
        
        \addplot[
            color=orange,
            mark=triangle,
            ]
            coordinates {
            (2, 0.044403073081369274)
            (3, 0.13215785067590613)
            (4, 0.2031747844453768)
            (5, 0.28858671321308377)
            (6, 0.3708788831847463)
            (7, 0.4552538910164904)
            (8, 0.5382145816013133)
            };
            \addlegendentry{$\frac{1}{2}\partial_n \tilde{S}_n\vert_{n=2}$}
        \end{axis}
        \end{tikzpicture}
           }
        \caption{}
    \end{subfigure}
    \hfill
    \begin{subfigure}[b]{0.4\textwidth}
        \centering
        \scalebox{0.75}{
        \begin{tikzpicture}
            \begin{axis}[
                title={$\mathcal{M}_2$ and $\partial_nS_n$ for random samples},
                xlabel={Samples},
                ylabel={$M_2$},
                xmin=0, xmax=13,
                ymin=0, 
                ymax=4, 
                xtick={1,2,...,12},
                ytick={0,1,...,4}, 
                legend pos=north west,
                grid style=dashed,
                scatter/classes={
                    a={mark=*,red},
                    b={mark=square,blue},
                    c={mark=triangle,orange}
                }
            ]
            
            \addplot[scatter,only marks,scatter src=explicit symbolic]
                coordinates {
                (1,1.6313001268796445)[a]
                (2,2.0960085432431907)[a]
                (3,2.8261617899455858)[a]
                (4,2.21701795524457)[a]
                (5,2.740026056362322)[a]
                (6,1.5321860145108637)[a]
                (7,1.7622913869001764)[a]
                (8,2.055347541498999)[a]
                (9,1.2638202464353943)[a]
                (10,2.3765282087280477)[a]
                (11,2.2131910133599826)[a]
                (12,1.7362921659012711)[a]
                };
                \addlegendentry{$\partial_n \tilde{S}_n\vert_{n=1}$}

            \addplot[scatter,only marks,scatter src=explicit symbolic]
                coordinates {
                (1,0.9368407085714446)[b]
                (2,1.0243860301630408)[b]
                (3,1.1735551691638577)[b]
                (4,1.0407553593874668)[b]
                (5,1.2056303937669703)[b]
                (6,0.9159318631290625)[b]
                (7,0.9496265212961039)[b]
                (8,1.0030008626048081)[b]
                (9,0.7647095619498042)[b]
                (10,1.0945398160220088)[b]
                (11,1.0546323613045863)[b]
                (12,0.8853456721851094)[b]
                };
                \addlegendentry{$M_2$}
            

            \addplot[scatter,only marks,scatter src=explicit symbolic]
                coordinates {
                (1,0.5567483087609976)[c]
                (2,0.54726892014064)[c]
                (3,0.44819896725805314)[c]
                (4,0.5455860132789314)[c]
                (5,0.597299095385845)[c]
                (6,0.6400977858963219)[c]
                (7,0.5802832928523738)[c]
                (8,0.5899668363746525)[c]
                (9,0.40702575591041695)[c]
                (10,0.5677481440401517)[c]
                (11,0.6851334147660781)[c]
                (12,0.39684244389082163)[c]
                };
                \addlegendentry{$\frac{1}{2}\partial_n \tilde{S}_n\vert_{n=2}$}
            \end{axis}
            \end{tikzpicture}
        }
        \caption{}
    \end{subfigure}
    \caption{(a) Scaling of $\mathcal{M}_2$ and $\partial_n\tilde{S}_n$ with state size $k$ for a particular sample of random matrix in MPS. (b) For randomly generated samples of MPS states with a fixed size $k=7$, $\mathcal{M}_2$ is bounded by antiflatness $\partial_n\tilde{S}_n$.  }
    \label{fig:mpsmagic}
\end{figure}

To substantiate this approximation, we simulate several MPS states using $k$ number of identical random matrices to construct the reduced state $\rho_A=\tr_B(\ket{\chi}\bra{\chi})$ and evaluate its entanglement spectrum. We then approximate this spectrum by fitting it to the tensor product of individual entangled pair spectra:
\begin{equation}
    \min_{\{\lambda_i\}}\left\vert\mathrm{Spec}(\rho_A)- \bigotimes_{i=1}^k \begin{pmatrix}
        \lambda_i & 0\\
        0 & 1-\lambda_i
    \end{pmatrix}\right\vert=\epsilon_k
\end{equation}
where $\epsilon_k$ quantifies the approximation error. Our numerical analysis up to $k=11$ reveals an exponential decrease in $\epsilon_k$ with increasing $k$. We present two distinct scenarios in \cref{fig:error}:  In the first scenario, we require all EPR pairs in the tensor product to be identical, yielding an error trend of $\epsilon_k\sim 0.1 \times 1.2^{-k}$. In the second scenario, we relax this constraint, allowing for variability among the EPR pairs, which results in a more pronounced error reduction, following $\epsilon_k\sim 0.05 \times 1.4^{-k}$.

With the distillation assumption justified we expect the inequality~\eqref{eq:branebound} to be true for general MPS state and therefore for a CFT. We plot the magic and the R\'enyi entropy (dual to brane area) for a set of randomly generated samples of MPS states in \cref{fig:mpsmagic} and verify the validity of the bound \eqref{eq:branebound}.

\section{Validity of various bound for magic}\label{app:bound}
In the main text, we introduced several approximations for nonlocal magic, noting its proportional relationship to antiflatness in certain regimes and its closeness to entropy in others. This section delineates the conditions under which these approximations hold true.
\subsection*{Near-flat limit}
 We begin by examining the approximation between nonlocal magic and antiflatness, specifically
\begin{equation}
    \mathcal{M}_2(|\psi\rangle_{AB})\approx \frac{\mathcal{F}(\rho_A)}{\pur ^2(\rho_A)},
\end{equation}

which is applicable primarily in the near-flat limit of the entanglement spectrum. This is because the LHS is additive and scales linearly with $n$, while the RHS can be expressed as
\begin{equation}
    \frac{\mathcal{F}(\rho_A)}{\pur^2(\rho_A)}=\frac{\tr(\rho_A^3)-\tr(\rho_A^2)^2}{\tr(\rho_A^2)^2}=e^{2(S_2(A)-S_3(A))}-1,
\end{equation}
which becomes additive only at the linear order of the Taylor expansion in the entropy difference. Hence, the condition  $S_2(A)-S_3(A)\ll \frac{1}{2}$ must be met, indicating an almost flat spectrum or very weak entanglement.

Additionally, this regime aligns with where the two antiflatness measures defined previously converge, particularly when:
\begin{equation}\label{eq:E3}
    \langle(\delta\log\rho)^2\rangle_{\rho}\approx \frac{\langle(\delta\rho)^2\rangle_{\rho}}{\langle\rho\rangle_{\rho}^2}
\end{equation}
{\color{black}To demonstrate this approximation,  we rewrite the LHS in terms of the spectrum, 
\begin{equation}
\begin{split}
    \langle(\delta\log\rho)^2\rangle_{\rho}=&\sum_i p_i \left(\log\lambda_i-\overline{\log\lambda}\right)^2\\
\end{split}
\end{equation}
where $p_i=\lambda_i$ is the probability distribution given by the density matrix, and $\overline{x}=\sum_i p_i x$ is the average over this distribution. When the spectrum is nearly flat, we make the Taylor expansion, 
\begin{equation}
    \begin{split}
        \log\lambda_i=&\log(\bar \lambda+\delta \lambda_i)\\
        =&\log\bar \lambda+\frac{\delta \lambda_i}{\bar{\lambda}}+\frac{1}{2}(\frac{\delta \lambda_i}{\bar{\lambda}})^2+O((\frac{\delta{\lambda_i}}{\bar{\lambda}})^3),\\
        \overline{\log \lambda}=& \sum_ip_i \log(\bar{\lambda}+\delta \lambda_i)\\
        =&\log\bar \lambda+\frac{1}{2}\sum_i p_i(\frac{\delta \lambda_i}{\bar \lambda})^2+O(\frac{\overline{\delta \lambda^3}}{\bar \lambda^3}). 
    \end{split}
\end{equation}
Hence the LHS of \eqref{eq:E3} becomes 
\begin{equation}
\begin{split}
    \langle(\delta\log\rho)^2\rangle_{\rho}=&\sum_i p_i\left(\frac{\delta \lambda_i}{\bar \lambda}+\frac{1}{2}(\frac{\delta \lambda_i^2-\overline{\delta \lambda^2}}{\bar{\lambda}^2})+O((\frac{\delta{\lambda_i}}{\bar{\lambda}})^3)\right)^2\\
    =& \frac{\overline{\delta \lambda^2}}{\bar\lambda^2}+O(\frac{\overline{\delta \lambda^3}}{\bar \lambda^3})\\
    =& \frac{\langle(\delta\rho)^2\rangle_{\rho}}{\langle\rho\rangle_{\rho}^2}+O\left(\frac{\langle(\delta\rho)^3\rangle_{\rho}}{\langle\rho\rangle_{\rho}^3}\right)
    \end{split}
\end{equation}
Therefore this approximation is valid under the condition $\overline{\delta\lambda^3}\ll \bar{\lambda}\overline{\delta\lambda^2}$. This also corresponds to near-flat regime (small $\delta\lambda_i$) or very weak entanglement (large $\bar\lambda$). }

\subsection*{Far from flat limit}
In contrast, for quantum states with a far-from-flat entanglement spectrum, where the entropy differences across Rényi indices are comparable to the entropy itself, the scenario changes. Referring to Theorem \ref{thm:nlSRE}, the upper bound for the second stabilizer Rényi entropy measure of nonlocal magic is

\begin{equation}
\begin{split}
    \mathcal{M}_2^{NL}(|\psi\rangle_{AB})&\leq\mathcal{M}_2(\{\lambda_i\})\\
    &\leq\min\{2S_2(A),4(S_0(A)-S_{1/2}(A))\},
\end{split}
\end{equation} 
indicating a transitional crossover around $S_0(A)-S_{1/2}(A)\sim \frac{1}{2}S_2(A)$. Beyond this point, nonlocal magic transitions from being proportional to antiflatness to being proportional to entropy. Our numerical analyses within the Ising model confirm this transition: in the disordered phase and at critical points, nonlocal magic correlates with entropy $S$ both when varying the model parameter and the subsystem size. However, in the symmetry-breaking phase (refer to \cref{app:symmbreak}), it deviates and becomes anticorrelated with entropy, as shown in \figref{fig:MvsA}.

\section{Supplemental results for Ising Model}
\subsection{Symmetry-breaking phase}\label{app:symmbreak}

 In the $g<0$ regime, the Ising model enters the symmetry-breaking phase in the thermodynamic limit. However, our analysis is conducted on a finite-size lattice, where the ground state remains symmetric to spin flipping. Heuristically, we can think of this ground state being approximated by something similar to the GHZ state:

 \begin{equation}\label{eq:symG}
     \ket{G}_{sym}\approx\frac{1}{\sqrt{2}}(\ket{00\cdots0}+\ket{11\cdots1}).
 \end{equation}

 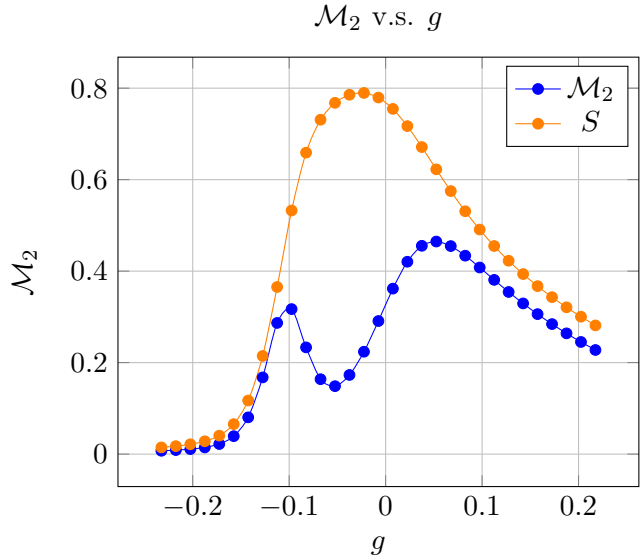
\begin{figure}
    \centering
    \begin{tikzpicture}
        \begin{axis}[
        xlabel={$g$},
        ylabel={$\mathcal{M}_2$},
        title={$\mathcal{M}_2$ vs $g$},
        grid={both}
    ]
    \addplot[color=blue,smooth,mark=*] coordinates {
        (-0.2324768366025517, 0.007223787236453515)
        (-0.2174768366025518, 0.008748730695846136)
        (-0.20247683660255178, 0.010926710694541172)
        (-0.18747683660255177, 0.01458468304470301)
        (-0.17247683660255175, 0.021989759114994827)
        (-0.15747683660255174, 0.03917264971677005)
        (-0.14247683660255173, 0.08035700169072464)
        (-0.1274768366025517, 0.16764680998652184)
        (-0.11247683660255181, 0.2867334739089084)
        (-0.09747683660255169, 0.3170054821476055)
        (-0.08247683660255178, 0.233293650708921)
        (-0.06747683660255177, 0.16375202920327683)
        (-0.05247683660255176, 0.14849097486137633)
        (-0.037476836602551744, 0.1730143829660695)
        (-0.02247683660255173, 0.2237686039807405)
        (-0.0074768366025518285, 0.2908969759535281)
        (0.007523163397448296, 0.3617439850129772)
        (0.022523163397448198, 0.42049487003948655)
        (0.03752316339744832, 0.45535876881380544)
        (0.052523163397448225, 0.4646698074563559)
        (0.06752316339744824, 0.4547522488064993)
        (0.08252316339744825, 0.43376417359303615)
        (0.09752316339744815, 0.40791779310947524)
        (0.11252316339744828, 0.38087234049479696)
        (0.12752316339744818, 0.3544456797479232)
        (0.1425231633974483, 0.32940432557212895)
        (0.1575231633974482, 0.30598887725715823)
        (0.17252316339744822, 0.28420249005417625)
        (0.18752316339744823, 0.2639536107301466)
        (0.20252316339744825, 0.24512150642745928)
        (0.21752316339744826, 0.22758407535400627)
    };    
    \addlegendentry{$\mathcal{M}_2$}
    \addplot[color=orange,mark=*,smooth] coordinates {
        (-0.2324768366025517, 0.01461178390556179)
        (-0.2174768366025518, 0.017416796403754247)
        (-0.20247683660255178, 0.0214297514406032)
        (-0.18747683660255177, 0.028014722238020845)
        (-0.17247683660255175, 0.04036083757912194)
        (-0.15747683660255174, 0.06549872277150799)
        (-0.14247683660255173, 0.11707699227960004)
        (-0.1274768366025517, 0.21456086879545458)
        (-0.11247683660255181, 0.3653251204975156)
        (-0.09747683660255169, 0.5325930140867744)
        (-0.08247683660255178, 0.6589732880629487)
        (-0.06747683660255177, 0.7310116119917656)
        (-0.05247683660255176, 0.7679080628052476)
        (-0.037476836602551744, 0.7855963300861846)
        (-0.02247683660255173, 0.7895360222061816)
        (-0.0074768366025518285, 0.7794869557462782)
        (0.007523163397448296, 0.75469489883342)
        (0.022523163397448198, 0.7170592299429965)
        (0.03752316339744832, 0.6712356927065845)
        (0.052523163397448225, 0.6225427177920325)
        (0.06752316339744824, 0.5750215914538588)
        (0.08252316339744825, 0.5308729056604539)
        (0.09752316339744815, 0.49086029922811086)
        (0.11252316339744828, 0.45492904778342713)
        (0.12752316339744818, 0.422669944398889)
        (0.1425231633974483, 0.3935800090837694)
        (0.1575231633974482, 0.3671815763943674)
        (0.17252316339744822, 0.3430640729680953)
        (0.18752316339744823, 0.32088993615140493)
        (0.20252316339744825, 0.3003866225713615)
        (0.21752316339744826, 0.28133492698065415)
    };
    \addlegendentry{$S$}
    \end{axis}
    \end{tikzpicture}
    \caption{Plot of $\mathcal{M}_2$ vs $g$, at $b=10^{-4}$ and $|A|=9$. }
    \label{fig:valley}
\end{figure}

\begin{figure}
    \centering
    \begin{subfigure}[b]{0.4\textwidth}
    \scalebox{0.75}{
    \begin{tikzpicture}
    \begin{axis}[
    title={$\mathcal{M}_2$ vs $|A|$},
    xlabel={$|A|$},
    ylabel={$\mathcal{M}_2$},
    legend style={at={(0.9,1.2)},anchor=north},
    grid=both,
    minor tick num=1,
]

\addplot coordinates {
    (1, 0.2689501026534531)
    (2, 0.24634756602268076)
    (3, 0.23742530624036945)
    (4, 0.2341664183399669)
    (5, 0.23313755401323172)
    (6, 0.23297365188896818)
    (7, 0.233092371527223)
    (8, 0.23323569080131712)
    (9, 0.233293650708921)
};
\addlegendentry{$g=-0.09$}

\addplot coordinates {
    (1, 0.24384767003162575)
    (2, 0.19876334120774347)
    (3, 0.1795866756827645)
    (4, 0.17091478801827334)
    (5, 0.16686418543136147)
    (6, 0.16499089767821182)
    (7, 0.16416477899486343)
    (8, 0.16383662079327974)
    (9, 0.16375202920327683)
};
\addlegendentry{$g=-0.08$}

\addplot coordinates {
    (1, 0.23839783569239717)
    (2, 0.18955310403593092)
    (3, 0.16809913606613852)
    (4, 0.15791302724907258)
    (5, 0.15286686130812485)
    (6, 0.15036294406760037)
    (7, 0.14916089721150316)
    (8, 0.14863627875794155)
    (9, 0.14849097486137633)
};
\addlegendentry{$g=-0.06$}

\addplot coordinates {
    (1, 0.24587523917318302)
    (2, 0.20549824561143945)
    (3, 0.18798522386416197)
    (4, 0.17986867395697714)
    (5, 0.17599510113906647)
    (6, 0.17418589317601396)
    (7, 0.17339509465078115)
    (8, 0.17309029627445577)
    (9, 0.1730143829660695)
};
\addlegendentry{$g=-0.04$}

\addplot coordinates {
    (1, 0.25826482031576664)
    (2, 0.23387389829982486)
    (3, 0.22518745256634734)
    (4, 0.2224731577645139)
    (5, 0.22206669463778722)
    (6, 0.22250574062437792)
    (7, 0.22312413998687008)
    (8, 0.22359615874074287)
    (9, 0.2237686039807405)
};
\addlegendentry{$g=-0.02$}

\end{axis}
    \end{tikzpicture}
    }
    \caption{}
    \label{fig:MvsAa}
    \end{subfigure}
    \hfill
    \begin{subfigure}[b]{0.4\textwidth}
    \scalebox{0.75}{
        \begin{tikzpicture}
        \begin{axis}[
    title={$S$ vs $|A|$},
    xlabel={$|A|$},
    ylabel={$S$},
    legend style={at={(0.85,0.44)}, anchor=north},
    grid=both,
    minor tick num=1,
]

\addplot coordinates {
    (1, 0.4882591491810362)
    (2, 0.5732661692916436)
    (3, 0.6110127980169322)
    (4, 0.6316082158129683)
    (5, 0.643830287493819)
    (6, 0.6513098074317467)
    (7, 0.6558042836506663)
    (8, 0.6582136549531553)
    (9, 0.6589732880629487)
};
\addlegendentry{$g=-0.09$}

\addplot coordinates {
    (1, 0.5251676197278689)
    (2, 0.6218593298388037)
    (3, 0.667228502555712)
    (4, 0.6933234776238961)
    (5, 0.7095691714738573)
    (6, 0.7199287515551955)
    (7, 0.7263602411156028)
    (8, 0.7298868121301564)
    (9, 0.7310116119917656)
};
\addlegendentry{$g=-0.08$}

\addplot coordinates {
    (1, 0.5315614343935706)
    (2, 0.6363808878074139)
    (3, 0.6881586158913452)
    (4, 0.7193994273394969)
    (5, 0.7396787045457957)
    (6, 0.7530676361495676)
    (7, 0.7616056520333332)
    (8, 0.7663733432572958)
    (9, 0.7679080628052476)
};
\addlegendentry{$g=-0.06$}

\addplot coordinates {
    (1, 0.5226859724979404)
    (2, 0.6334301601541248)
    (3, 0.6906057714327959)
    (4, 0.7265118349936612)
    (5, 0.7506177919230987)
    (6, 0.7669719090196478)
    (7, 0.7776168549095465)
    (8, 0.7836431182868497)
    (9, 0.7855963300861846)
};
\addlegendentry{$g=-0.04$}

\addplot coordinates {
    (1, 0.5059548186719761)
    (2, 0.6207548641952213)
    (3, 0.6820164453523454)
    (4, 0.7216323203848032)
    (5, 0.748869817389466)
    (6, 0.7676981101667162)
    (7, 0.7801247556788512)
    (8, 0.7872244076630824)
    (9, 0.7895360222061816)
};
\addlegendentry{$g=-0.02$}

\end{axis}
    \end{tikzpicture}
    }
    \caption{}
    \label{fig:MvsAb}
    \end{subfigure}
    \caption{(a)nonlocal magic $\mathcal{M}_2$ vs $|A|$. (b) Entropy $S$ vs $|A|$, at $b=10^{-5}$ and $g>-0.1$. }
    \label{fig:MvsA}
\end{figure}
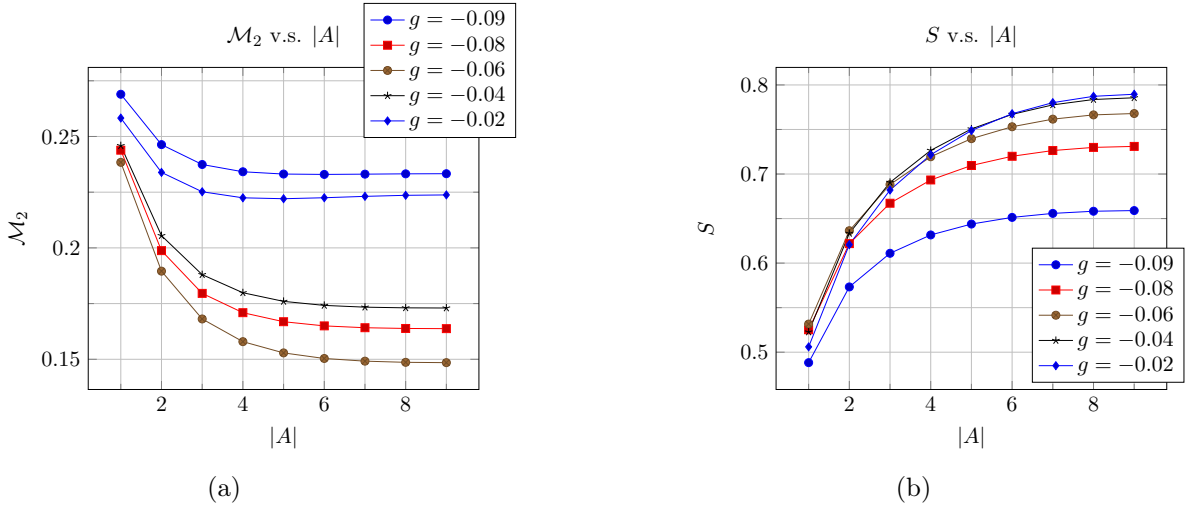

To approximate the true ground state achievable in the thermodynamic limit within our finite lattice model, we introduce a small bias field in the
$z$-direction: 
\ba
H=H_{\rm Ising}(g)+b\sum_i Z_i. 
\ea

As the bias $b$ increases,  the ground state transitions toward one of the two symmetry-broken states:
\ba
    &\ket{G}_{\uparrow}&=\ket{\uparrow\uparrow\cdots \uparrow}\\
    &\ket{G}_{\downarrow}&=\ket{\downarrow\downarrow\cdots \downarrow}.
\ea

Exploring how nonlocal magic $\mathcal{M}_2$ behaves as we adjust different parameters led to some fascinating results that are particularly noticeable when a nonzero bias field is applied. As shown in \cref{fig:valley}, a distinctive ``valley'' emerges in the  $\mathcal{M}_2$ plot within the $g<0$ regime.  We juxtapose entropy and nonlocal magic in our plots to underscore their divergent behaviors and the unique information conveyed by nonlocal magic.

This valley can be understood as arising from the competition between two types of ground states. Within the valley, the system's ground state approximates the symmetric GHZ-like state $\ket{G}_{sym}$, as defined in~\eqref{eq:symG}. In this region, nonlocal magic values are minimized because the reduced density matrix of  $\ket{G}_{sym}$ resembles that of a maximally mixed single-qubit state, leading to a flat spectrum and, consequently, lower $\mathcal{M}_2$ follows from \cref{th:srebound}. Additionally, we observe diminished $\mathcal{M}_2$ values in regions far from the critical point, where $|g|$ is sufficiently large, as indicated by the plateau beyond $g<-0.2$ in \cref{fig:valley}.  Here, the ground state transitions to a symmetry-broken state $\ket{G}_{\uparrow/\downarrow}$, which lacks nonlocal magic due to its tensor product structure.

Despite the low nonlocal magic values associated with both  $\ket{G}_{sym}$ and $\ket{G}_{\uparrow/\downarrow}$, the transition between these states has to past through a regime of nontrivial nonlocal magic. This occurs because continuous parameter changes cannot be approximated by discrete Clifford transformations, resulting in a notable increase in nonlocal magic.   The $\mathcal{M}_2$ measure captures this as a pronounced peak, delineating the transition between the two ground states near $g\sim -0.1$ in \cref{fig:valley}.  

An additional noteworthy aspect of nonlocal magic inside the valley is its counterintuitive decrease with increasing subregion size $|A|$, as depicted in  \cref{fig:MvsAa}.  This phenomenon is unique to the valley. In contrast, entropy consistently increases with $|A|$. This unusual trend in $\mathcal{M}_2$ is also linked to the proximity to the symmetric state $\ket{G}_{sym}$, which results in an almost flat entanglement spectrum within the valley.  Consequently, $\mathcal{M}_2$ aligns more closely with the entropy differential $S_0-S$  rather than the entropy itself, as discussed in \cref{section:estimate}, offering an explanation for the inverse relationship observed between $S$ and $\mathcal{M}_2$ in this region.  

it is also important to note that the competition between  $\ket{G}_{sym}$ and $\ket{G}_{\uparrow/\downarrow}$  is a manifestation of finite-size effects. As demonstrated in  \cref{fig:size},   the valley tends to diminish with increasing lattice size $n$. Specifically, when we set $b=10^{-4}$ (see \cref{fig:sizea}), the peak of nonlocal magic shifts closer to $g=0$ with larger lattice sizes. Similarly, with $g=-0.11$ (see \cref{fig:sizeb}), the peak moves toward $b=0$ as the lattice size expands. This suggests that the parameter space favoring the symmetric state narrows in both dimensions with increasing lattice size.


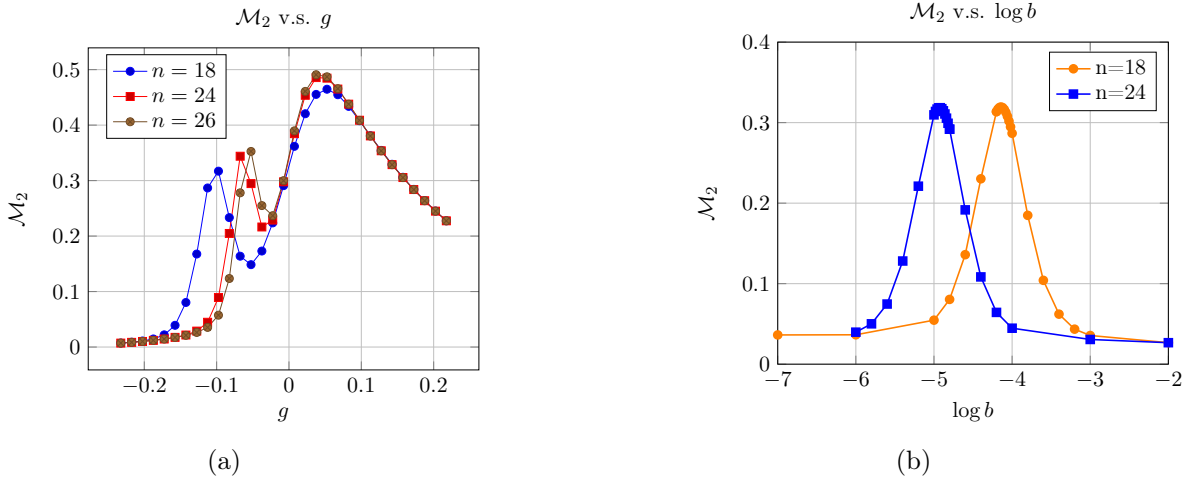
\begin{figure}
    \centering
    \begin{subfigure}[b]{0.4\textwidth}
        \scalebox{0.75}{
        \begin{tikzpicture}
            \begin{axis}[
        xlabel={$g$}, 
        ylabel={$\mathcal{M}_2$},
        title={$\mathcal{M}_2$ vs $g$},
        legend style={at={(0.2,0.98)},anchor=north},
        grid={both}
        ]
    \addplot coordinates {
        (-0.2324768366025517, 0.007223787236453515)
        (-0.2174768366025518, 0.008748730695846136)
        (-0.20247683660255178, 0.010926710694541172)
        (-0.18747683660255177, 0.01458468304470301)
        (-0.17247683660255175, 0.021989759114994827)
        (-0.15747683660255174, 0.03917264971677005)
        (-0.14247683660255173, 0.08035700169072464)
        (-0.1274768366025517, 0.16764680998652184)
        (-0.11247683660255181, 0.2867334739089084)
        (-0.09747683660255169, 0.3170054821476055)
        (-0.08247683660255178, 0.233293650708921)
        (-0.06747683660255177, 0.16375202920327683)
        (-0.05247683660255176, 0.14849097486137633)
        (-0.037476836602551744, 0.1730143829660695)
        (-0.02247683660255173, 0.2237686039807405)
        (-0.0074768366025518285, 0.2908969759535281)
        (0.007523163397448296, 0.3617439850129772)
        (0.022523163397448198, 0.42049487003948655)
        (0.03752316339744832, 0.45535876881380544)
        (0.052523163397448225, 0.4646698074563559)
        (0.06752316339744824, 0.4547522488064993)
        (0.08252316339744825, 0.43376417359303615)
        (0.09752316339744815, 0.40791779310947524)
        (0.11252316339744828, 0.38087234049479696)
        (0.12752316339744818, 0.3544456797479232)
        (0.1425231633974483, 0.32940432557212895)
        (0.1575231633974482, 0.30598887725715823)
        (0.17252316339744822, 0.28420249005417625)
        (0.18752316339744823, 0.2639536107301466)
        (0.20252316339744825, 0.24512150642745928)
        (0.21752316339744826, 0.22758407535400627)
    };
    \addlegendentry{$n=18$}
    \addplot coordinates {
        (-0.2324768366025517, 0.007138323500646095)
        (-0.2174768366025518, 0.008478137420016225)
        (-0.20247683660255178, 0.010087538081067002)
        (-0.18747683660255177, 0.012036072476847123)
        (-0.17247683660255175, 0.014429732020143615)
        (-0.15747683660255174, 0.01747103174356087)
        (-0.14247683660255173, 0.021682284953683427)
        (-0.1274768366025517, 0.02875667639269592)
        (-0.11247683660255181, 0.04468008032379104)
        (-0.09747683660255169, 0.08927084067997512)
        (-0.08247683660255178, 0.20474773189791878)
        (-0.06747683660255177, 0.343935624133068)
        (-0.05247683660255176, 0.29458702358256)
        (-0.037476836602551744, 0.21639555782368375)
        (-0.02247683660255173, 0.2288576762236713)
        (-0.0074768366025518285, 0.29780205322798725)
        (0.007523163397448296, 0.3845718675314944)
        (0.022523163397448198, 0.4537579542896862)
        (0.03752316339744832, 0.4856460489390489)
        (0.052523163397448225, 0.4846369370554632)
        (0.06752316339744824, 0.4649908048348451)
        (0.08252316339744825, 0.43779072468821545)
        (0.09752316339744815, 0.4088117511954463)
        (0.11252316339744828, 0.38049969574410225)
        (0.12752316339744818, 0.35372882542412704)
        (0.1425231633974483, 0.3287227825545529)
        (0.1575231633974482, 0.30545833795560945)
        (0.17252316339744822, 0.2838286043166396)
        (0.18752316339744823, 0.26370581053604863)
        (0.20252316339744825, 0.2449641774289902)
        (0.21752316339744826, 0.22748741229171957)
    };
    \addlegendentry{$n=24$}
    \addplot coordinates {
        (-0.2324768366025517, 0.007138227912831164)
        (-0.2174768366025518, 0.008477694356282846)
        (-0.20247683660255178, 0.010085541655881018)
        (-0.18747683660255177, 0.012027317911429997)
        (-0.17247683660255175, 0.014392344398262721)
        (-0.15747683660255174, 0.01731547711723815)
        (-0.14247683660255173, 0.02105222614873403)
        (-0.1274768366025517, 0.026281070467837425)
        (-0.11247683660255181, 0.03536996261958408)
        (-0.09747683660255169, 0.05742683034275903)
        (-0.08247683660255178, 0.12352601687391683)
        (-0.06747683660255177, 0.2779640900307487)
        (-0.05247683660255176, 0.3525557793174458)
        (-0.037476836602551744, 0.25496987082659067)
        (-0.02247683660255173, 0.2367824080235811)
        (-0.0074768366025518285, 0.29989716620332707)
        (0.007523163397448296, 0.38971078542149157)
        (0.022523163397448198, 0.4607632643053709)
        (0.03752316339744832, 0.49084971817892564)
        (0.052523163397448225, 0.4871000147650576)
        (0.06752316339744824, 0.4656894698078343)
        (0.08252316339744825, 0.43774367409216186)
        (0.09752316339744815, 0.40856786113387106)
        (0.11252316339744828, 0.38026941416269877)
        (0.12752316339744818, 0.35356398349312906)
        (0.1425231633974483, 0.32861871489030337)
        (0.1575231633974482, 0.30539735413686264)
        (0.17252316339744822, 0.2837946474498429)
        (0.18752316339744823, 0.2636876197956447)
        (0.20252316339744825, 0.2449547347721947)
        (0.21752316339744826, 0.22748264215185654)
    };
    \addlegendentry{$n=26$}
    \end{axis}
        \end{tikzpicture}
        }
        \caption{}
        \label{fig:sizea}
    \end{subfigure}
    \hfill
    \begin{subfigure}[b]{0.4\textwidth}
        \scalebox{0.75}{
        \begin{tikzpicture}
\begin{axis}[
    title={$\mathcal{M}_2$ vs $\log b$},
    xlabel={$\log b$},
    ylabel={$\mathcal{M}_2$},
    xmin=-7, xmax=-2,
    ymin=0, ymax=0.4,
    legend pos=north east,
    grid={both},
    every axis plot/.append style={thick}
]

\addplot[
    color=orange,
    mark=*,
    ]
    coordinates {
    (-8, 0.035634345760970211646198095199)
    (-7, 0.03631833084924441)
    (-6, 0.03650516400246957)
    (-5.0, 0.05465961086097081)
    (-4.8, 0.08045085735915408)
    (-4.599999999999999, 0.13598974833468974)
    (-4.3999999999999995, 0.23023981318059344)
    (-4.2, 0.31365841803612904)
    (-4.199970640755866, 0.31366418738740776)
    (-4.179973576680279, 0.31694128531956933)
    (-4.159976512604692, 0.318873262956139)
    (-4.139979448529106, 0.31941765657095933)
    (-4.119982384453518, 0.3185608183049997)
    (-4.0999853203779315, 0.31631858715807953)
    (-4.079988256302346, 0.3127355487012461)
    (-4.059991192226759, 0.30788294744476635)
    (-4.039994128151172, 0.3018554297073823)
    (-4.019997064075586, 0.294766884738027)
    (-3.999999999999999, 0.2867457070623583)
    (-3.7999999999999994, 0.18484455356621748)
    (-3.599999999999999, 0.1041265640094637)
    (-3.3999999999999995, 0.062134224406794954)
    (-3.1999999999999997, 0.04349412266078626)
    (-2.9999999999999996, 0.03563434576097021)
    (-2, 0.02670719766200395)
    };
    \addlegendentry{n=18}

\addplot[
    color=blue,
    mark=square*,
    ]
    coordinates {
    (-5.999999999999999, 0.039754416471600426)
    (-5.799999999999999, 0.05011609228992938)
    (-5.6, 0.07469590800766968)
    (-5.4, 0.12813141910521825)
    (-5.199999999999999, 0.22107426054724227)
    (-5.0, 0.30964631505179396)
    (-4.979997073344623, 0.3138833533429611)
    (-4.959994146689245, 0.31681435786145246)
    (-4.9399912200338685, 0.31837814663203196)
    (-4.919988293378491, 0.3185412810675044)
    (-4.899985366723115, 0.3172996623181483)
    (-4.879982440067738, 0.3146787332854024)
    (-4.85997951341236, 0.31073226831681294)
    (-4.8399765867569835, 0.3055398559595673)
    (-4.819973660101606, 0.29920328721229517)
    (-4.8, 0.2918535976034281)
    (-4.79997073344623, 0.29184214223277216)
    (-4.599999999999999, 0.19165102401443654)
    (-4.3999999999999995, 0.10838340169117754)
    (-4.2, 0.06431979738533082)
    (-3.999999999999999, 0.04468314071476549)
    (-3, 0.030687732641793277)
    (-2, 0.026683527440224727)
    };
    \addlegendentry{n=24}
\end{axis}
\end{tikzpicture}
        }
        \caption{}
        \label{fig:sizeb}
    \end{subfigure}
    \caption{(a) Plot of $\mathcal{M}_2$ vs $g$, at $b=10^{-4}$.  The lattice size takes $n=18$, $n=24$ and $n=26$; (b) Plot of $\mathcal{M}_2$ vs $\log b$, at $g=-0.11$, with lattice size taking $n=18$ and $n=24$ }
    \label{fig:size}
\end{figure}



Expanding our analysis, \cref{fig:Mvsbrane} explores the nonlocal magic across a broader range of the bias field $b$. We find that beyond $b>0.01$ the valley disappears, and the  $g<0$ phase transitions to being governed by the symmetry-broken ground state $\ket{G}_{\uparrow/\downarrow}$. As $b$ decreases toward zero, the peak is pushed to the left where the valley widens, signifying the growing significance of the symmetric ground state $\ket{G}_{sym}$, which becomes dominant for all $g<0$ in the absence of $b$.



\figref{fig:NonLocalMagicSurface2} depicts the $\mathcal{M}_2$ surface as a function of subregion size $|A|$ and critical angle $\theta$. As we decrease the magnitude of the bias field, the symmetry-breaking peak is pushed toward lower and lower $\theta$ values.
\begin{figure}[H]
\begin{center}
	\includegraphics[width=9cm]{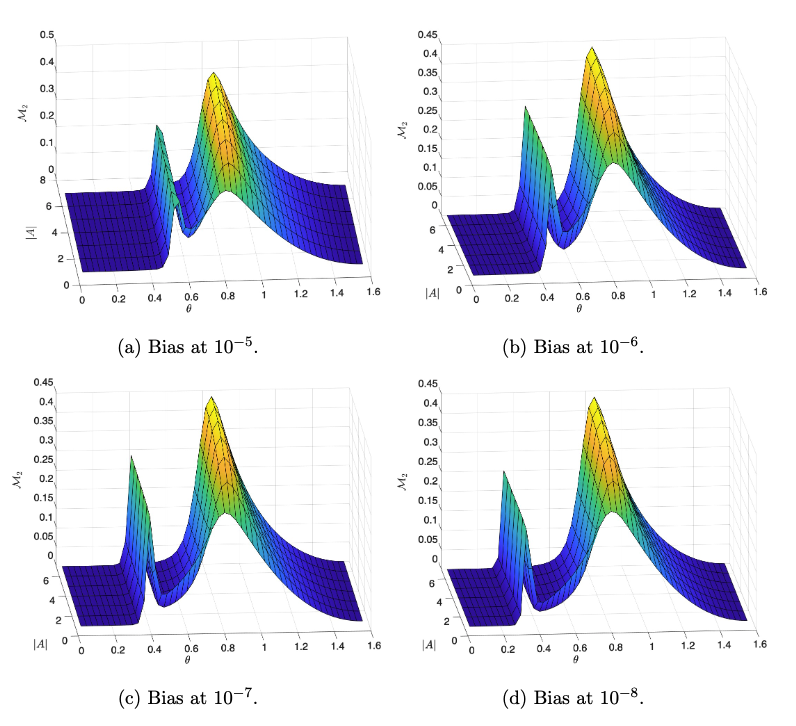}
    \caption{Nonlocal magic $\mathcal{M}_2$ surface, as a function of critical angle $\theta$ and subregion size $|A|$, for different bias offset fields. The bias magnetic field decreases, the peak indicating a symmetry-breaking effect in the system is pushed further away from criticality.}
    \label{fig:NonLocalMagicSurface2}
\end{center}
\end{figure}

\onecolumngrid
\begin{figure}[H]
    \centering
    \includegraphics[scale=0.55]{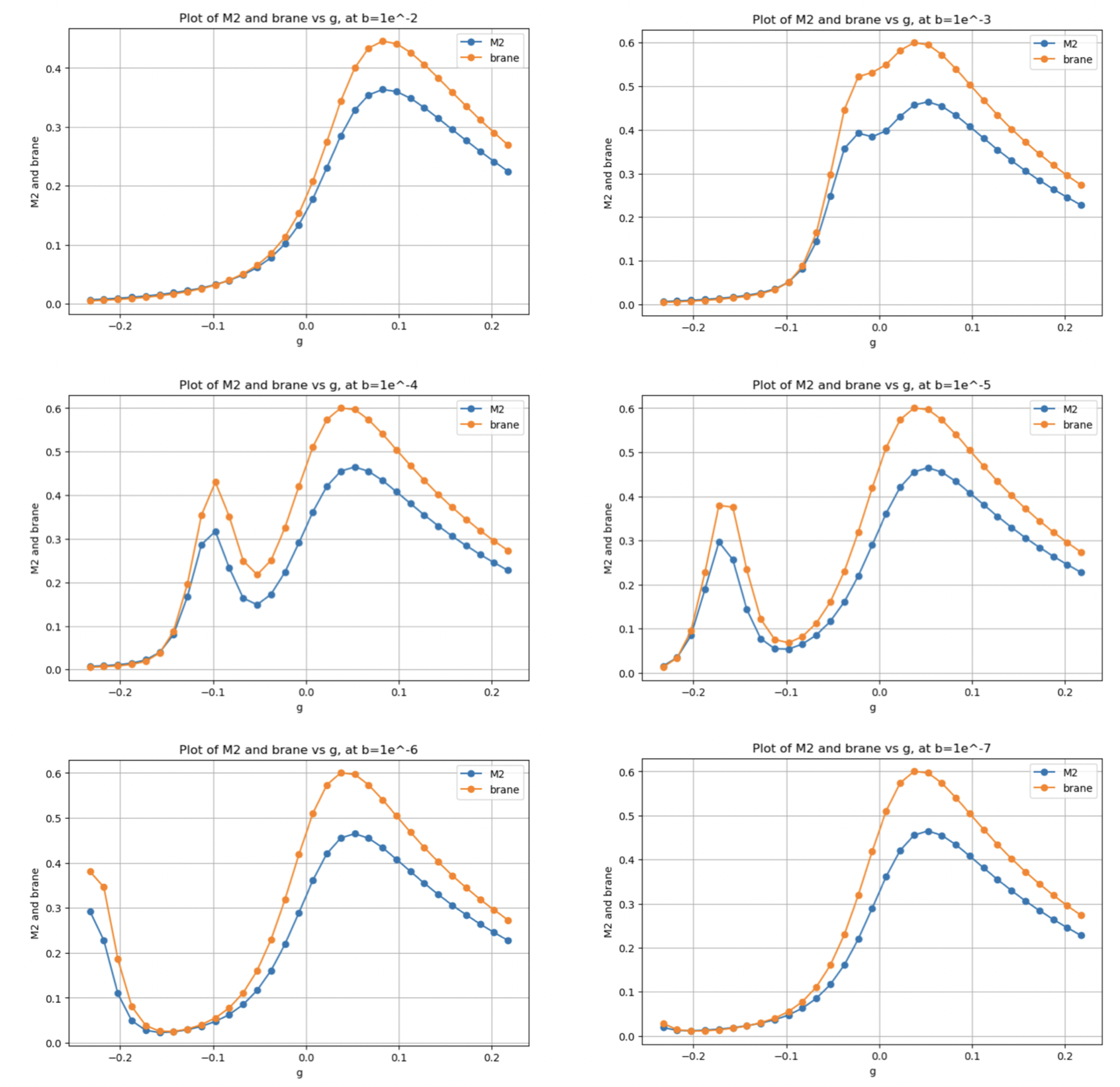}
    \caption{Comparison between $\mathcal{M}_2$ and $|\partial_n\tilde{S}_n|$ (labeled as brane) at various magnetic field $b$ and model parameter $g$. The small peak separates and is pushed to the left as $b$ decreases.}
    \label{fig:Mvsbrane}
\end{figure}

\twocolumngrid

\bibliographystyle{bibliost}
\bibliography{ref,bib} 

@article{Leone_2024_monotonicity,
   title={Stabilizer entropies are monotones for magic-state resource theory},
   volume={110},
   ISSN={2469-9934},
   url={http://dx.doi.org/10.1103/PhysRevA.110.L040403},
   DOI={10.1103/physreva.110.l040403},
   number={4},
   journal={Physical Review A},
   publisher={American Physical Society (APS)},
   author={Leone, Lorenzo and Bittel, Lennart},
   year={2024},
   month=oct }

@inproceedings{harrow_2010_entanglement,
   title={Entanglement spread and clean resource inequalities},
   url={http://dx.doi.org/10.1142/9789814304634_0046},
   DOI={10.1142/9789814304634_0046},
   booktitle={XVIth International Congress on Mathematical Physics},
   publisher={WORLD SCIENTIFIC},
   author={HARROW, ARAM W.},
   year={2010},
   month=mar }

@inproceedings{Coudron_2019_entanglement,
  doi = {10.4230/LIPICS.CCC.2019.20},
  
  url = {https://drops.dagstuhl.de/entities/document/10.4230/LIPIcs.CCC.2019.20},
  
  author = {Coudron, Matthew and Harrow, Aram W.},
  
  language = {en},
  
  title = {Universality of EPR Pairs in Entanglement-Assisted Communication Complexity, and the Communication Cost of State Conversion},
  
  publisher = {Schloss Dagstuhl – Leibniz-Zentrum für Informatik},
  
  year = {2019},
  
  copyright = {Creative Commons Attribution 3.0 Unported license}
}

@misc{harrow_2002_tight,
      title={A tight lower bound on the classical communication cost of entanglement dilution}, 
      author={Aram Harrow and Hoi-Kwong Lo},
      year={2002},
      eprint={quant-ph/0204096},
      archivePrefix={arXiv},
      primaryClass={quant-ph}
}

@article{Anshu_2022_entanglement,
   title={Entanglement spread area law in gapped ground states},
   volume={18},
   ISSN={1745-2481},
   url={http://dx.doi.org/10.1038/s41567-022-01740-7},
   DOI={10.1038/s41567-022-01740-7},
   number={11},
   journal={Nature Physics},
   publisher={Springer Science and Business Media LLC},
   author={Anshu, Anurag and Harrow, Aram W. and Soleimanifar, Mehdi},
   year={2022},
   month=sep, pages={1362–1366} }

@article{Bennett_2014_reverse,
   title={The Quantum Reverse Shannon Theorem and Resource Tradeoffs for Simulating Quantum Channels},
   volume={60},
   ISSN={1557-9654},
   url={http://dx.doi.org/10.1109/TIT.2014.2309968},
   DOI={10.1109/tit.2014.2309968},
   number={5},
   journal={IEEE Transactions on Information Theory},
   publisher={Institute of Electrical and Electronics Engineers (IEEE)},
   author={Bennett, Charles H. and Devetak, Igor and Harrow, Aram W. and Shor, Peter W. and Winter, Andreas},
   year={2014},
   month=may, pages={2926–2959} }

@article{aaronson_improved_2004,
  title = {Improved Simulation of Stabilizer Circuits},
  author = {Aaronson, Scott and Gottesman, Daniel},
  year = {2004},
  month = nov,
  journal = {Physical Review A},
  volume = {70},
  pages = {052328--052328},
  doi = {10.1103/PhysRevA.70.052328}
}

@article{Pollack_2022,
   title={Understanding holographic error correction via unique algebras and atomic examples},
   volume={2022},
   ISSN={1029-8479},
   url={http://dx.doi.org/10.1007/JHEP06(2022)056},
   DOI={10.1007/jhep06(2022)056},
   number={6},
   journal={Journal of High Energy Physics},
   publisher={Springer Science and Business Media LLC},
   author={Pollack, Jason and Rall, Patrick and Rocchetto, Andrea},
   year={2022},
   month=jun }

@article{bravyi_improved_2016,
  title = {Improved {{Classical Simulation}} of {{Quantum Circuits Dominated}} by {{Clifford Gates}}},
  author = {Bravyi, Sergey and Gosset, David},
  year = {2016},
  month = jun,
  journal = {Physical Review Letters},
  volume = {116},
  pages = {250501--250501},
  doi = {10.1103/PhysRevLett.116.250501}
}

@article{bravyi_simulation_2019,
  title = {Simulation of Quantum Circuits by Low-Rank Stabilizer Decompositions},
  author = {Bravyi, Sergey and Browne, Dan and Calpin, Padraic and Campbell, Earl and Gosset, David and Howard, Mark},
  year = {2019},
  journal = {Quantum},
  volume = {3},
  pages = {181--181},
  doi = {10.22331/q-2019-09-02-181}
}

@article{bravyi_trading_2016,
  title = {Trading {{Classical}} and {{Quantum Computational Resources}}},
  author = {Bravyi, Sergey and Smith, Graeme and Smolin, John A.},
  year = {2016},
  month = jun,
  journal = {Physical Review X},
  volume = {6},
  pages = {021043--021043},
  doi = {10.1103/PhysRevX.6.021043}
}

@article{bravyi_universal_2005,
  title = {Universal Quantum Computation with Ideal {{Clifford}} Gates and Noisy Ancillas},
  author = {Bravyi, Sergey and Kitaev, Alexei},
  year = {2005},
  month = feb,
  journal = {Physical Review A},
  volume = {71},
  pages = {022316--022316},
  doi = {10.1103/PhysRevA.71.022316}
}

@misc{bu_stabilizer_2023,
  title = {Stabilizer {{Testing}} and {{Magic Entropy}}},
  author = {Bu, Kaifeng and Gu, Weichen and Jaffe, Arthur},
  year = {2023},
  month = jun,
  number = {arXiv:2306.09292},
  eprint = {2306.09292},
  primaryclass = {math-ph, physics:quant-ph},
  publisher = {{arXiv}},
  doi = {10.48550/arXiv.2306.09292},
  urldate = {2023-07-12},
  archiveprefix = {arxiv},
  pages = {2306.09292},
  journaltitle = {}
}

@article{campbell_catalysis_2011,
  title = {Catalysis and Activation of Magic States in Fault-Tolerant Architectures},
  author = {Campbell, Earl T.},
  year = {2011},
  month = mar,
  journal = {Physical Review A},
  volume = {83},
  pages = {032317--032317},
  doi = {10.1103/PhysRevA.83.032317}
}

@article{jafarzadeh_randomized_2020,
  title = {Randomized Benchmarking for Qudit {{Clifford}} Gates},
  author = {Jafarzadeh, Mahnaz and Wu, Ya-Dong and Sanders, Yuval R. and Sanders, Barry C.},
  year = {2020},
  month = jun,
  journal = {New J. Phys.},
  volume = {22},
  number = {6},
  pages = {063014},
  publisher = {{IOP Publishing}},
  issn = {1367-2630},
  doi = {10.1088/1367-2630/ab8ab1},
  urldate = {2024-02-09},
  langid = {english}
}

@misc{tarabunga2023critical,
      title={Critical behaviours of non-stabilizerness in quantum spin chains}, 
      author={Poetri Sonya Tarabunga},
      year={2023},
      eprint={2309.00676},
      archivePrefix={arXiv},
      primaryClass={quant-ph}
}

@misc{tarabunga2024nonstabilizerness,
      title={Nonstabilizerness via matrix product states in the Pauli basis}, 
      author={Poetri Sonya Tarabunga and Emanuele Tirrito and Mari Carmen Bañuls and Marcello Dalmonte},
      year={2024},
      eprint={2401.16498},
      archivePrefix={arXiv},
      primaryClass={quant-ph}
}

@article{tarabunga2023manybody,
  title = {Many-Body Magic Via Pauli-Markov Chains---From Criticality to Gauge Theories},
  author = {Tarabunga, Poetri Sonya and Tirrito, Emanuele and Chanda, Titas and Dalmonte, Marcello},
  journal = {PRX Quantum},
  volume = {4},
  issue = {4},
  pages = {040317},
  numpages = {19},
  year = {2023},
  month = {Oct},
  publisher = {American Physical Society},
  doi = {10.1103/PRXQuantum.4.040317},
  url = {https://link.aps.org/doi/10.1103/PRXQuantum.4.040317}
}

@article{ge_area_2016,
  title = {Area Laws and Efficient Descriptions of Quantum Many-Body States},
  author = {Ge, Yimin and Eisert, Jens},
  year = {2016},
  journal = {New Journal of Physics},
  volume = {18},
  number = {8},
  pages = {083026--083026},
  doi = {10.1088/1367-2630/18/8/083026}
}

@misc{gottesman_heisenberg_1998,
  title = {The {{Heisenberg Representation}} of {{Quantum Computers}}},
  author = {Gottesman, Daniel},
  year = {1998},
  month = jul,
  number = {arXiv:quant-ph/9807006},
  eprint = {quant-ph/9807006},
  publisher = {{arXiv}},
  doi = {10.48550/arXiv.quant-ph/9807006},
  urldate = {2022-10-22},
  archiveprefix = {arxiv},
  pages = {quant-ph/9807006},
  journaltitle = {}
}

@article{hamma_bipartite_2005,
  title = {Bipartite Entanglement and Entropic Boundary Law in Lattice Spin Systems},
  author = {Hamma, Alioscia and Ionicioiu, Radu and Zanardi, Paolo},
  year = {2005},
  month = feb,
  journal = {Physical Review A},
  volume = {71},
  pages = {022315--022315},
  doi = {10.1103/PhysRevA.71.022315}
}

@misc{haug_efficient_2023,
  title = {Efficient Stabilizer Entropies for Quantum Computers},
  author = {Haug, Tobias and Lee, Soovin and Kim, M. S.},
  year = {2023},
  month = dec,
  number = {arXiv:2305.19152},
  eprint = {2305.19152},
  primaryclass = {quant-ph},
  publisher = {{arXiv}},
  doi = {10.48550/arXiv.2305.19152},
  urldate = {2024-01-10},
  archiveprefix = {arxiv},
  pages = {2305.19152},
  journaltitle = {}
}

@article{hebenstreit_computational_2020,
  title = {Computational Power of Matchgates with Supplementary Resources},
  author = {Hebenstreit, M. and Jozsa, R. and Kraus, B. and Strelchuk, S.},
  year = {2020},
  month = nov,
  journal = {Physical Review A},
  volume = {102},
  pages = {052604--052604},
  doi = {10.1103/PhysRevA.102.052604}
}

@article{hosur_chaos_2016,
  title = {Chaos in Quantum Channels},
  author = {Hosur, Pavan and Qi, Xiao-Liang and Roberts, Daniel A. and Yoshida, Beni},
  year = {2016},
  month = feb,
  journal = {Journal of High Energy Physics},
  volume = {2016},
  number = {2},
  pages = {4--4},
  doi = {10.1007/JHEP02(2016)004}
}

@article{Keeler:2022ajf,
    author = "Keeler, Cynthia and Munizzi, William and Pollack, Jason",
    title = {Entropic lens on stabilizer states},
    eprint = "2204.07593",
    archivePrefix = "arXiv",
    primaryClass = "quant-ph",
    doi = "10.1103/PhysRevA.106.062418",
    journal = "Phys. Rev. A",
    volume = "106",
    number = "6",
    pages = "062418",
    year = "2022"
}

@article{Keeler:2023xcx,
    author = "Keeler, Cynthia and Munizzi, William and Pollack, Jason",
    title = "{Clifford orbits from cayley graph quotients}",
    eprint = "2306.01043",
    archivePrefix = "arXiv",
    primaryClass = "quant-ph",
    doi = "10.26421/qic24.1-2-1",
    journal = "Quant. Inf. Comput.",
    volume = "24",
    pages = "1--36",
    year = "2024"
}

@article{hfhg-2z68,
  title = {Bounding entanglement entropy with Clifford double cosets},
  author = {Keeler, Cynthia and Munizzi, William and Pollack, Jason},
  journal = {Phys. Rev. A},
  volume = {112},
  issue = {2},
  pages = {022413},
  numpages = {23},
  year = {2025},
  month = {Aug},
  publisher = {American Physical Society},
  doi = {10.1103/hfhg-2z68},
  url = {https://link.aps.org/doi/10.1103/hfhg-2z68}
}

@article{kitaev_topological_2006,
  title = {Topological {{Entanglement Entropy}}},
  author = {Kitaev, Alexei and Preskill, John},
  year = {2006},
  month = mar,
  journal = {Physical Review Letters},
  volume = {96},
  pages = {110404--110404},
  doi = {10.1103/PhysRevLett.96.110404}
}

@article{leone_nonstabilizerness_2023,
  title = {Nonstabilizerness Determining the Hardness of Direct Fidelity Estimation},
  author = {Leone, Lorenzo and Oliviero, Salvatore F. E. and Hamma, Alioscia},
  year = {2023},
  month = feb,
  journal = {Phys. Rev. A},
  volume = {107},
  number = {2},
  pages = {022429},
  publisher = {{American Physical Society}},
  doi = {10.1103/PhysRevA.107.022429},
  urldate = {2023-04-21}
}

@article{leone_quantum_2021,
  title = {Quantum {{Chaos}} Is {{Quantum}}},
  author = {Leone, Lorenzo and Oliviero, Salvatore F. E. and Zhou, You and Hamma, Alioscia},
  year = {2021},
  month = may,
  journal = {Quantum},
  volume = {5},
  pages = {453--453},
  doi = {10.22331/q-2021-05-04-453}
}

@article{levin_detecting_2006,
  title = {Detecting {{Topological Order}} in a {{Ground State Wave Function}}},
  author = {Levin, Michael and Wen, Xiao-Gang},
  year = {2006},
  month = mar,
  journal = {Physical Review Letters},
  volume = {96},
  pages = {110405--110405},
  doi = {10.1103/PhysRevLett.96.110405}
}

@misc{bejan2023dynamical,
      title={Dynamical Magic Transitions in Monitored Clifford+T Circuits}, 
      author={Mircea Bejan and Campbell McLauchlan and Benjamin Béri},
      year={2023},
      eprint={2312.00132},
      archivePrefix={arXiv},
      primaryClass={quant-ph}
}

@article{Munizzi:2023ihc,
    author = "Munizzi, William and Schnitzer, Howard J.",
    title = {Entropy cones and entanglement evolution for Dicke states},
    eprint = "2306.13146",
    archivePrefix = "arXiv",
    primaryClass = "quant-ph",
    doi = "10.1103/PhysRevA.109.012405",
    journal = "Phys. Rev. A",
    volume = "109",
    number = "1",
    pages = "012405",
    year = "2024"
}

@article{nahum_operator_2018,
  title = {Operator {{Spreading}} in {{Random Unitary Circuits}}},
  author = {Nahum, Adam and Vijay, Sagar and Haah, Jeongwan},
  year = {2018},
  month = apr,
  journal = {Physical Review X},
  volume = {8},
  pages = {021014--021014},
  doi = {10.1103/PhysRevX.8.021014}
}

@book{nielsen_quantum_2000,
  title = {Quantum {{Computation}} and {{Quantum Information}}},
  author = {Nielsen, Michael A. and Chuang, Isaac L.},
  year = {2000},
  publisher = {{Cambridge University Press}}
}

@article{oliviero_transitions_2021,
  title = {Transitions in Entanglement Complexity in Random Quantum Circuits by Measurements},
  author = {Oliviero, Salvatore F. E. and Leone, Lorenzo and Hamma, Alioscia},
  year = {2021},
  journal = {Physics Letters A},
  volume = {418},
  pages = {127721--127721},
  doi = {10.1016/j.physleta.2021.127721}
}

@article{PhysRevLett.128.050402,
  title = {Stabilizer R\'enyi Entropy},
  author = {Leone, Lorenzo and Oliviero, Salvatore F. E. and Hamma, Alioscia},
  journal = {Phys. Rev. Lett.},
  volume = {128},
  issue = {5},
  pages = {050402},
  numpages = {5},
  year = {2022},
  month = {Feb},
  publisher = {American Physical Society},
  doi = {10.1103/PhysRevLett.128.050402},
  url = {https://link.aps.org/doi/10.1103/PhysRevLett.128.050402}
}

@misc{rattacaso_stabilizer_2023,
  title = {Stabilizer Entropy Dynamics after a Quantum Quench},
  author = {Rattacaso, Davide and Leone, Lorenzo and Oliviero, Salvatore F. E. and Hamma, Alioscia},
  year = {2023},
  month = may,
  number = {arXiv:2304.13768},
  eprint = {2304.13768},
  primaryclass = {quant-ph},
  publisher = {{arXiv}},
  doi = {10.48550/arXiv.2304.13768},
  urldate = {2023-06-27},
  archiveprefix = {arxiv},
  pages = {2304.13768},
  journaltitle = {}
}

@article{sarkar_characterization_2020,
  title = {Characterization of an Operational Quantum Resource in a Critical Many-Body System},
  author = {Sarkar, S. and Mukhopadhyay, C. and Bayat, A.},
  year = {2020},
  month = aug,
  journal = {New Journal of Physics},
  volume = {22},
  number = {8},
  pages = {083077--083077},
  doi = {10.1088/1367-2630/aba919}
}

@article{saxena_quantifying_2022,
  title = {Quantifying Multiqubit Magic Channels with Completely Stabilizer-Preserving Operations},
  author = {Saxena, Gaurav and Gour, Gilad},
  year = {2022},
  month = oct,
  journal = {Phys. Rev. A},
  volume = {106},
  number = {4},
  pages = {042422},
  publisher = {{American Physical Society}},
  doi = {10.1103/PhysRevA.106.042422},
  urldate = {2022-12-05}
}

@article{skinner_measurementinduced_2019,
  title = {Measurement-{{Induced Phase Transitions}} in the {{Dynamics}} of {{Entanglement}}},
  author = {Skinner, Brian and Ruhman, Jonathan and Nahum, Adam},
  year = {2019},
  month = jul,
  journal = {Physical Review X},
  volume = {9},
  pages = {031009--031009},
  doi = {10.1103/PhysRevX.9.031009}
}

@article{vonkeyserlingk_operator_2018,
  title = {Operator {{Hydrodynamics}}, {{OTOCs}}, and {{Entanglement Growth}} in {{Systems}} without {{Conservation Laws}}},
  author = {{von Keyserlingk}, Curt W. and Rakovszky, Tibor and Pollmann, Frank and Sondhi, S. L.},
  year = {2018},
  month = apr,
  journal = {Physical Review X},
  volume = {8},
  pages = {021013--021013},
  doi = {10.1103/PhysRevX.8.021013}
}

@article{weedbrook_gaussian_2012,
  title = {Gaussian Quantum Information},
  author = {Weedbrook, Christian and Pirandola, Stefano and {Garc{\'{\i}}a-Patr{\'o}n}, Ra{\'u}l and Cerf, Nicolas J. and Ralph, Timothy C. and Shapiro, Jeffrey H. and Lloyd, Seth},
  year = {2012},
  journal = {Reviews of Modern Physics},
  volume = {84},
  number = {2},
  pages = {621--621},
  doi = {10.1103/RevModPhys.84.621}
}

@article{white_conformal_2021,
  title = {Conformal Field Theories Are Magical},
  author = {White, Christopher David and Cao, ChunJun and Swingle, Brian},
  year = {2021},
  month = feb,
  journal = {Physical Review B},
  volume = {103},
  pages = {075145--075145},
  doi = {10.1103/PhysRevB.103.075145}
}

@article{Caputa:2014vaa,
    author = "Caputa, Pawel and Nozaki, Masahiro and Takayanagi, Tadashi",
    title = {Entanglement of local operators in large-N conformal field theories},
    eprint = "1405.5946",
    archivePrefix = "arXiv",
    primaryClass = "hep-th",
    reportNumber = "YITP-14-42, IPMU-14-0123",
    doi = "10.1093/ptep/ptu122",
    journal = "PTEP",
    volume = "2014",
    pages = "093B06",
    year = "2014"
}

@article{PhysRevLett.80.5239,
  title = {Mixed-State Entanglement and Distillation: Is there a ``Bound'' Entanglement in Nature?},
  author = {Horodecki, Micha\l{} and Horodecki, Pawe\l{} and Horodecki, Ryszard},
  journal = {Phys. Rev. Lett.},
  volume = {80},
  issue = {24},
  pages = {5239--5242},
  numpages = {0},
  year = {1998},
  month = {Jun},
  publisher = {American Physical Society},
  doi = {10.1103/PhysRevLett.80.5239},
  url = {https://link.aps.org/doi/10.1103/PhysRevLett.80.5239}
}

@misc{iannotti2025entanglementstabilizerentropiesrandom,
      title={Entanglement and Stabilizer entropies of random bipartite pure quantum states}, 
      author={Daniele Iannotti and Gianluca Esposito and Lorenzo Campos Venuti and Alioscia Hamma},
      year={2025},
      eprint={2501.19261},
      archivePrefix={arXiv},
      primaryClass={quant-ph},
      url={https://arxiv.org/abs/2501.19261}, 
}

@ARTICLE{HarlowRT,
       author = {{Harlow}, Daniel},
        title = {The Ryu-Takayanagi Formula from Quantum Error Correction},
      journal = {Communications in Mathematical Physics},
         year = 2017,
        month = sep,
       volume = {354},
       number = {3},
        pages = {865-912},
          doi = {10.1007/s00220-017-2904-z},
archivePrefix = {arXiv},
       eprint = {1607.03901},
 primaryClass = {hep-th},
       adsurl = {https://ui.adsabs.harvard.edu/abs/2017CMaPh.354..865H}
}

@article{Maldacena_2013,
   title={Cool horizons for entangled black holes},
   volume={61},
   ISSN={1521-3978},
   url={http://dx.doi.org/10.1002/prop.201300020},
   DOI={10.1002/prop.201300020},
   number={9},
   journal={Fortschritte der Physik},
   publisher={Wiley},
   author={Maldacena, J. and Susskind, L.},
   year={2013},
   month=aug, pages={781–811} }

@article{Akers:2018fow,
    author = {Akers, Chris and Rath, Pratik},
    title = {Holographic Renyi Entropy from Quantum Error Correction},
    eprint = "1811.05171",
    archivePrefix = "arXiv",
    primaryClass = "hep-th",
    doi = "10.1007/JHEP05(2019)052",
    journal = "JHEP",
    volume = "05",
    pages = "052",
    year = "2019"
}

@article{Ryu_2006,
   title={Holographic Derivation of Entanglement Entropy from the anti–de Sitter Space/Conformal Field Theory Correspondence},
   volume={96},
   ISSN={1079-7114},
   url={http://dx.doi.org/10.1103/PhysRevLett.96.181602},
   DOI={10.1103/physrevlett.96.181602},
   number={18},
   journal={Physical Review Letters},
   publisher={American Physical Society (APS)},
   author={Ryu, Shinsei and Takayanagi, Tadashi},
   year={2006},
   month=may }

@article{Dong_2019,
   title={Flat entanglement spectra in fixed-area states of quantum gravity},
   volume={2019},
   ISSN={1029-8479},
   url={http://dx.doi.org/10.1007/JHEP10(2019)240},
   DOI={10.1007/jhep10(2019)240},
   number={10},
   journal={Journal of High Energy Physics},
   publisher={Springer Science and Business Media LLC},
   author={Dong, Xi and Harlow, Daniel and Marolf, Donald},
   year={2019},
   month=oct }

@article{Akers:2019gcv,
    author = "Akers, Chris and Rath, Pratik",
    title = {Entanglement Wedge Cross Sections Require Tripartite Entanglement},
    eprint = "1911.07852",
    archivePrefix = "arXiv",
    primaryClass = "hep-th",
    doi = "10.1007/JHEP04(2020)208",
    journal = "JHEP",
    volume = "04",
    pages = "208",
    year = "2020"
}

@article{Hayden:2021gno,
    author = "Hayden, Patrick and Parrikar, Onkar and Sorce, Jonathan",
    title = {The Markov gap for geometric reflected entropy},
    eprint = "2107.00009",
    archivePrefix = "arXiv",
    primaryClass = "hep-th",
    doi = "10.1007/JHEP10(2021)047",
    journal = "JHEP",
    volume = "10",
    pages = "047",
    year = "2021"
}

@article{magicising,
  title = {Magic-state resource theory for the ground state of the transverse-field Ising model},
  author = {Oliviero, Salvatore F. E. and Leone, Lorenzo and Hamma, Alioscia},
  journal = {Phys. Rev. A},
  volume = {106},
  issue = {4},
  pages = {042426},
  numpages = {6},
  year = {2022},
  month = {Oct},
  publisher = {American Physical Society},
  doi = {10.1103/PhysRevA.106.042426},
  url = {https://link.aps.org/doi/10.1103/PhysRevA.106.042426}
}

@article{magicMC,
  title = {Estimating Outcome Probabilities of Quantum Circuits Using Quasiprobabilities},
  author = {Pashayan, Hakop and Wallman, Joel J. and Bartlett, Stephen D.},
  journal = {Phys. Rev. Lett.},
  volume = {115},
  issue = {7},
  pages = {070501},
  numpages = {5},
  year = {2015},
  month = {Aug},
  publisher = {American Physical Society},
  doi = {10.1103/PhysRevLett.115.070501},
  url = {https://link.aps.org/doi/10.1103/PhysRevLett.115.070501}
}

@article{Niroula:2023meg,
    author = "Niroula, Pradeep and White, Christopher David and Wang, Qingfeng and Johri, Sonika and Zhu, Daiwei and Monroe, Christopher and Noel, Crystal and Gullans, Michael J.",
    title = "{Phase transition in magic with random quantum circuits}",
    eprint = "2304.10481",
    archivePrefix = "arXiv",
    primaryClass = "quant-ph",
    doi = "10.1038/s41567-024-02637-3",
    journal = "Nature Phys.",
    volume = "20",
    number = "11",
    pages = "1786--1792",
    year = "2024"
}

@article{leone2023phase,
  title = {Phase transition in stabilizer entropy and efficient purity estimation},
  author = {Leone, Lorenzo and Oliviero, Salvatore F. E. and Esposito, Gianluca and Hamma, Alioscia},
  journal = {Phys. Rev. A},
  volume = {109},
  issue = {3},
  pages = {032403},
  numpages = {11},
  year = {2024},
  month = {Mar},
  publisher = {American Physical Society},
  doi = {10.1103/PhysRevA.109.032403},
  url = {https://link.aps.org/doi/10.1103/PhysRevA.109.032403}
}

@article{PhysRevA.106.042426,
  title = {Magic-state resource theory for the ground state of the transverse-field Ising model},
  author = {Oliviero, Salvatore F. E. and Leone, Lorenzo and Hamma, Alioscia},
  journal = {Phys. Rev. A},
  volume = {106},
  issue = {4},
  pages = {042426},
  numpages = {6},
  year = {2022},
  month = {Oct},
  publisher = {American Physical Society},
  doi = {10.1103/PhysRevA.106.042426},
  url = {https://link.aps.org/doi/10.1103/PhysRevA.106.042426}
}

@misc{stab0-inprep,
title={A complete discussion on the theory of nonstabilizerness},
author={Lorenzo Leone and Salvatore F. E. Oliviero and Gianluca Esposito and Alioscia Hamma},
year={In preparation}
}

@article{Pastawski_2015,
   title={Holographic quantum error-correcting codes: toy models for the bulk/boundary correspondence},
   volume={2015},
   ISSN={1029-8479},
   url={http://dx.doi.org/10.1007/JHEP06(2015)149},
   DOI={10.1007/jhep06(2015)149},
   number={6},
   journal={Journal of High Energy Physics},
   publisher={Springer Science and Business Media LLC},
   author={Pastawski, Fernando and Yoshida, Beni and Harlow, Daniel and Preskill, John},
   year={2015},
   month=jun }

@article{Hayden_2003,
   title={Communication cost of entanglement transformations},
   volume={67},
   ISSN={1094-1622},
   url={http://dx.doi.org/10.1103/PhysRevA.67.012326},
   DOI={10.1103/physreva.67.012326},
   number={1},
   journal={Physical Review A},
   publisher={American Physical Society (APS)},
   author={Hayden, Patrick and Winter, Andreas},
   year={2003},
   month=jan }

@misc{cao2023quantum,
      title={Quantum Lego Expansion Pack: Enumerators from Tensor Networks}, 
      author={ChunJun Cao and Michael J. Gullans and Brad Lackey and Zitao Wang},
      year={2023},
      eprint={2308.05152},
      archivePrefix={arXiv},
      primaryClass={quant-ph}
}

@article{Radon,
   title={Building bulk geometry from the tensor Radon transform},
   volume={2020},
   ISSN={1029-8479},
   url={http://dx.doi.org/10.1007/JHEP12(2020)033},
   DOI={10.1007/jhep12(2020)033},
   number={12},
   journal={Journal of High Energy Physics},
   publisher={Springer Science and Business Media LLC},
   author={Cao, ChunJun and Qi, Xiao-Liang and Swingle, Brian and Tang, Eugene},
   year={2020},
   month=dec }

@article{Bao:2019bib,
    author = "Bao, Ning and Cao, ChunJun and Fischetti, Sebastian and Keeler, Cynthia",
    title = {Towards Bulk Metric Reconstruction from Extremal Area Variations},
    eprint = "1904.04834",
    archivePrefix = "arXiv",
    primaryClass = "hep-th",
    doi = "10.1088/1361-6382/ab377f",
    journal = "Class. Quant. Grav.",
    volume = "36",
    number = "18",
    pages = "185002",
    year = "2019"
}

@article{Almheiri_2013,
   title={Black holes: complementarity or firewalls?},
   volume={2013},
   ISSN={1029-8479},
   url={http://dx.doi.org/10.1007/JHEP02(2013)062},
   DOI={10.1007/jhep02(2013)062},
   number={2},
   journal={Journal of High Energy Physics},
   publisher={Springer Science and Business Media LLC},
   author={Almheiri, Ahmed and Marolf, Donald and Polchinski, Joseph and Sully, James},
   year={2013},
   month=feb }

@article{Molina-Vilaplana:2012rmg,
    author = {Molina-Vilaplana, Javier},
    title = {Holographic Geometries of one-dimensional gapped quantum systems from Tensor Network States},
    eprint = {1210.6759},
    archivePrefix = {arXiv},
    primaryClass = {hep-th},
    doi = {10.1007/JHEP05(2013)024},
    journal = "JHEP",
    volume = "05",
    pages = "024",
    year = "2013"
}

@article{Jozsa_2008,
   title={Matchgates and classical simulation of quantum circuits},
   volume={464},
   ISSN={1471-2946},
   url={http://dx.doi.org/10.1098/rspa.2008.0189},
   DOI={10.1098/rspa.2008.0189},
   number={2100},
   journal={Proceedings of the Royal Society A: Mathematical, Physical and Engineering Sciences},
   publisher={The Royal Society},
   author={Jozsa, Richard and Miyake, Akimasa},
   year={2008},
   month=jul, pages={3089–3106} }

@article{VanRaamsdonk:2010pw,
    author = "Van Raamsdonk, Mark",
    title = {Building up spacetime with quantum entanglement},
    eprint = "1005.3035",
    archivePrefix = "arXiv",
    primaryClass = "hep-th",
    doi = "10.1142/S0218271810018529",
    journal = "Gen. Rel. Grav.",
    volume = "42",
    pages = "2323--2329",
    year = "2010"
}

@article{Czech_2014,
   title={Holographic reconstruction of general bulk surfaces},
   volume={2014},
   ISSN={1029-8479},
   url={http://dx.doi.org/10.1007/JHEP11(2014)015},
   DOI={10.1007/jhep11(2014)015},
   number={11},
   journal={Journal of High Energy Physics},
   publisher={Springer Science and Business Media LLC},
   author={Czech, Bartlomiej and Dong, Xi and Sully, James},
   year={2014},
   month=nov }

@ARTICLE{Czech_Lampros,
       author = {{Czech}, Bart{\l}omiej and {Lamprou}, Lampros},
        title = {Holographic definition of points and distances},
      journal = {Phys.Rev. D},
         year = 2014,
        month = nov,
       volume = {90},
       number = {10},
          eid = {106005},
        pages = {106005},
          doi = {10.1103/PhysRevD.90.106005},
       adsurl = {https://ui.adsabs.harvard.edu/abs/2014PhRvD..90j6005C}
}

@article{Cao_2017,
   title={Space from Hilbert space: Recovering geometry from bulk entanglement},
   volume={95},
   ISSN={2470-0029},
   url={http://dx.doi.org/10.1103/PhysRevD.95.024031},
   DOI={10.1103/physrevd.95.024031},
   number={2},
   journal={Physical Review D},
   publisher={American Physical Society (APS)},
   author={Cao, ChunJun and Carroll, Sean M. and Michalakis, Spyridon},
   year={2017},
   month=jan }

@article{Jacobson_2016,
   title={Entanglement Equilibrium and the Einstein Equation},
   volume={116},
   ISSN={1079-7114},
   url={http://dx.doi.org/10.1103/PhysRevLett.116.201101},
   DOI={10.1103/physrevlett.116.201101},
   number={20},
   journal={Physical Review Letters},
   publisher={American Physical Society (APS)},
   author={Jacobson, Ted},
   year={2016},
   month=may }

@article{HMERA,
  title = {Hyperinvariant multiscale entanglement renormalization ansatz: Approximate holographic error correction codes with power-law correlations},
  author = {Cao, ChunJun and Pollack, Jason and Wang, Yixu},
  journal = {Phys. Rev. D},
  volume = {105},
  issue = {2},
  pages = {026018},
  numpages = {16},
  year = {2022},
  month = {Jan},
  publisher = {American Physical Society},
  doi = {10.1103/PhysRevD.105.026018},
  url = {https://link.aps.org/doi/10.1103/PhysRevD.105.026018}
}

@misc{cheng2022random,
      title={Random tensor networks with nontrivial links}, 
      author={Newton Cheng and Cécilia Lancien and Geoff Penington and Michael Walter and Freek Witteveen},
      year={2022},
      eprint={2206.10482},
      archivePrefix={arXiv},
      primaryClass={quant-ph}
}

@misc{akers2024background,
      title={Background independent tensor networks}, 
      author={Chris Akers and Annie Y. Wei},
      year={2024},
      eprint={2402.05910},
      archivePrefix={arXiv},
      primaryClass={hep-th}
}

@misc{dong2023holographic,
      title={Holographic Tensor Networks with Bulk Gauge Symmetries}, 
      author={Xi Dong and Sean McBride and Wayne W. Weng},
      year={2023},
      eprint={2309.06436},
      archivePrefix={arXiv},
      primaryClass={hep-th}
}

@article{Cao_2018,
   title={Bulk entanglement gravity without a boundary: Towards finding Einstein’s equation in Hilbert space},
   volume={97},
   ISSN={2470-0029},
   url={http://dx.doi.org/10.1103/PhysRevD.97.086003},
   DOI={10.1103/physrevd.97.086003},
   number={8},
   journal={Physical Review D},
   publisher={American Physical Society (APS)},
   author={Cao, ChunJun and Carroll, Sean M.},
   year={2018},
   month=apr }

@article{Czech_2017,
   title={Equivalent equations of motion for gravity and entropy},
   volume={2017},
   ISSN={1029-8479},
   url={http://dx.doi.org/10.1007/JHEP02(2017)004},
   DOI={10.1007/jhep02(2017)004},
   number={2},
   journal={Journal of High Energy Physics},
   publisher={Springer Science and Business Media LLC},
   author={Czech, Bartlomiej and Lamprou, Lampros and McCandlish, Samuel and Mosk, Benjamin and Sully, James},
   year={2017},
   month=feb }

@article{Akers_2021,
   title={Leading order corrections to the quantum extremal surface prescription},
   volume={2021},
   ISSN={1029-8479},
   url={http://dx.doi.org/10.1007/JHEP04(2021)062},
   DOI={10.1007/jhep04(2021)062},
   number={4},
   journal={Journal of High Energy Physics},
   publisher={Springer Science and Business Media LLC},
   author={Akers, Chris and Penington, Geoff},
   year={2021},
   month=apr }

@article{flammia_topological_2009,
	title = {Topological {Entanglement} {Rényi} {Entropy} and {Reduced} {Density} {Matrix} {Structure}},
	volume = {103},
	url = {https://link.aps.org/doi/10.1103/PhysRevLett.103.261601},
	doi = {10.1103/PhysRevLett.103.261601},
	journal = {Physical Review Letters},
	author = {Flammia, Steven T. and Hamma, Alioscia and Hughes, Taylor L. and Wen, Xiao-Gang},
	month = dec,
	year = {2009},
	pages = {261601--261601},
}

@ARTICLE{avgrenyi,
       author = {{Kim}, MuSeong and {Hwang}, Mi-Ra and {Jung}, Eylee and {Park}, DaeKil},
        title = {Average R{\'e}nyi Entropy of a Subsystem in Random Pure State},
      journal = {arXiv e-prints},
         year = 2023,
        month = jan,
          eid = {arXiv:2301.09074},
        pages = {arXiv:2301.09074},
          doi = {10.48550/arXiv.2301.09074},
archivePrefix = {arXiv},
       eprint = {2301.09074},
 primaryClass = {quant-ph},
       adsurl = {https://ui.adsabs.harvard.edu/abs/2023arXiv230109074K}
}

@article{Liu_2022,
   title={Many-Body Quantum Magic},
   volume={3},
   ISSN={2691-3399},
   url={http://dx.doi.org/10.1103/PRXQuantum.3.020333},
   DOI={10.1103/prxquantum.3.020333},
   number={2},
   journal={PRX Quantum},
   publisher={American Physical Society (APS)},
   author={Liu, Zi-Wen and Winter, Andreas},
   year={2022},
   month=may }

@article{Aharony_2004,
   title={The deconfinement and Hagedorn phase transitions in weakly coupled large N gauge theories},
   volume={5},
   ISSN={1631-0705},
   url={http://dx.doi.org/10.1016/j.crhy.2004.09.012},
   DOI={10.1016/j.crhy.2004.09.012},
   number={9–10},
   journal={Comptes Rendus Physique},
   publisher={Cellule MathDoc/CEDRAM},
   author={Aharony, Ofer and Marsano, Joseph and Minwalla, Shiraz and Papadodimas, Kyriakos and Van Raamsdonk, Mark},
   year={2004},
   month=nov, pages={945–954} }

@article{ABSC,
   title={Approximate Bacon-Shor code and holography},
   volume={2021},
   ISSN={1029-8479},
   url={http://dx.doi.org/10.1007/JHEP05(2021)127},
   DOI={10.1007/jhep05(2021)127},
   number={5},
   journal={Journal of High Energy Physics},
   publisher={Springer Science and Business Media LLC},
   author={Cao, ChunJun and Lackey, Brad},
   year={2021},
   month=may }

@article{Steinberg_2023,
   title={Holographic codes from hyperinvariant tensor networks},
   volume={14},
   ISSN={2041-1723},
   url={http://dx.doi.org/10.1038/s41467-023-42743-z},
   DOI={10.1038/s41467-023-42743-z},
   number={1},
   journal={Nature Communications},
   publisher={Springer Science and Business Media LLC},
   author={Steinberg, Matthew and Feld, Sebastian and Jahn, Alexander},
   year={2023},
   month=nov }

@article{Harris_2018,
   title={Calderbank-Shor-Steane holographic quantum error-correcting codes},
   volume={98},
   ISSN={2469-9934},
   url={http://dx.doi.org/10.1103/PhysRevA.98.052301},
   DOI={10.1103/physreva.98.052301},
   number={5},
   journal={Physical Review A},
   publisher={American Physical Society (APS)},
   author={Harris, Robert J. and McMahon, Nathan A. and Brennen, Gavin K. and Stace, Thomas M.},
   year={2018},
   month=nov }

@article{Faulkner_2014,
   title={Gravitation from entanglement in holographic CFTs},
   volume={2014},
   ISSN={1029-8479},
   url={http://dx.doi.org/10.1007/JHEP03(2014)051},
   DOI={10.1007/jhep03(2014)051},
   number={3},
   journal={Journal of High Energy Physics},
   publisher={Springer Science and Business Media LLC},
   author={Faulkner, Thomas and Guica, Monica and Hartman, Thomas and Myers, Robert C. and Van Raamsdonk, Mark},
   year={2014},
   month=mar }

@misc{swingle2014universality,
      title={Universality of Gravity from Entanglement}, 
      author={Brian Swingle and Mark Van Raamsdonk},
      year={2014},
      eprint={1405.2933},
      archivePrefix={arXiv},
      primaryClass={hep-th}
}

@article{Blanco_2013,
   title={Relative entropy and holography},
   volume={2013},
   ISSN={1029-8479},
   url={http://dx.doi.org/10.1007/JHEP08(2013)060},
   DOI={10.1007/jhep08(2013)060},
   number={8},
   journal={Journal of High Energy Physics},
   publisher={Springer Science and Business Media LLC},
   author={Blanco, David D. and Casini, Horacio and Hung, Ling-Yan and Myers, Robert C.},
   year={2013},
   month=aug }

@ARTICLE{RK_wavefcn,
       author = {{Sonya Tarabunga}, Poetri and {Castelnovo}, Claudio},
        title = {Magic in generalized Rokhsar-Kivelson wavefunctions},
      journal = {arXiv e-prints},
         year = 2023,
        month = nov,
          eid = {arXiv:2311.08463},
        pages = {arXiv:2311.08463},
          doi = {10.48550/arXiv.2311.08463},
archivePrefix = {arXiv},
       eprint = {2311.08463},
 primaryClass = {quant-ph},
       adsurl = {https://ui.adsabs.harvard.edu/abs/2023arXiv231108463S}
}

@article{Bao_2022,
   title={Magic state distillation from entangled states},
   volume={105},
   ISSN={2469-9934},
   url={http://dx.doi.org/10.1103/PhysRevA.105.022602},
   DOI={10.1103/physreva.105.022602},
   number={2},
   journal={Physical Review A},
   publisher={American Physical Society (APS)},
   author={Bao, Ning and Cao, ChunJun and Su, Vincent Paul},
   year={2022},
   month=feb }

@article{sewell_mana_2022,
  title = {Mana and thermalization: Probing the feasibility of near-Clifford Hamiltonian simulation},
  author = {Sewell, Troy J. and White, Christopher David},
  journal = {Phys. Rev. B},
  volume = {106},
  issue = {12},
  pages = {125130},
  numpages = {14},
  year = {2022},
  month = {Sep},
  publisher = {American Physical Society},
  doi = {10.1103/PhysRevB.106.125130},
  url = {https://link.aps.org/doi/10.1103/PhysRevB.106.125130}
}

@article{sewell,
  title = {Mana and thermalization: Probing the feasibility of near-Clifford Hamiltonian simulation},
  author = {Sewell, Troy J. and White, Christopher David},
  journal = {Phys. Rev. B},
  volume = {106},
  issue = {12},
  pages = {125130},
  numpages = {14},
  year = {2022},
  month = {Sep},
  publisher = {American Physical Society},
  doi = {10.1103/PhysRevB.106.125130},
  url = {https://link.aps.org/doi/10.1103/PhysRevB.106.125130}
}

@misc{white2020mana,
      title={Mana in Haar-random states}, 
      author={Christopher David White and Justin H. Wilson},
      year={2020},
      eprint={2011.13937},
      archivePrefix={arXiv},
      primaryClass={quant-ph}
}

@article{Hayden_2016,
   title={Holographic duality from random tensor networks},
   volume={2016},
   ISSN={1029-8479},
   url={http://dx.doi.org/10.1007/JHEP11(2016)009},
   DOI={10.1007/jhep11(2016)009},
   number={11},
   journal={Journal of High Energy Physics},
   publisher={Springer Science and Business Media LLC},
   author={Hayden, Patrick and Nezami, Sepehr and Qi, Xiao-Liang and Thomas, Nathaniel and Walter, Michael and Yang, Zhao},
   year={2016},
   month=nov }

@ARTICLE{flatness,
       author = {{Tirrito}, Emanuele and {Sonya Tarabunga}, Poetri and {Lami}, Gugliemo and {Chanda}, Titas and {Leone}, Lorenzo and {Oliviero}, Salvatore F.~E. and {Dalmonte}, Marcello and {Collura}, Mario and {Hamma}, Alioscia},
        title = {Quantifying non-stabilizerness through entanglement spectrum flatness},
      journal = {arXiv e-prints},
         year = 2023,
        month = apr,
          eid = {arXiv:2304.01175},
        pages = {arXiv:2304.01175},
          doi = {10.48550/arXiv.2304.01175},
archivePrefix = {arXiv},
       eprint = {2304.01175},
 primaryClass = {quant-ph},
       adsurl = {https://ui.adsabs.harvard.edu/abs/2023arXiv230401175T}
}

@article{nogo,
    author = "Cao, ChunJun",
    title = {Non-trivial area operators require non-local magic},
    eprint = "2306.14996",
    archivePrefix = "arXiv",
    primaryClass = "hep-th",
    doi = "10.1007/JHEP11(2024)105",
    journal = "JHEP",
    volume = "11",
    pages = "105",
    year = "2024"
}

@ARTICLE{stabrenyi,
       author = {{Leone}, Lorenzo and {Oliviero}, Salvatore F.~E. and {Hamma}, Alioscia},
        title = {Stabilizer R{\'e}nyi Entropy},
      journal = {Phys. Rev. Lett.},
         year = 2022,
        month = feb,
       volume = {128},
       number = {5},
          eid = {050402},
        pages = {050402},
          doi = {10.1103/PhysRevLett.128.050402},
archivePrefix = {arXiv},
       eprint = {2106.12587},
 primaryClass = {quant-ph},
       adsurl = {https://ui.adsabs.harvard.edu/abs/2022PhRvL.128e0402L}
}

@ARTICLE{Dong1,
       author = {{Dong}, Xi},
        title = {The gravity dual of R{\'e}nyi entropy},
      journal = {Nature Communications},
         year = 2016,
        month = aug,
       volume = {7},
          eid = {12472},
        pages = {12472},
          doi = {10.1038/ncomms12472},
archivePrefix = {arXiv},
       eprint = {1601.06788},
 primaryClass = {hep-th},
       adsurl = {https://ui.adsabs.harvard.edu/abs/2016NatCo...712472D}
}

@article{chen_sharp_2016,
  title = {Sharp Continuity Bounds for Entropy and Conditional Entropy},
  author = {Chen, ZhiHua and Ma, ZhiHao and Nikoufar, Ismail and Fei, Shao-Ming},
  year = {2016},
  month = nov,
  journal = {Science China Physics, Mechanics \& Astronomy},
  volume = {60},
  number = {2},
  pages = {020321},
  issn = {1869-1927},
  doi = {10.1007/s11433-016-0367-x},
  urldate = {2023-09-21},
  langid = {english}
}

@article{chaosbymagic,
  title = {Probing chaos by magic monotones},
  author = {Goto, Kanato and Nosaka, Tomoki and Nozaki, Masahiro},
  journal = {Phys. Rev. D},
  volume = {106},
  issue = {12},
  pages = {126009},
  numpages = {26},
  year = {2022},
  month = {Dec},
  publisher = {American Physical Society},
  doi = {10.1103/PhysRevD.106.126009},
  url = {https://link.aps.org/doi/10.1103/PhysRevD.106.126009}
}

@article{veitch_resource_2014,
  title = {The Resource Theory of Stabilizer Quantum Computation},
  author = {Veitch, Victor and Mousavian, S. A. Hamed and Gottesman, Daniel and Emerson, Joseph},
  year = {2014},
  month = jan,
  journal = {New Journal of Physics},
  volume = {16},
  number = {1},
  pages = {013009--013009},
  doi = {10.1088/1367-2630/16/1/013009}
}

@article{Nezami:2016zni,
    author = "Nezami, Sepehr and Walter, Michael",
    title = {Multipartite Entanglement in Stabilizer Tensor Networks},
    eprint = "1608.02595",
    archivePrefix = "arXiv",
    primaryClass = "quant-ph",
    doi = "10.1103/PhysRevLett.125.241602",
    journal = "Phys. Rev. Lett.",
    volume = "125",
    pages = "241602",
    year = "2020"
}

@article{Dong:2018lsk,
    author = "Dong, Xi",
    title = {Holographic R\'enyi Entropy at High Energy Density},
    eprint = "1811.04081",
    archivePrefix = "arXiv",
    primaryClass = "hep-th",
    doi = "10.1103/PhysRevLett.122.041602",
    journal = "Phys. Rev. Lett.",
    volume = "122",
    number = "4",
    pages = "041602",
    year = "2019"
}

@ARTICLE{maxminentropy,
       author = {{Czech}, Bartlomiej and {Hayden}, Patrick and {Lashkari}, Nima and {Swingle}, Brian},
        title = {The information theoretic interpretation of the length of a curve},
      journal = {Journal of High Energy Physics},
         year = 2015,
        month = jun,
       volume = {2015},
          eid = {157},
        pages = {157},
          doi = {10.1007/JHEP06(2015)157},
archivePrefix = {arXiv},
       eprint = {1410.1540},
 primaryClass = {hep-th},
       adsurl = {https://ui.adsabs.harvard.edu/abs/2015JHEP...06..157C}
}

@misc{bu2023discrete,
      title={Discrete Quantum Gaussians and Central Limit Theorem}, 
      author={Kaifeng Bu and Weichen Gu and Arthur Jaffe},
      year={2023},
      eprint={2302.08423},
      archivePrefix={arXiv},
      primaryClass={quant-ph}
}

@article{Hebenstreit_2019,
   title={All Pure Fermionic Non-Gaussian States Are Magic States for Matchgate Computations},
   volume={123},
   ISSN={1079-7114},
   url={http://dx.doi.org/10.1103/PhysRevLett.123.080503},
   DOI={10.1103/physrevlett.123.080503},
   number={8},
   journal={Physical Review Letters},
   publisher={American Physical Society (APS)},
   author={Hebenstreit, M. and Jozsa, R. and Kraus, B. and Strelchuk, S. and Yoganathan, M.},
   year={2019},
   month=aug }

@misc{gottesman1997stabilizer,
      title={Stabilizer Codes and Quantum Error Correction}, 
      author={Daniel Gottesman},
      year={1997},
      eprint={quant-ph/9705052},
      archivePrefix={arXiv},
      primaryClass={quant-ph}
}

@misc{zhang2024unconditional,
      title={Unconditional quantum MAGIC advantage in shallow circuit computation}, 
      author={Xingjian Zhang and Zhaokai Pan and Guoding Liu},
      year={2024},
      eprint={2402.12246},
      archivePrefix={arXiv},
      primaryClass={quant-ph}
}

@misc{vairogs2024extracting,
      title={Extracting randomness from quantum 'magic'}, 
      author={Christopher Vairogs and Bin Yan},
      year={2024},
      eprint={2402.10181},
      archivePrefix={arXiv},
      primaryClass={quant-ph}
}

@ARTICLE{CFTspec,
       author = {{Calabrese}, Pasquale and {Lefevre}, Alexandre},
        title = {Entanglement spectrum in one-dimensional systems},
      journal = {Phys.Rev.A},
         year = 2008,
        month = sep,
       volume = {78},
       number = {3},
          eid = {032329},
        pages = {032329},
          doi = {10.1103/PhysRevA.78.032329},
archivePrefix = {arXiv},
       eprint = {0806.3059},
 primaryClass = {cond-mat.str-el},
       adsurl = {https://ui.adsabs.harvard.edu/abs/2008PhRvA..78c2329C}
}

@article{Maldacena_1999, volume={38},
    title={ The large-N limit of superconformal field
theories and supergravity},
   ISSN={0020-7748},
   url={http://dx.doi.org/10.1023/A:1026654312961},
   DOI={10.1023/a:1026654312961},
   number={4},
   journal={International Journal of Theoretical Physics},
   publisher={Springer Science and Business Media LLC},
   author={Maldacena, Juan},
   year={1999},
   pages={1113–1133} }

@article{Witten:1998qj,
    author = "Witten, Edward",
    title = {Anti-de Sitter space and holography},
    eprint = "hep-th/9802150",
    archivePrefix = "arXiv",
    reportNumber = "IASSNS-HEP-98-15",
    doi = "10.4310/ATMP.1998.v2.n2.a2",
    journal = "Adv. Theor. Math. Phys.",
    volume = "2",
    pages = "253--291",
    year = "1998"
}

@article{Faulkner_2013,
   title={Quantum corrections to holographic entanglement entropy},
   volume={2013},
   ISSN={1029-8479},
   url={http://dx.doi.org/10.1007/JHEP11(2013)074},
   DOI={10.1007/jhep11(2013)074},
   number={11},
   journal={Journal of High Energy Physics},
   publisher={Springer Science and Business Media LLC},
   author={Faulkner, Thomas and Lewkowycz, Aitor and Maldacena, Juan},
   year={2013},
   month=nov }

@article{Hubeny_2007,
   title={A covariant holographic entanglement entropy proposal},
   volume={2007},
   ISSN={1029-8479},
   url={http://dx.doi.org/10.1088/1126-6708/2007/07/062},
   DOI={10.1088/1126-6708/2007/07/062},
   number={07},
   journal={Journal of High Energy Physics},
   publisher={Springer Science and Business Media LLC},
   author={Hubeny, Veronika E and Rangamani, Mukund and Takayanagi, Tadashi},
   year={2007},
   month=jul, pages={062–062} }

@article{Yang_2016,
   title={Bidirectional holographic codes and sub-AdS locality},
   volume={2016},
   ISSN={1029-8479},
   url={http://dx.doi.org/10.1007/JHEP01(2016)175},
   DOI={10.1007/jhep01(2016)175},
   number={1},
   journal={Journal of High Energy Physics},
   publisher={Springer Science and Business Media LLC},
   author={Yang, Zhao and Hayden, Patrick and Qi, Xiao-Liang},
   year={2016},
   month=jan }

@article{Lewkowycz_2013,
   title={Generalized gravitational entropy},
   volume={2013},
   ISSN={1029-8479},
   url={http://dx.doi.org/10.1007/JHEP08(2013)090},
   DOI={10.1007/jhep08(2013)090},
   number={8},
   journal={Journal of High Energy Physics},
   publisher={Springer Science and Business Media LLC},
   author={Lewkowycz, Aitor and Maldacena, Juan},
   year={2013},
   month=aug }

@article{Hartman:2013qma,
    author = "Hartman, Thomas and Maldacena, Juan",
    title = {Time Evolution of Entanglement Entropy from Black Hole Interiors},
    eprint = "1303.1080",
    archivePrefix = "arXiv",
    primaryClass = "hep-th",
    doi = "10.1007/JHEP05(2013)014",
    journal = "JHEP",
    volume = "05",
    pages = "014",
    year = "2013"
}

@article{Balasubramanian:2011ur,
    author = "Balasubramanian, V. and Bernamonti, A. and de Boer, J. and Copland, N. and Craps, B. and Keski-Vakkuri, E. and Muller, B. and Schafer, A. and Shigemori, M. and Staessens, W.",
    title = {Holographic Thermalization},
    eprint = "1103.2683",
    archivePrefix = "arXiv",
    primaryClass = "hep-th",
    reportNumber = "HIP-2011-07-TH, UUITP-06-11",
    doi = "10.1103/PhysRevD.84.026010",
    journal = "Phys. Rev. D",
    volume = "84",
    pages = "026010",
    year = "2011"
}

@article{Nozaki:2013wia,
    author = "Nozaki, Masahiro and Numasawa, Tokiro and Takayanagi, Tadashi",
    title = {Holographic Local Quenches and Entanglement Density},
    eprint = "1302.5703",
    archivePrefix = "arXiv",
    primaryClass = "hep-th",
    reportNumber = "YITP-13-14, IPMU-13-0045",
    doi = "10.1007/JHEP05(2013)080",
    journal = "JHEP",
    volume = "05",
    pages = "080",
    year = "2013"
}

@article{Bao:2019,
    author = "Bao, Ning and Penington, Geoffrey and Sorce, Jonathan and Wall, Aron",
    title = "Beyond Toy Models: Distilling Tensor Networks in Full AdS/CFT",
    eprint = "1812.01171",
    archivePrefix = "arXiv",
    primaryClass = "quant-ph",
    doi = "10.1007/JHEP11(2019)069",
    journal = "JHEP",
    volume = "2019",
    number = "11",
    pages = "69",
    year = "2019"
}

@article{Bueno:2022jbl,
    author = "Bueno, Pablo and Cano, Pablo A. and Murcia, \'Angel and Rivadulla S\'anchez, Alberto",
    title = {Universal Feature of Charged Entanglement Entropy},
    eprint = "2203.04325",
    archivePrefix = "arXiv",
    primaryClass = "hep-th",
    reportNumber = "IFT-UAM/CSIC-22-18, CERN-TH-2022-033",
    doi = "10.1103/PhysRevLett.129.021601",
    journal = "Phys. Rev. Lett.",
    volume = "129",
    number = "2",
    pages = "021601",
    year = "2022"
}

@article{Nakaguchi:2016zqi,
    author = "Nakaguchi, Yuki and Nishioka, Tatsuma",
    title = {A holographic proof of R\'enyi entropic inequalities},
    eprint = "1606.08443",
    archivePrefix = "arXiv",
    primaryClass = "hep-th",
    reportNumber = "IPMU-16-0090, UT-16-26",
    doi = "10.1007/JHEP12(2016)129",
    journal = "JHEP",
    volume = "12",
    pages = "129",
    year = "2016"
}

@article{Zurek:2022xzl,
    author = "Zurek, Kathryn M.",
    title = {Snowmass 2021 White Paper: Observational Signatures of Quantum Gravity},
    eprint = "2205.01799",
    archivePrefix = "arXiv",
    primaryClass = "gr-qc",
    month = "5",
    year = "2022"
}

@article{PhysRevD.99.066012,
  title = {Aspects of capacity of entanglement},
  author = {de Boer, Jan and J\"arvel\"a, Jarkko and Keski-Vakkuri, Esko},
  journal = {Phys. Rev. D},
  volume = {99},
  issue = {6},
  pages = {066012},
  numpages = {34},
  year = {2019},
  month = {Mar},
  publisher = {American Physical Society},
  doi = {10.1103/PhysRevD.99.066012},
  url = {https://link.aps.org/doi/10.1103/PhysRevD.99.066012}
}

@article{PhysRevLett.105.080501,
  title = {Entanglement Entropy and Entanglement Spectrum of the Kitaev Model},
  author = {Yao, Hong and Qi, Xiao-Liang},
  journal = {Phys. Rev. Lett.},
  volume = {105},
  issue = {8},
  pages = {080501},
  numpages = {4},
  year = {2010},
  month = {Aug},
  publisher = {American Physical Society},
  doi = {10.1103/PhysRevLett.105.080501},
  url = {https://link.aps.org/doi/10.1103/PhysRevLett.105.080501}
}

@article{lieb_ruskai,
  title = {A Fundamental Property of Quantum-Mechanical Entropy},
  author = {Lieb, Elliott H. and Ruskai, Mary Beth},
  journal = {Phys. Rev. Lett.},
  volume = {30},
  issue = {10},
  pages = {434--436},
  numpages = {0},
  year = {1973},
  month = {Mar},
  publisher = {American Physical Society},
  doi = {10.1103/PhysRevLett.30.434},
  url = {https://link.aps.org/doi/10.1103/PhysRevLett.30.434}
}

@article{PhysRevB.83.115322,
  title = {Entanglement spectrum and entanglement thermodynamics of quantum Hall bilayers at $\ensuremath{\nu}=1$},
  author = {Schliemann, John},
  journal = {Phys. Rev. B},
  volume = {83},
  issue = {11},
  pages = {115322},
  numpages = {5},
  year = {2011},
  month = {Mar},
  publisher = {American Physical Society},
  doi = {10.1103/PhysRevB.83.115322},
  url = {https://link.aps.org/doi/10.1103/PhysRevB.83.115322}
}

@misc{cheng2025dataset,
  author       = {},
  title        = {Code and data for this article},
  year         = 2025,
  url = {https://github.com/gongc123/holographic-magic},
  note={\url{https://github.com/gongc123/holographic-magic}}
}
\end{document}